\documentclass[11pt]{article}
\usepackage{graphicx}
\usepackage{fullpage}
\usepackage{amsmath,amssymb,amsthm}
\usepackage{enumerate}
\usepackage{verbatim}
\usepackage{mathrsfs}
\usepackage{cite}
\numberwithin{equation}{section}

\usepackage[usenames]{color}
\definecolor{plum}  {rgb}{.4,0,.4}
\definecolor{BrickRed} {rgb}{0.6,0,0}

\usepackage[plainpages=false,pdfpagelabels,colorlinks=true,linkcolor=BrickRed,citecolor=plum]{hyperref}

\def\cA{{\mathcal A}}
\def\cB{{\mathcal B}}
\def\cC{{\mathcal C}}

\def\cE{{\mathcal E}}
\def\cF{{\mathcal F}}

\def\cH{{\mathcal H}}

\def\cM{{\mathcal M}}

\def\cP{{\mathcal P}}

\def\sA{{\mathsf A}}
\def\sB{{\mathsf B}}

\def\sE{{\mathsf E}}

\def\sM{{\mathsf M}}

\def\sP{{\mathsf P}}

\def\sU{{\mathsf U}}
\def\sV{{\mathsf V}}

\def\sX{{\mathsf X}}
\def\sY{{\mathsf Y}}
\def\sZ{{\mathsf Z}}

\def\E{\mathbb{E}}
\def\MM{\mathbb{M}}
\def\PP{\mathbb{P}}
\def\QQ{\mathbb{Q}}

\def\Reals{\mathbb{R}}
\def\Naturals{\mathbb{N}}

\def\argmin{\operatornamewithlimits{arg\,min}}
\def\Ent{\operatorname{Ent}}
\def\Var{\operatorname{Var}}
\def\Cov{\operatorname{Cov}}
\def\MMSE{\operatorname{MMSE}}
\def\LSI{\operatorname{LSI}}

\def\SDPI{\operatorname{SDPI}}
\def\PI{\operatorname{PI}}
\def\Bernoulli{\mathrm{Bern}}
\def\BSC{\mathrm{BSC}}

\def\se{\mathsf{e}} 
\def\LC{\operatorname{LC}}
\def\TV{{\rm TV}}

\def\Prob{{\mathscr P}}
\def\PProb{\Prob_*}
\def\Func{{\mathscr F}}
\def\PFunc{\Func^0_*}
\def\SPFunc{\Func_*}
\def\Chan{{\mathscr M}}

\def\1{\mathbf{1}}
\def\id{\operatorname{id}}
\def\d{{\text {\rm d}}}
\def\diag{\operatorname{diag}}

\def\deq{\triangleq}

\def\ave#1{\langle #1 \rangle}
\def\eps{\varepsilon}

\newtheorem{definition}{Definition}[section]
\newtheorem{theorem}{Theorem}[section]
\newtheorem{lemma}{Lemma}[section]
\newtheorem{proposition}{Proposition}[section]
\newtheorem{corollary}{Corollary}[section]

\newtheorem{remark}{Remark}[section]
\newtheorem{example}{Example}[section]

\begin{document}

\title{\bf Strong Data Processing Inequalities\\
and $\Phi$-Sobolev Inequalities for Discrete Channels}

\author{Maxim Raginsky%
\thanks{The author is with the Department of Electrical and Computer Engineering and the Coordinated Science Laboratory, University of Illinois, Urbana, IL 61801, USA. E-mail: maxim@illinois.edu.}
\thanks{This work was supported in part by the NSF under CAREER award no.\ CCF-1254041 and by the Center for Science of Information (CSoI), an NSF Science and Technology Center, under grant agreement CCF-0939370. The material in this paper was presented in part at the 2013 IEEE International Symposium on Information Theory.}}
\date{March 30, 2016}

\maketitle
\thispagestyle{empty}


\begin{abstract} The noisiness of a channel can be measured by comparing suitable 
functionals of the input and output distributions. For instance, the worst-case ratio of output relative
entropy to input relative entropy for all possible pairs of input distributions is bounded from above by unity, by the data processing theorem. However, for a fixed reference input
distribution, this quantity may be strictly smaller than one, giving so-called
strong data processing inequalities (SDPIs). The same considerations apply to an arbitrary $\Phi$-divergence. This paper presents a systematic study of optimal constants in SDPIs for discrete channels, including their variational characterizations, upper and lower bounds, structural results for channels on product probability spaces, and the relationship between SDPIs and so-called $\Phi$-Sobolev inequalities (another class of inequalities that can be used to quantify the noisiness of a channel by controlling entropy-like functionals of the input distribution by suitable measures of input-output correlation). Several applications to  information theory, discrete probability, and statistical physics are discussed.
\end{abstract}

\tableofcontents

\section{Introduction}

The well-known data processing inequality for the relative entropy states that, for any two probability distributions $\mu,\nu$ over an alphabet $\sX$ and for any stochastic transformation (channel) $K$ with input alphabet $\sX$ and output alphabet $\sY$,
\begin{align*}
	D(\nu K \| \mu K) \le D(\nu \| \mu),
\end{align*}
where $\mu K$ denotes the distribution at the output of $K$ when the input has distribution $\mu$ (and similarly for $\nu K$). However, if we fix the \textit{reference distribution} $\mu$ and vary only $\nu$, then in many cases it is possible to show that $D(\nu K \| \mu K)$ is {\em strictly} smaller than $D(\nu \| \mu)$ unless $\nu \equiv \mu$. To capture this effect, we define the quantity
\begin{align*}
	\eta(\mu,K) \deq \sup_{\nu \neq \mu} \frac{D(\nu K \| \mu K)}{D(\nu \| \mu)},
\end{align*}
and we say that the channel $K$ satisfies a \textit{strong data processing inequality} (SDPI) at input distribution $\mu$ if $\eta(\mu,K) < 1$. In a remarkable paper \cite{Ahlswede_Gacs_hypercont}, Ahlswede and G\'acs have uncovered deep relationships between $\eta(\mu,K)$ and several other quantities, such as the maximal correlation (see \cite{Witsenhausen_correlation} and references therein) and so-called hypercontractivity constants of a certain Markov operator associated to the pair $(\mu,K)$.
For example, they have shown that if $\sX = \sY = \{0,1\}$, $\mu = \Bernoulli(1/2)$, and $K = \BSC(\eps)$, then $\eta(\mu,K) = (1-2\eps)^2$, which is also equal to the squared maximal correlation in the joint distribution $P_{XY}$ with $P_X = P_Y = \Bernoulli(1/2)$ and $P_{Y|X} = \BSC(\eps)$, the so-called doubly symmetric binary source (DSBS) with parameter $\eps$ \cite{Wyner_common_info}.

After the pioneering work of Ahlswede and G\'acs, the contraction properties of relative entropy (and other $\Phi$-divergences \cite{Csiszar_divergence,LieseVajda06}) under the action of stochastic transformations have been studied by several other authors \cite{Cohen_etal_dataproc,Choi_Ruskai_Seneta,Miclo_hypercontractive,Cohen_etal_book,DelMoral_contraction}. In particular, Cohen et al.~\cite{Cohen_etal_dataproc}, who were the first ones to take up this subject after \cite{Ahlswede_Gacs_hypercont}, showed that the SDPI constant of any channel $K$ with respect to any $\Phi$-divergence is always upper-bounded by the so-called \textit{Dobrushin contraction coefficient} of $K$ \cite{Dobrushin_CLT_MC_1,Dobrushin_CLT_MC_2}, another well-known numerical measure of the amount of noise introduced by a channel. (This result of Cohen et al.\ was rediscovered five years later in the machine learning community \cite{BoyenKoller}.) In the last couple of years, strong data processing inequalities became the subject of intense interest in the information theory community \cite{Kamath_Anantharam_hypercont,Anantharam_etal_HGR,Courtade_SDPI,Raginsky_SDPI_ISIT,YP_YW_dissipation,Anantharam_etal_hypercont_2,Liu_etal_key_capacity,PW_BayesSDPI,MakurZheng_contraction} due to their apparent usefulness for establishing various converse results.

In this paper, we revisit the problem of characterizing the strong data processing constant $\eta(\mu,K)$ [and its generalizations for arbitrary $\Phi$-divergence] and establish a number of new upper and lower bounds, as well as new structural results on SDPI constants in product probability spaces. We also address the relationship between strong data processing inequalities and so-called \textit{$\Phi$-Sobolev inequalities} \cite{Chafai_entropy}. These inequalities also quantify the noisiness of a Markov operator (probability transition kernel) by relating certain ``entropy-like" functionals of the input to the rate of increase of suitable ``energy-like'' quantities from the input to the output. (Logarithmic Sobolev inequalities, widely studied in the theory of probability and Markov chains \cite{Bakry_logsob_notes,Diaconis_Saloff_logsob,Miclo_hypercontractive,Bobkov_Tetali_logsob,Mossel_reverse_hypercont}, are a special case.) In particular, we show that the optimal constants in $\Phi$-Sobolev inequalities for a reversible Markov chain can be related to SDPI constants of certain factorizations of the transition kernel of the chain as a product of a forward channel and a backward channel. Such factorizations correspond to all possible realizations of the one-step transition of the chain as a two-component Gibbs sampler  \cite{Diaconis_stoch_alt_proj}, which is a standard technique in Markov chain Monte Carlo \cite{Gilks_etal_MCMC,Robert_Casella_MCMC}. Conversely, for a fixed input distribution $\mu$ on $\sX$, the SDPI constants of a given channel $K$ with input in $\sX$ and output in $\sY$ are related to $\Phi$-Sobolev constants of the reversible Markov chain on $\sX$ obtained by composing the forward channel $K$ with the backward channel $K^*$ determined via Bayes' rule. To keep things simple, we focus on the discrete case, when both $\sX$ and $\sY$ are finite, although some of our results generalize easily to the case of arbitrary Polish alphabets (see, e.g., \cite{PW_BayesSDPI}).

The remainder of the paper is organized as follows. After giving some necessary background  on $\Phi$-entropies and $\Phi$-divergences in Section~\ref{sec:entropies}, we proceed to the study of strong data processing inequalities in Section~\ref{sec:SDPI}. Next, in Section~\ref{sec:phisob}, we define the $\Phi$-Sobolev inequalities and characterize their relation with SDPIs. Several examples of applications are given in Section~\ref{sec:applications}. Section~\ref{sec:summary} provides a summary of key contributions. A number of auxiliary technical results are stated and proved in the Appendices.

\subsection{Notation}
\label{ssec:notation}

We will denote by $\Prob(\sX)$ the set of all probability distributions on an alphabet $\sX$ and by $\PProb(\sX)$ the subset of $\Prob(\sX)$ consisting of all strictly positive distributions. The set of all real-valued functions on $\sX$ is denoted by $\Func(\sX)$; $\SPFunc(\sX)$ and $\PFunc(\sX)$ are the subsets of $\Func(\sX)$ consisting of all strictly positive and nonnegative functions, respectively. Any channel\footnote{We will also use the terms ``stochastic transformation" or ``Markov kernel."} with input alphabet $\sX$, output alphabet $\sY$, and transition probabilities $\{K(y|x) : x \in \sX,y \in \sY\}$ acts on probability distributions $\mu \in \Prob(\sX)$ from the right by
\begin{align*}
	\mu K(y) &= \sum_{x \in \sX} \mu(x)K(y|x), \qquad y \in \sY
\end{align*}
or on functions $f \in \Func(\sY)$ from the left by
\begin{align*}
	Kf(x) &= \sum_{y \in \sY} K(y|x)f(y), \qquad x \in \sX.
\end{align*}
The set of all such channels will be denoted by $\Chan(\sY|\sX)$. The affine map $\mu \mapsto \mu K$  naturally extends to a linear map on the signed measures on $\sX$, since any such measure $\nu$ can be uniquely represented as $\alpha_1 \mu_1 - \alpha_2 \mu_2$ for some constants $\alpha_1,\alpha_2 \ge 0$ and some $\mu_1,\mu_2 \in \Prob(\sX)$; thus, we set $\nu K = \alpha_1 \mu_1 K - \alpha_2 \mu_2 K$. The linear map $f \mapsto Kf$ is positive [i.e., $K(\PFunc(\sY)) \subseteq \PFunc(\sX)$], and unital [i.e., $K1=1$, where $1$ denotes the constant function that takes the value $1$ everywhere on its domain]. If $\mu \otimes K \in \Prob(\sX \times \sY)$ denotes the distribution of a random pair $(X,Y) \in \sX \times \sY$ with $P_X = \mu$ and $P_{Y|X} = K$, then $K f(x) = \E[f(Y)|X=x]$ for any $f \in \Func(\sY)$ and $x \in \sX$.

We will say that a pair $(\mu,K) \in \Prob(\sX) \times \Chan(\sY|\sX)$ is \textit{admissible} if $\mu \in \PProb(\sX)$ and $\mu K \in \PProb(\sY)$. For any such pair, there exists a unique channel $K^* \in \Chan(\sX|\sY)$ with the property that
\begin{align}\label{eq:backward_1}
	\E[g(Y) Kf(Y)] = \E[K^* g(X)f(X)]
\end{align}
for all $g \in \Func(\sY), f \in \Func(\sX)$. This \textit{backward} or \textit{adjoint} channel can be specified explicitly via the transition probabilities
\begin{align}\label{eq:backward_2}
	K^*(x|y) = \frac{K(y|x)\mu(x)}{\mu K(y)}, \qquad  (x,y) \in \sX \times \sY
\end{align}
(this is simply an application of Bayes' rule). If $(X,Y) \sim \mu \otimes K$, then $K^* = P_{X|Y}$, so in particular $K^*f(y) = \E[f(X)|Y=y]$ for any $f \in \Func(\sX)$ and $y \in \sY$. Strictly speaking, $K^*$ depends on both $\mu$ and $K$, and we may occasionally indicate this fact by writing $K^*_\mu$ instead of $K^*$. 

Given a number $p \in [0,1]$, we will often write $\bar{p}$ for $1-p$. For $p,q \in [0,1]$, we let $p \star q \deq p \bar{q} + \bar{p} q$. Thus, if $X \sim \Bernoulli(p)$ and $Z \sim \Bernoulli(q)$ are independent random variables, then $Y = X \oplus Z$ has distribution $\Bernoulli(p \star q)$. For $a,b \in \Reals$, we let $a \vee b \deq \max\{a,b\}$ and $a \wedge b \deq \min\{a,b\}$. Other notation and definitions will be introduced in the sequel as needed.

\section{Background on $\Phi$-entropies and $\Phi$-divergences}
\label{sec:entropies}

Let $\cF$ denote the set of all convex functions $\Phi \colon \Reals^+ \to \Reals$. For any $\Phi \in \cF$,  the \textit{$\Phi$-entropy} of a nonnegative real-valued random variable $U$ is defined by
\begin{align}\label{eq:phi_entropy}
	\Ent_\Phi[U] \deq \E[\Phi(U)] - \Phi(\E U),
\end{align}
provided $\E[\Phi(U)]< \infty$ (see \cite{Chafai_entropy} and \cite[Chap.~14]{Boucheron_etal_concentration_book}).  For example, if $\Phi(u)=u^2$, then $\Ent_\Phi[U] = \Var[U]$; if $\Phi(u) = u \log u$, then
\begin{align*}
	\Ent_\Phi[U] = \E[U \log U] - \E[U] \log \E[U].
\end{align*}
The $\Phi$-entropy is nonnegative by Jensen's inequality.

The $\Phi$-divergences\footnote{We use the term ``$\Phi$-divergence'' instead of the more common ``$f$-divergence'' because we reserve $f$ for real-valued functions on $\sX$.} between probability distributions \cite{Csiszar_divergence,LieseVajda06} arise as a special case of the above definition. Fix some $\mu \in \PProb(\sX)$ (this restriction is sufficient for our purposes, and helps avoid certain technicalities involving division by zero). Then, for any $\Phi \in \cF$, the $\Phi$-divergence between an arbitrary probability distribution $\nu \in \Prob(\sX)$ and $\mu$ is defined as
\begin{align*}
	D_\Phi(\nu \| \mu) \deq \E_\mu\left[ \Phi\left(\frac{\d \nu}{\d \mu}\right) \right] - \Phi(1).
\end{align*}
Note that this differs from the usual definition by the subtraction of $\Phi(1)$. There are two reasons behind this modification: (a) $D_\Phi(\mu \| \mu) = 0$ for any $\mu$,\footnote{However, unless $u \mapsto \Phi(u)$ is strictly convex at $1$, $D_\Phi(\nu \| \mu) = 0$ does not necessarily imply that $\nu = \mu$.} and (b) any two $\Phi,\Phi'$ such that $\Phi - \Phi'$ is affine determine the same divergence. If we now consider a random variable $X \in \sX$ with distribution $\mu$ and let $f = \d \nu/ \d \mu$, then
\begin{align*}
	D_\Phi(\nu \| \mu) = \Ent_\Phi\left[f(X)\right].
\end{align*}
Moreover, if $\Phi(1) = 0$, we can write $D_\Phi(\nu \| \mu) = \E_\mu[\Phi \circ f]$ since  $\E[f(X)]=1$. Here are some important examples of $\Phi$-divergences \cite{LieseVajda06}:
\begin{enumerate}
	\item The relative entropy
	\begin{align*}
		D(\nu \| \mu) = \E_\nu\left[ \log \frac{\d \nu}{\d \mu}\right] = \E_\mu \left[ \frac{\d \nu}{\d \mu}\log \frac{\d \nu}{\d \mu}\right]
	\end{align*}
	is a $\Phi$-divergence with $\Phi(u) = u \log u$.
	\item The total variation distance 
	\begin{align*}
		\| \nu - \mu \|_\TV = \frac{1}{2}\E_\mu \left| \frac{\d \nu}{\d \mu} - 1 \right|
	\end{align*}
	is a $\Phi$-divergence with $\Phi(u) = \frac{1}{2}|u-1|$.
	\item The $\chi^2$-divergence
	\begin{align*}
		\chi^2(\nu\|\mu) = \E_\mu \left[ \left( \frac{\d \nu}{\d \mu} -1 \right)^2\right]
	\end{align*}
	is a $\Phi$-divergence with $\Phi(u) = (u-1)^2$ or $\Phi(u) = u^2-1$. This is a particular instance of the fact that any two $\Phi,\Phi' \in \cF$ that differ by an affine function determine the same divergence.
	\item The squared Hellinger distance
	\begin{align*}
		H^2(\nu,\mu) = \E_\mu \left[ \left(\sqrt{\frac{\d \nu}{\d \mu}}-1\right)^2\right] 
	\end{align*}
	is a $\Phi$-divergence with $\Phi(u) = (\sqrt{u}-1)^2$ or $\Phi(u) = 2-2\sqrt{u}$.
\end{enumerate}
An important class of $\Phi$-divergences arises in the context of Bayesian estimation. Given a parameter $\lambda \in (0,1)$, consider a random pair $(\Theta,X)$ with
\begin{align*}
	\Theta \sim \Bernoulli(\lambda) \quad \text{and} \quad P_{X|\Theta = \theta} = \begin{cases} \mu, & \text{if $\theta=0$} \\
	\nu, &\text{if $\theta = 1$}\end{cases}.
\end{align*}
Fix an action space $\sA$ and a loss function $\ell : \{0,1\} \times \sA \to \Reals$ --- in other words, if $\Theta = \theta$ and an action $a \in \sA$ is selected, then we incur the loss of $\ell(\theta,a)$. Consider the problem of selecting an action in $\sA$ based on some observation $Z$ related to $(\Theta,X)$ via the Markov chain $\Theta \to X \to Z$ --- i.e., $Z$ and $\Theta$ are conditionally independent given $X$. If $A = \gamma(Z)$ for some function $\gamma$, then we incur the average loss
\begin{align*}
	 \E[\ell(\Theta,\gamma(Z))] &= \bar{\lambda}\, \E[\ell(0,\gamma(Z))] + \lambda\, \E[\ell(1,\gamma(Z))].
\end{align*}
The goal is to pick $\gamma$ to minimize this expected loss for a given observation channel $P_{Z|X}$. In the extreme case when $Z$ is independent of $X$, the best we can do is to take
\begin{align*}
	a^* = \argmin_{a \in \sA} \left[\bar{\lambda}\ell(0,a) + \lambda \ell(1,a)\right],
\end{align*}
giving us the average loss of
\begin{align*}
	L^*_\lambda &\deq \inf_{a \in \sA} \left[\bar{\lambda}\ell(0,a) + \lambda \ell(1,a)\right].
\end{align*}	
On the other hand, if $Z=X$, then we can attain the {\em minimum Bayes risk}
\begin{align*}
	 L^*_\lambda(\nu,\mu) &\deq \inf_{\gamma} \E[\ell(\Theta,\gamma(X))] \\
	&= \inf_{\gamma} \left\{ \bar{\lambda} \int_\sX \ell(0,\gamma(x)) \nu(\d x) + \lambda \int_\sX \ell(1,\gamma(x)) \mu(\d x)\right\},
\end{align*}
where the infimum is over all measurable functions $\gamma : \sX \to \sA$. The following result is well-known (see, e.g., \cite[p.~882]{Nguyen_etal_surrogate_loss}), but the proof is so simple that we give it here:
\begin{proposition} The quantity
	\begin{align*}
		D_{\ell,\lambda}(\nu \| \mu) \deq L^*_\lambda - L^*_\lambda(\nu,\mu)
	\end{align*}
	is a $\Phi$-divergence.
\end{proposition}
\begin{proof} Define the function
	\begin{align*}
		\Phi_{\ell,\lambda}(u) \deq \sup_{a \in \sA} \left[L^*_\lambda - \bar{\lambda}\ell(0,a) - \lambda \ell(1,a)u \right], \qquad u \ge 0.
	\end{align*}
	Being a pointwise supremum of affine functions of $u$, it is convex. Moreover, $\Phi_{\ell,\lambda}(1)=0$. With this, we can write
	\begin{align*}
		L^*_\lambda - L^*_\lambda(\nu,\mu) &= \sup_\gamma \left(L^*_\lambda - \int_\sX \mu(\d x) \left[\bar{\lambda}\ell(0,\gamma(x)) + \lambda \frac{\d \nu}{\d \mu}(x) \ell(1,\gamma(x))\right] \right)\\
		&=\int_\sX \mu(\d x)\sup_{a \in \sA}\left[L^*_\lambda - \bar{\lambda}\ell(0,a) - \lambda \ell(1,a) \frac{\d \nu}{\d \mu}(x)\right] \\
		&= \E_\mu\left[ \Phi_{\ell,\lambda}\left(\frac{\d \nu}{\d \mu}\right)\right].
	\end{align*}
\end{proof}
We consider two particular cases:
\begin{itemize}
	\item $\sA = \{0,1\}$, $\ell(\theta,a) = \1_{\{\theta \neq a\}}$. An easy calculation shows that $L^*_\lambda = \lambda \wedge \bar{\lambda}$ and
	\begin{align*}
		\Phi_{\ell,\lambda}(u) &= \left[\lambda \wedge \bar{\lambda} - \bar\lambda u\right] \vee \left[\lambda \wedge \bar{\lambda} - \bar\lambda \right] \\
		&= \lambda \wedge \bar{\lambda} - (\lambda u) \wedge \bar{\lambda}.
	\end{align*}
Alternatively, we can write
	\begin{align*}
		L^*_\lambda = \frac{1}{2} - \frac{1}{2} \| \Bernoulli(\lambda) - \Bernoulli(\bar{\lambda}) \|_\TV = \frac{1}{2} - \frac{1}{2}|1-2\lambda|	
	\end{align*}
	and
	\begin{align*}
		L^*_\lambda(\nu,\mu) = \frac{1}{2} - \frac{1}{2} \| \lambda \nu - \bar{\lambda} \mu \|_\TV,
	\end{align*}
	where the total variation norm $\| \nu \|_\TV$ of a signed measure $\nu$ on $\sX$ is given by
	\begin{align*}
		\| \nu \|_\TV = \frac{1}{2}\sum_{x \in\sX}|\nu(x)|.
	\end{align*}
The optimal decision function is
	\begin{align*}
		\gamma^*(x) = \1_{\left\{ \lambda \frac{\d \nu}{\d \mu}(x) \le \bar{\lambda}\right\}}.
	\end{align*}
	The resulting divergence is known as the {\em Bayes} or {\em statistical information} \cite{DeGroot62}
	\begin{align*}
		B_\lambda(\nu \| \mu) = \frac{1}{2} \| \lambda \nu - \bar{\lambda} \mu \|_\TV - \frac{1}{2}|1-2\lambda|.
	\end{align*}
	In fact, any $\Phi$-divergence can be expressed as an integral of statistical informations \cite[Thm.~11]{LieseVajda06}: for any $\Phi \in \cF$, there exists a unique Borel measure $\sM_\Phi$ on $[0,1]$, such that
	\begin{align}\label{eq:information_integral}
		D_\Phi(\nu \| \mu) = \int_{[0,1]} B_\lambda(\nu\|\mu) \sM_\Phi(\d\lambda).
	\end{align}
	
	\item $\sA = \Reals$, $\ell(\theta,a) = (a-\theta)^2$. Then $L^*_\lambda = \lambda \bar{\lambda}$ and
	\begin{align*}
		\Phi_{\ell,\lambda}(u) = \lambda\bar{\lambda}\left(1 - \frac{u}{\lambda u + \bar{\lambda}}\right),
	\end{align*}
	which gives
	\begin{align*}
		L^*_\lambda(\nu,\mu) = \lambda \bar{\lambda}\, \E_\mu\left[ \frac{\d \nu/\d \mu}{\lambda \d \nu/\d \mu + \bar{\lambda}}\right],
	\end{align*}
with the optimum decision function
	\begin{align*}
		\gamma^*(x) &= \frac{\lambda \frac{\d \nu}{\d \mu}(x)}{\lambda \frac{\d \nu}{\d \mu}(x) + \bar{\lambda}}.
	\end{align*}
The corresponding divergence is then given by
\begin{align*}
	D_{\ell,\lambda}(\nu\|\mu) &= \lambda\bar{\lambda}\left(1 - \E_\mu \left[ \frac{\d \nu/\d \mu}{\lambda\, \d \nu/ \d \mu + \bar{\lambda}}\right]\right) \\
	&= (\lambda\bar{\lambda})^2 \E_\mu \left[ \frac{(\d \nu/\d \mu-1)^2}{\lambda\, \d \nu/\d \mu + \bar{\lambda}}\right],
\end{align*}
where the second expression follows after some algebraic manipulations. Note that the functions $u \mapsto \frac{(u-1)^2}{\lambda u + \bar{\lambda}}$ for $0 < \lambda < 1$ also belong to $\cF$. The divergences generated by these functions (modulo multiplicative constants) have appeared throughout the statistical literature \cite{LeCamBook,GyorfiVajda}. In particular, Le Cam \cite{LeCamBook} considers the case $\lambda = 1/2$ with the above Bayesian hypothesis testing interpretation, while Gy\"orfi and Vajda \cite{GyorfiVajda} look at arbitrary $\lambda$ (including the endpoints $0$ and $1$). For our purposes, it will be convenient to work with the function $u \mapsto \lambda \bar{\lambda}\frac{(u-1)^2}{\lambda u +\bar{\lambda}}$, which gives the \textit{Le Cam divergence} with parameter $\lambda \in (0,1)$:
\begin{align}\label{eq:LC}
	\LC_\lambda(\nu \| \mu) &\deq \lambda\bar{\lambda}\, \E_\mu\left[ \frac{(\d \nu/\d \mu-1)^2}{\lambda \d \nu/ \d \mu + \bar{\lambda}}\right] \equiv \frac{1}{\lambda\bar{\lambda}}D_{\ell,\lambda}(\nu \| \mu).
\end{align}
The Le Cam divergences $\LC_0(\cdot \| \cdot)$ and $\LC_1(\cdot \| \cdot)$ are also well-defined and are identically zero.
\end{itemize}
More examples of $\Phi$-divergences, as well as a wide variety of inequalities between them, can be found in \cite{Sason_Verdu_fdiv}.

From now on, when dealing with quantities indexed by $\Phi$, we will often substitute $\Phi$ with some mnemonic notation related to the corresponding $\Phi$-divergence, e.g., $\TV$, $\chi^2$, etc. Moreover, for the case of the relative entropy we will often omit the index $\Phi$ altogether and write $\Ent(\cdot)$, $D(\cdot\|\cdot)$, etc.

\subsection{Subadditivity of $\Phi$-entropies}

Let $U$ and $Y$ be jointly distributed random variables, where $U$ takes nonnegative real vaues and $Y$ is arbitrary. Given a  function $\Phi\in\cF$, define the \textit{conditional $\Phi$-entropy} of $U$ given $Y$:
\begin{align}\label{eq:conditional_phi_entropy}
	\Ent_\Phi[U|Y] &\deq  \E[\Phi(U)|Y] - \Phi(\E[U|Y]).
\end{align}
This is a random variable, since it depends on $Y$. Combining \eqref{eq:conditional_phi_entropy} with \eqref{eq:phi_entropy} gives the following generalization of the law of total variance:
\begin{align}\label{eq:total_entropy}
	\Ent_\Phi[U] = \E\left[\Ent_\Phi[U|Y]\right] + \Ent_\Phi[\E[U|Y]]
\end{align}
(see \cite[pp.~351--352]{Chafai_entropy}).

\begin{remark} {\em We may think of
	\begin{align*}
		J_\Phi(U|Y) \deq \E\left[\Ent_\Phi[U|Y]\right]
	\end{align*}
as a kind of ``Fisher $\Phi$-information'' about $U$ contained in $Y$.\footnote{We are grateful to P.~Tetali for suggesting this interpretation.} Indeed, let us consider the following special case: let $(Y,Y')$ be an exchangeable pair on some space $\sY$ (i.e., $P_{Y,Y'}(y,y') = P_{Y,Y'}(y',y)$ for all $y,y'$), and let $U = f(Y)$ for some $f : \sY \to \Reals^+$. Let $K$ be the stochastic transformation $P_{Y'|Y}$. Then $\E[U|Y'] = \E[f(Y)|Y'] = K^*f(Y')$ has the same distribution as $\E[f(Y')|Y] = K^*f(Y)$, and
	\begin{align*}
		J_\Phi(U|Y) &= \Ent_\Phi[U] - \Ent_\Phi[\E[U|Y]] \\
		&= \Ent_\Phi[f(Y)] - \Ent_\Phi[K^*f(Y)].
	\end{align*}
By convexity of $\Phi$,
\begin{align*}
	\Phi(u+v) \ge \Phi(u) + v\Phi'(u).
\end{align*}
If we write $K^* = \id + L$, where $\id$ is the identity operator on $\Func(\sY)$, then
\begin{align*}
	J_\Phi(U|Y) &= \Ent_\Phi[f(Y)] - \Ent_\Phi[f(Y) + Lf(Y)] \\
	&\le -\E[\Phi'(f(Y))Lf(Y)].
\end{align*}
Moreover, if we have a continuous-time reversible Markov chain on $\sY$ with stationary distribution $P_Y$ and with infinitesimal generator $L$, then $(Y_0,Y_t)$ is an exchangeable pair for each $t$, and
\begin{align*}
	J_\Phi(f(Y_0)|Y_t) &= \Ent_\Phi[f(Y_0)] - \Ent_\Phi[K^*_tf(Y_0)] \\
	&= -t\, \E[\Phi'(f(Y_0))L f(Y_0)] + o(t)
\end{align*}
Dividing both sides by $t$ and taking the limit as $t \to 0$, we get
\begin{align*}
	\frac{\d}{\d t} J_\Phi(f(Y_0)|Y_t)\Big|_{t=0} = \lim_{t \to 0} \frac{J_\Phi(f(Y_0)|Y_t)}{t} = - \E[\Phi'(f(Y_0)) Lf(Y_0)],
\end{align*}
which coincides with the $\Phi$-Fisher information functional of Chafa\"i \cite[Eq.~(1.14)]{Chafai_entropy}.\hfill$\diamond$}
\end{remark}

We say that the $\Phi$-entropy is \textit{subadditive} if the inequality
\begin{align}\label{eq:subadditivity}
	\Ent_\Phi\left[f(X^n)\right] \le \sum^n_{i=1} \E\left[\Ent_\Phi\big[f(X^n)\big|X^{\backslash i}\big]\right]
\end{align}
holds for any tuple $X^n = (X_1,\ldots,X_n)$ of independent random variables taking values in some spaces $\sX_1,\ldots,\sX_n$ and for any function $f \colon \sX_1 \times \ldots \times \sX_n \to \Reals^+$, such that $\Ent_\Phi[f(X^n)] < + \infty$. Here, $X^{\backslash i}$ denotes the $(n-1)$-tuple $(X_1,\ldots,X_{i-1},X_{i+1},\ldots,X_n)$ obtained by deleting $X_i$ from $X^n$. We are interested in the following question: what conditions on $\Phi$ ensure that this subadditivity property holds?

For example, if $\Phi(u)=u^2$, then $\Ent_\Phi[U]=\Var[U]$, and in this case the subadditivity property \eqref{eq:subadditivity} is the well-known \textit{Efron--Stein--Steele inequality} \cite{Efron_Stein,Steele}
\begin{align*}
	\Var[U] &\le \sum^n_{i=1} \E\left[\Var[U|X^{\backslash i}]\right], \qquad U = f(X^n).
\end{align*}
It is also not hard to show that the ``ordinary'' entropy $\Ent[U]$ [i.e., the $\Phi$-entropy with $\Phi(u) = u \log u$] is subadditive. In general, an induction argument can be used to show that subadditivity is equivalent to the following convexity property \cite{LO_Sobolev_Poincare}: for any two probability spaces $(\sX_1,\nu_1)$ and $(\sX_2,\nu_2)$ and any function $f \colon \sX_1 \times \sX_2 \to \Reals^+$,
\begin{align}\label{eq:phi_convexity}
	\Ent_\Phi\left[ \int_{\sX_2} f(X_1,x_2)\nu_2(\d x_2)\right] & \le \int_{\sX_2} \Ent_\Phi\big[f(X_1,x_2)\big]\nu_2(\d x_2),
\end{align}
where $X_1 \sim \nu_1$. The following criterion for subadditivity is useful \cite{LO_Sobolev_Poincare,Boucheron_etal_moment_inequalities}:
\begin{proposition}\label{prop:LO} Let $\cC$ be the class of all convex  functions $\Phi \colon \Reals^+ \to \Reals$ that are twice differentiable on $(0,\infty)$, and such that either $\Phi$ is affine or $\Phi'' > 0$ and $1/\Phi''$ is concave. Then the $\Phi$-entropy is subadditive for all $\Phi \in \cC$. Conversely, if $\Phi$ is twice differentiable with $\Phi'' > 0$ and the $\Phi$-entropy is subadditive, then $1/\Phi''$ is concave.
\end{proposition}

\section{Strong data processing inequalities}
\label{sec:SDPI}

We now turn to the main subject of the paper: strong data processing inequalities.

\begin{definition} Given an admissible pair $(\mu,K) \in \PProb(\sX) \times \Chan(\sY|\sX)$ and a function $\Phi \in \cF$, we say that $K$ satisfies a {\em $\Phi$-type strong data processing inequality (SDPI) at $\mu$ with constant $c \in [0,1)$}, or $\SDPI_\Phi(\mu,c)$ for short, if
	\begin{align}\label{eq:f_SDPI}
		D_\Phi(\nu K \| \mu K) \le c D_\Phi(\nu \| \mu)
	\end{align}
	for all $\nu \in \Prob(\sX)$. We say that $K$ satisfies $\SDPI_\Phi(c)$ if it satisfies $\SDPI_\Phi(\mu,c)$ for all $\mu \in \PProb(\sX)$.
\end{definition}
\noindent We are interested in the tightest constants in SDPIs; with that in mind, we define
\begin{align*}
	\eta_\Phi(\mu,K) &\deq \sup_{\nu \neq \mu} \frac{D_\Phi(\nu K \| \mu K)}{D_\Phi(\nu \| \mu)}, \\
	\eta_\Phi(K) &\deq \sup_{\mu \in \PProb(\sX)} \eta_\Phi(\mu,K).
\end{align*}
For future reference, we record the following straightforward results:
\begin{proposition}[Functional form of SDPI]\label{prop:functional_SDPI} Fix an admissible pair $(\mu,K)$ and let $(X,Y)$ be a random pair with probability law $\mu \otimes K$. Then $\eta_\Phi(\mu,K) \le c$ if and only if the inequality
\begin{align}\label{eq:entropy_production_inequality}
\Ent_\Phi[f(X)] \le \frac{1}{1-c} \E\left[ \Ent_\Phi[f(X)|Y]\right]
\end{align}
holds for all nonconstant $f \in \PFunc(\sX)$ with $\E[f(X)]=1$. Consequently,
	\begin{align}\label{eq:functional_SDPI}
		\eta_\Phi(\mu,K) &= \sup\left\{ \frac{\Ent_\Phi\left[K^* f(Y)\right]}{\Ent_\Phi\left[f(X)\right]} : f \in \PFunc(\sX), \, f \neq {\rm const},\, \E[f(X)]=1\right\} \\
		&= 1 - \inf\left\{ \frac{\E\left[ \Ent_\Phi[f(X)|Y]\right]}{\Ent_\Phi[f(X)]} :  f \in \PFunc(\sX), \, f \neq {\rm const},\, \E[f(X)]=1\right\}.
	\end{align}
\end{proposition}

\begin{proof} Fix a probability distribution $\nu \neq \mu$ and let $f = \d \nu/\d \mu$. Then $f \neq \text{const}$, $\E[f(X)]=1$, and
	\begin{align*}
		\frac{\d (\nu K)}{\d (\mu K)} = K^* f
	\end{align*}
		 by Lemma~\ref{lm:density_update} in the Appendix. Therefore,
	\begin{align*}
		D_\Phi(\nu \| \mu) = \Ent_\Phi\left[\frac{\d \nu}{\d \mu}(X)\right]  \qquad \text{and} \qquad D_\Phi(\nu K \| \mu K) = \Ent_\Phi\left[ \frac{\d(\nu K)}{\d(\mu K)}(Y)\right].
	\end{align*}
Conversely, for any nonconstant $f \in \PFunc(\sX)$ with $\E[f(X)]=1$ there exists a probability distribution $\nu \in \Prob(\sX)$ such that $\nu \neq \mu$ and $f = \d \nu/\d \mu$. In that case, the above formulas for the $\Phi$-entropies hold as well.

Now, if $c=1$, then \eqref{eq:entropy_production_inequality} holds trivially, so assume $c < 1$. In that case, the result follows from Eq.~\eqref{eq:entropy_production_inequality} and the law of total $\Phi$-entropy, Eq.~\eqref{eq:total_entropy}.
\end{proof}

\begin{definition}\label{def:gen_hom} We say that the $\Phi$-entropy $\Ent_\Phi[\cdot]$ is {\em homogeneous} if there exists some function $\kappa \colon (0,\infty) \to (0,\infty)$, such that the equality
\begin{align}\label{eq:gen_hom}
	\Ent_\Phi[c U] = \kappa(c)\Ent_\Phi[U]
\end{align}
holds for any nonnegative random variable $U$ such that $\Ent_\Phi[U] < + \infty$ and for any positive real number $c$.
\end{definition}
\noindent For example, $\Phi(u) = u \log u$ satisfies \eqref{eq:gen_hom} with $\kappa(c) = c$, while $\Phi(u) = \frac{u^\alpha-1}{\alpha-1}$, $\alpha > 1$, satisfies \eqref{eq:gen_hom} with $\kappa(c) = c^\alpha$.

\begin{proposition}\label{prop:functional_SDPI_gen_hom} Suppose that \eqref{eq:gen_hom} holds. Then
	\begin{align}
		\eta_\Phi(\mu,K) &= \sup\left\{ \frac{\Ent_\Phi\left[K^* f(Y)\right]}{\Ent_\Phi\left[f(X)\right]} :  f \in \PFunc(\sX), \, f \neq {\rm const}\right\}\nonumber \\
				&= 1 - \inf\left\{ \frac{\E\left[ \Ent_\Phi[f(X)|Y]\right]}{\Ent_\Phi[f(X)]} :  f \in \PFunc(\sX), \, f \neq {\rm const}\right\}. \label{eq:functional_SDPI_gen_hom}
	\end{align}
Moreover, if $\kappa$ is an invertible function, then
\begin{align}\label{eq:SDPI_levelset}
	\eta_\Phi(\mu,K) = \eta_\Phi(\mu,K,t) \deq \sup \left\{ \frac{\Ent_\Phi[K^*f (Y)]}{\Ent_\Phi[f(X)]} : f \in \PFunc(\sX), \Ent_\Phi[f(X)] \le t\right\}, \qquad \forall t > 0.
\end{align}
Again, $(X,Y)$ is a random pair with law $\mu \otimes K$.
\end{proposition}
\begin{proof} Eq.~\eqref{eq:functional_SDPI_gen_hom} is obvious from homogeneity. To prove \eqref{eq:SDPI_levelset}, pick an arbitrary nonconstant $f \in \PFunc(\sX)$ and let
	$$
	c = \kappa^{-1}\left( \frac{t}{\Ent_\Phi[f(X)]}\right).
	$$
Let $g = cf$. Then $\Ent_\Phi[g(X)] = \Ent_\Phi[cf(X)] = \kappa(c) \Ent_\Phi[f(X)] = t$. Therefore,
\begin{align*}
	\Ent_\Phi[K^*g(Y)] &\le \eta_\Phi(\mu,K,t) \Ent_\Phi[g(X)].
\end{align*}
Since $\Ent_\Phi[K^* g(Y)] = \Ent_\Phi[cK^* g(Y)] = \kappa(c) \Ent_\Phi[K^* g(Y)]$, and since $c > 0$ by the properties of $\kappa$, we conclude that $\Ent_\Phi[K^*f(Y)] \le \eta_\Phi(\mu,K,t)\Ent_\Phi[f(X)]$, which implies that $\eta_\Phi(\mu,K) \le \eta_\Phi(\mu,K,t)$. The reverse inequality, $\eta_\Phi(\mu,K) \ge \eta_\Phi(\mu,K,t)$, is obvious. 
\end{proof}

\begin{proposition}[Convexity in the kernel]\label{prop:SDPI_convexity} For a given choice of $\sX$, $\sY$, and $\mu \in \PProb(\sX)$, the SDPI constants $\eta_\Phi(\mu,K)$ and $\eta_\Phi(K)$ are convex in $K \in \Chan(\sY|\sX)$.
\end{proposition}
\begin{proof} For fixed $\nu,\mu \in \Prob(\sX)$, the functional $K \mapsto \frac{D_\Phi(\nu K \| \mu K)}{D_\Phi(\nu \| \mu)}$ is convex because of the joint convexity of $D_\Phi(\cdot \| \cdot)$ \cite[Lemma~4.1]{Csiszar_Shields_FnT}.\footnote{Joint convexity of $D_\Phi(\cdot\|\cdot)$ follows from the fact that, for any convex function $\Phi : \Reals^+ \to \Reals$, the \textit{perspective function} $(p,q) \mapsto q \Phi(p/q)$ is jointly convex in $(p,q) \in \Reals^+ \times \Reals^+$ \cite[Prop.~2.2.1]{Hiriart_book}.} Now,
	\begin{align*}
		\eta_\Phi(\mu,K) &= \sup_\nu \frac{D_\Phi(\nu K \| \mu K)}{D_\Phi(\nu \| \mu)} \qquad \text{and} \qquad \eta_\Phi(K) = \sup_{\mu} \sup_{\nu} \frac{D_\Phi(\nu K \| \mu K)}{D_\Phi(\nu \| \mu)}
	\end{align*}
are pointwise suprema of convex functionals of $K$, and therefore are convex in $K$.
\end{proof}

\subsection{A universal upper bound via Markov contraction}

A universal upper bound on $\eta_\Phi(K)$ was originally obtained by Cohen et al.~\cite{Cohen_etal_dataproc} in the discrete case and subsequently extended by Del Moral et al.~\cite{DelMoral_contraction} to the general case. We state this bound and give a proof which is more information-theoretic in nature:
\begin{theorem}\label{thm:Markov_contraction_bound} Define the {\em Dobrushin contraction coefficient \cite{Dobrushin_CLT_MC_1,Dobrushin_CLT_MC_2}} of a channel $K \in \Chan(\sY|\sX)$ by
	\begin{align}\label{eq:Dobrushin}
		\vartheta(K) \deq \max_{x,x' \in \sX} \| K(\cdot|x) - K(\cdot|x') \|_\TV.
	\end{align}
	Then for any $\Phi \in \cF$ we have
	\begin{align}\label{eq:universal_SDPI_bound}
		\eta_\Phi(K) \le  \vartheta(K).
	\end{align}
	Moreover, $\eta_\TV(K) \equiv \vartheta(K)$.
\end{theorem}
\begin{proof} By the integral representation \eqref{eq:information_integral}, it suffices to show that \eqref{eq:universal_SDPI_bound} holds for the statistical informations $B_\lambda(\cdot\|\cdot)$, $0 \le \lambda \le 1$. For that, we need the following {\em strong Markov contraction lemma} \cite[Lemma~3.2]{Cohen_etal_dataproc}: for any signed measure $\tilde{\nu}$ on $\sX$ and any Markov kernel $K \in \Chan(\sY|\sX)$,
	\begin{align}\label{eq:strong_Markov_contraction}
		\| \tilde{\nu} K \|_\TV \le \vartheta(K) \| \tilde{\nu} \|_\TV + \frac{1-\vartheta(K)}{2} |\tilde{\nu}(\sX)|.
	\end{align}
Let $\tilde{\nu} = \lambda \nu - \bar{\lambda} \mu$. Then $\tilde{\nu} K = \lambda \nu K - \bar{\lambda} \mu K$ and $\tilde{\nu}(\sX) = 2\lambda - 1$. Thus, using \eqref{eq:strong_Markov_contraction}, we get
\begin{align*}
	\| \lambda \nu K - \bar{\lambda} \mu K \|_\TV  \le \vartheta(K) \| \lambda \nu - \bar{\lambda} \mu \|_\TV + \frac{1-\vartheta(K)}{2}|1-2\lambda|.
\end{align*}
Therefore,
\begin{align*}
	B_\lambda(\nu K\| \mu K) &= \frac{1}{2} \| \lambda \nu K - \bar{\lambda} \mu K \|_\TV - \frac{1}{2} |1-2\lambda| \\
	&\le \vartheta(K) \cdot \left(\frac{1}{2} \| \lambda \nu - \bar{\lambda} \mu \|_\TV - \frac{1}{2}|1-2\lambda|\right) \\
	&= \vartheta(K) \cdot B_\lambda(\nu\|\mu).
\end{align*}
This establishes the bound \eqref{eq:universal_SDPI_bound}. It remains to show that this bound is achieved for $\| \cdot \|_{\TV}$.

To that end, let us first assume that $|\sX| > 2$. Let $x_0,x_1 \in \sX$ achieve the maximum in \eqref{eq:Dobrushin}, pick some $\eps_1,\eps_2,\eps \in (0,1)$ such that $\eps_1 \neq \eps_2$, $\eps_1+\eps < 1$, $\eps_2 + \eps < 1$, and consider the following probability distributions:
\begin{itemize}
	\item $\nu$ that puts the mass $1-\eps_1-\eps$ on $x_0$, $\eps_1$ on $x_1$, and distributes the remaining mass of $\eps$ evenly among the set $\sX \backslash \{x_0,x_1\}$;
	\item $\mu$ that puts the mass $1-\eps_2-\eps$ on $x_0$, $\eps_2$ on $x_1$, and distributes the remaining mass of $\eps$ evenly among the set $\sX \backslash \{x_0,x_1\}$.
\end{itemize}
Then a simple calculation gives
\begin{align*}
	\| \nu - \mu \|_\TV &= |\eps_1 - \eps_2|  \\
	\| \nu K - \mu K\|_\TV &= |\eps_1 - \eps_2| \cdot \| K(\cdot|x_0) - K(\cdot|x_1) \|_\TV \\
	&= \vartheta(K) \cdot \| \nu - \mu \|_\TV.
\end{align*}
For $|\sX|=2$, the idea is the same, except that there is no need for the extra slack $\eps$.
\end{proof}

\begin{remark} {\em Theorem~\ref{thm:Markov_contraction_bound} says that any channel $K$ with $\vartheta(K) < 1$ satisfies an SDPI for any $\Phi \in \cF$ at any reference input distribution $\mu \in \Prob(\sX)$. However, the bounds it gives  are generally loose. For example, for $K = \BSC(\eps)$ with $\eps \in (0,1)$, we have $\vartheta(K) = |1-2\eps| < 1$, so by Theorem~\ref{thm:Markov_contraction_bound}
\begin{align*}
	\eta_\Phi\left(\Bernoulli(p),\BSC(\eps)\right) \le |1-2\eps| < 1
\end{align*}
for all $\Phi \in \cF$ and all $p \in [0,1]$. However, as we know from \cite{Ahlswede_Gacs_hypercont},
\begin{align*}
	\eta\left(\Bernoulli(1/2),\BSC(\eps)\right) = (1-2\eps)^2 < |1-2\eps|.
\end{align*}
Later on, we will develop tighter bounds on SDPI constants for a broad class of $\Phi$-entropies.\hfill$\diamond$}
\end{remark}
	
\begin{remark}{\em Suppose that the channel $K \in \Chan(\sY|\sX)$ has the following property: There exist a constant $0 < \alpha \le 1$ and a probability distribution $\tilde{\mu} \in \Prob(\sY)$, such that
\begin{align}\label{eq:Doeblin}
	K(y|x) \ge \alpha \tilde{\mu}(y)
\end{align}
for all $x \in \sX$ and $y \in \sY$ (in Markov chain theory, this is known as a Doeblin minorization condition \cite[Sec.~4.3.3]{Cappe_etal_HMM}). Then $\eta_\Phi(K) \le 1-\alpha$. This bound can be proved using a nice operational argument. Indeed, if \eqref{eq:Doeblin} holds, then
\begin{align*}
	\tilde{K}(y|x) \deq \frac{K(y|x)-\alpha \tilde{\mu}(y)}{1-\alpha}, \qquad (x,y) \in \sX \times \sY
\end{align*}
defines a channel from $\sX$ to $\sY$. Let $\se$ be a special erasure symbol, and let $E_\alpha \in \Chan(\sX \cup \{\se\} | \sX)$ denote the symmetric erasure channel on $\sX$ with erasure probability $\alpha$: any input symbol $x \in \sX$ is erased with probability $\alpha$ and reproduced exactly with probability $\bar{\alpha}$. Then a simple calculation shows that $K = T \circ E_\alpha$,  where the channel $T \in \Chan(\sY|\sX \cup \{\se\})$ is defined by
\begin{align*}
	T(\cdot|x) &= \tilde{K}(\cdot|x), \qquad x \in \sX \\
	T(\cdot|\se) &= \tilde{\mu}(\cdot).
\end{align*}
In that case, for any $\mu,\nu \in \Prob(\sX)$,
\begin{align*}
	D_\Phi(\nu K \| \mu K) &= D_\Phi((\nu E_\alpha)T \| (\mu E_\alpha) T ) \\
	&\le D_\Phi(\nu E_\alpha \| \mu E_\alpha) \\
	&= D_\Phi(\bar{\alpha} \nu + \alpha \delta_\se \| \bar{\alpha} \mu + \alpha \delta_\se) \\
	&\le \bar{\alpha} D_\Phi(\nu \| \mu),
\end{align*}
where the first inequality is by the usual data processing inequality, while the second inequality is by convexity. It is not hard to show that if \eqref{eq:Doeblin} holds, then $\vartheta(K) \le 1-\alpha$.\hfill$\diamond$}
\end{remark}

\subsection{Bounds via maximal correlation}
\label{ssec:maxcorr}

For any pair $(\mu,K) \in \Prob(\sX) \times \Chan(\sY|\sX)$, the {\em maximal correlation}  is defined as
\begin{align*}
	S(\mu,K) &\deq \sup_{f,g} \E[f(X)g(Y)],
\end{align*}
where $(X,Y) \sim \mu \otimes K$, and the supremum is over all $f \in \Func(\sX)$, $g \in \Func(\sY)$ satisfying $\E[f(X)] = \E[g(Y)]=0$ and $\E[f^2(X)] = \E[g^2(Y)]=1$ (see \cite{Witsenhausen_correlation} and the references therein). The square of $S(\mu,K)$ is the SDPI constant of the pair $(\mu,K)$ for the $\chi^2$-divergence:
\begin{theorem}\label{thm:maxcorr} Consider the $\chi^2$-divergence
	\begin{align*}
		\chi^2(\nu \| \mu) = \E_\mu\left[\left(\frac{\d \nu}{\d \mu}-1\right)^2\right] \equiv \Var_\mu\left[ \frac{\d \nu}{\d \mu}\right].
	\end{align*}
Then, for any $(\mu,K) \in \Prob(\sX) \times \Chan(\sY|\sX)$,
	\begin{align*}
		\eta_{\chi^2}(\mu,K) = S^2(\mu,K).
	\end{align*}
\end{theorem}
\begin{remark}{\em This result has appeared in the literature in different forms (see, e.g., \cite{DelMoral_contraction}). We give a short proof for completeness.\hfill$\diamond$}
\end{remark}
\begin{proof} For this proof, it is convenient to use operator-theoretic ideas, following Witsenhausen \cite{Witsenhausen_correlation} (see also \cite{Sarmanov_maxcorr}). If we equip the space $\Func(\sX)$ with the inner product
	\begin{align*}
		\ave{f,g}_{\mu} \deq \E[f(X)g(X)], \qquad \text{where } X \sim \mu
	\end{align*}
then it becomes the Hilbert space $L^2(\sX,\mu)$; the Hilbert space $L^2(\sY,\mu K)$ is constructed in the same way. Moreover, the channels $K$ and $K^*$ become mutually adjoint linear operators $K : L^2(\sY,\mu K) \to L^2(\sX,\mu)$ and $K^* : L^2(\sX,\mu) \to L^2(\sY,\mu K)$, i.e.,
\begin{align*}
	\ave{f,Kg}_{\mu} = \ave{K^* f,g}_{\mu K}, \quad \forall f \in L^2(\sX,\mu), g \in L^2(\sY,\mu K).
\end{align*}
For $f = \d \nu/\d \mu$, we have $\chi^2(\nu \| \mu) = \Var\big[f(X)\big]$ and $\chi^2(\nu K \| \mu K) = \Var\big[K^* f(Y)\big]$. Using this together with the fact that $\Var[U+c] = \Var[U]$ for any $c \in \Reals$ and that $\E[K^* f(Y)] = \E[f(X)]$, we can write
	\begin{align*}
		\eta_{\chi^2}(\nu,K) &= \sup_{\nu \neq \mu} \frac{\chi^2(\nu K \| \mu K)}{\chi^2(\nu \| \mu)}\\ &= \sup_{f \in \cH_0(\sX)} \frac{\Var\big[K^* f(Y)\big]}{\Var\big[f(X)\big]},
	\end{align*}
	where $\cH_0(\sX)$ is the closed linear subspace of $L^2(\sX,\mu)$ consisting of all $f$ satisfying $\ave{f,1}_\mu=0$, i.e., $\E[f(X)]=0$. For any $f \in \cH_0(\sX)$,
	\begin{align*}
		\Var \big[f(X)\big] = \| f \|^2_\mu, \qquad \Var\big[K^* f(Y)\big] = \| K^* f \|^2_{\mu K}.
	\end{align*}
Since $K$ and $K^*$ are adjoint operators, we have
	\begin{align*}
		\| K^* f \|^2_{\mu K} &= \ave{K^* f, K^* f}_{\mu K} = \ave{f, KK^* f}_\mu,
	\end{align*}
which gives
	\begin{align*}
		\eta_{\chi^2}(\mu,K) = \sup_{f \in \cH_0(\sX)} \frac{\ave{f,KK^* f}_{\mu}}{\ave{f,f}_\mu}.
	\end{align*}
Moreover, $K^*$ maps $\cH_0(\sX)$ into $\cH_0(\sY)$, and $K$ maps $\cH_0(\sY)$ into $\cH_0(\sX)$. Thus, by the Courant--Fischer--Weyl minimax principle \cite{Bhatia_Matrix_Analysis}, $\eta_{\chi^2}(\nu,\mu)$ is the largest eigenvalue of the operator $KK^* : \cH_0(\sX) \to \cH_0(\sX)$. The square root of this largest eigenvalue is the largest singular value of the operator $K^* : \cH_0(\sX) \to \cH_0(\sY)$, so, by definition,
\begin{align*}
	\sqrt{\eta_{\chi^2}(\nu,\mu)} &= \sup_{f,g} \ave{K^* f,g} \\
	&= \sup_{f,g} \E[\E[f(X)|Y]g(Y)] \\
	&= \sup_{f,g} \E[f(X)g(Y)],
\end{align*}
where the supremum is over all $f \in \cH_0(\sX)$ and $g \in \cH_0(\sY)$ with $\| f \|_\mu = \| g \|_{\mu K} = 1$. This is precisely the maximal correlation $S(\mu,K)$.\end{proof}

\begin{remark}[Maximum correlation and the spectral gap]\label{rem:spectral_gap}{\em In the literature on Markov chains (see, e.g., \cite{Diaconis_Saloff_logsob,Montenegro_Tetali,Levin_Peres_Wilmer}), one often sees the following definition: given a pair $(\mu,K) \in \Prob(\sX) \times \Chan(\sX|\sX)$ such that $\mu$ is invariant w.r.t.\ $K$, i.e., $\mu = \mu K$, the {\em (absolute) spectral gap} of $K$ is equal to
	\begin{align*}
		\gamma_*(\mu,K) \deq 1-\sup_{f \in \cH_0(\sX)} \frac{\| K^* f \|_\mu}{\| f \|_\mu}
	\end{align*}
(here we are using the Hilbert space notation from the proof above). Thus, the spectral gap and the maximal correlation are related by $\gamma_*(\mu, K) = 1-S(\mu,K) = 1-\sqrt{\eta_{\chi^2}(\mu,K)}$.\hfill$\diamond$}
\end{remark}

The maximal correlation $S^2(\mu,K)$ also provides a \textit{lower bound} on the SDPI constants $\eta_\Phi(\mu,K)$ for a certain subset of $\cF$:
\begin{theorem}\label{thm:chi_2_lower_bound} For any $\Phi \in \cF$ which is three times differentiable and has $\Phi''(1) > 0$, we have
	\begin{align}
		\eta_\Phi(\mu,K) &\ge S^2(\mu,K), \label{eq:chi2_pointwise} \\
		\eta_\Phi(K) &\ge S^2(K), \label{eq:chi2_global}
	\end{align}
	where $S^2(K) \deq \sup_{\mu \in \Prob(\sX)}S^2(\mu,K)$.
\end{theorem}
\begin{remark} {\em The second bound, Eq.~\eqref{eq:chi2_global}, was proved by Cohen et al.~\cite{Cohen_etal_dataproc}, generalizing the results of Ahlswede and G\'acs \cite{Ahlswede_Gacs_hypercont} for $\Phi(u)=u\log u$. However, more or less the same proof technique also gives the distribution-dependent bound \eqref{eq:chi2_pointwise}. A recent paper of Polyanskiy and Wu \cite{PW_BayesSDPI} presents an extension of Theorem~\ref{thm:chi_2_lower_bound} to abstract alphabets.}
\end{remark}
\begin{proof} Without loss of generality, we assume that $\Phi(1) = 0$. Let us expand $\Phi$ in a Taylor series around $u=1$:
	\begin{align*}
		\Phi(u) &= \Phi(1) + \Phi'(1)(u-1)  + \frac{1}{2}\Phi''(1)(u-1)^2 + o\big((u-1)^2\big) \\
		&=  \Phi'(1)(u-1) + \frac{1}{2}\Phi''(1)(u-1)^2 + o\big((u-1)^2\big),
	\end{align*}
	where the second step uses the fact that $\Phi(1)=0$. Therefore, for any bounded real-valued random variable $U$ and any $\eps > 0$ such that $1+\eps U \ge 0$ a.s., we have
	\begin{align*}
		\Ent_\Phi[1+\eps U] = \frac{\Phi''(1)}{2}\eps^2 \Var[U] + O(\eps^3).
	\end{align*}
	Now, fix an admissible pair $(\mu,K)$. For any $\nu \neq \mu$, consider the mixture $\nu_\eps \deq \bar{\eps}\mu + \eps \nu$. Let $f = \d \nu/\d \mu - 1$. Then
	\begin{align*}
		D_\Phi(\nu_\eps \| \mu) &= \Ent_\Phi[1+\eps f(X)] \\
		&= \frac{\Phi''(1)}{2}\eps^2 \Var[f(X)] + o(\eps^2) \\
		&= \frac{\Phi''(1)}{2}\eps^2 \chi^2(\nu \| \mu) + o(\eps^2)
	\end{align*}
	and
	\begin{align*}
		D_\Phi(\nu_\eps K \| \mu K) &= \Ent_\Phi[1 + \eps K^* f(Y)] \\
		&= \frac{\Phi''(1)}{2}\eps^2 \Var[K^* f(Y)] + o(\eps^2) \\
		&= \frac{\Phi''(1)}{2}\eps^2 \chi^2(\nu K \| \mu K) + o(\eps^2),
	\end{align*}
	where in the first step we have used Lemma~\ref{lm:density_update} in the Appendix and the linearity of $K^*$. Using the fact that $\Phi''(1) > 0$, for any $\eps > 0$ we have
	\begin{align*}
		\eta_\Phi(\mu,K) &\ge \sup_{\nu \neq \mu} \frac{D_\Phi(\nu_\eps K \| \mu K)}{D_\Phi(\nu_\eps\|\mu)} \\
		&= \sup_{\nu \neq \mu} \frac{\chi^2(\nu K \| \mu K) + o(\eps)}{\chi^2(\nu \| \mu) + o(\eps)}.
	\end{align*}
	Taking the limit as $\eps \searrow 0$, we get
	\begin{align*}
		\eta_\Phi(\mu,K) \ge \sup_{\nu \neq \mu} \frac{\chi^2(\nu K \| \mu K)}{\chi^2( \nu \| \mu )} = \eta_{\chi^2}(\mu,K).
	\end{align*}
	This proves \eqref{eq:chi2_pointwise}, and \eqref{eq:chi2_global} follows after taking the supremum over all $\mu$.
\end{proof}
\noindent For example, the function $\Phi(u) = u \log u$ that induces the usual relative entropy satisfies the conditions of Theorem~\ref{thm:chi_2_lower_bound}, as does the function $\Phi(u)=(\sqrt{u}-1)^2$ that gives rise to the squared Hellinger distance.

Under additional regularity conditions on $\Phi$, we can obtain an upper bound on $\eta_\Phi$ which is proportional to the maximal correlation $S^2(\mu,K)$:

\begin{theorem}\label{thm:maxcorr_UB} Suppose that $\Phi \in \cF$ is twice differentiable, strictly convex, has a nonincreasing second derivative, and the function
\begin{align}\label{eq:first_diff_0}
	\Psi(u) \deq \frac{\Phi(u) - \Phi(0)}{u}
\end{align}
is concave. Then, for any admissible pair $(\mu,K)$,
	\begin{align}\label{eq:maxcorr_UB}
		\eta_\Phi(\mu,K) \le \frac{2\Psi'(1)}{\Phi''(1/\mu_*)} S^2(\mu,K),
	\end{align}
	where $\mu_* \deq \min_{x \in \sX} \mu(x)$ is the smallest mass of $\mu$.
\end{theorem}

\begin{remark}{\em It can be shown (see, e.g., \cite[Lm.~14.5]{Boucheron_etal_concentration_book}) that if $\Phi \in \cF \cap \cC$, where the function class $\cC$ is defined in Proposition~\ref{prop:LO}, then the function $\Psi$ defined in \eqref{eq:first_diff_0} is concave. For example, if $\Phi(u) = u \log u$, then $\Psi(u) = \log u$.\hfill$\diamond$}
\end{remark}

\begin{proof} Let $(X,Y)$ be a random pair with law $\mu \otimes K$. Fix any probability distribution $\nu \neq \mu$ and let $f = {\d\nu}/{\d\mu}$. Then we have the following chain of estimates:
	\begin{align}
		\Ent_\Phi[K^* f(Y)] &\le \Psi'(1) \Var[K^* f(Y)] \label{eq:maxcorr_UB_1} \\
		&\le \Psi'(1) S^2(\mu,K) \Var[f(X)] \label{eq:maxcorr_UB_2} \\
		&\le \frac{2 \Psi'(1)}{\Phi''(\| f(X) \|_\infty)} S^2(\mu,K) \Ent_\Phi[f(X)], \label{eq:maxcorr_UB_3}
	\end{align}
	where \eqref{eq:maxcorr_UB_1} is by Lemma~\ref{lm:f_entropy_UB} in Appendix~\ref{app:lemmas}, \eqref{eq:maxcorr_UB_2} is by Theorem~\ref{thm:maxcorr}, and \eqref{eq:maxcorr_UB_3} is by Lemma~\ref{lm:f_entropy_LB} in Appendix~\ref{app:lemmas}. Now, since
	\begin{align*}
		 \| f(X) \|_\infty = \left\| \frac{\d\nu}{\d\mu} \right\|_\infty \le \frac{1}{\mu_*},
	\end{align*}
we have $\Phi''(\| f(X) \|_\infty) \ge \Phi''(1/\mu_*)$. By the arbitrariness of $\nu$ (and hence $f$), we obtain \eqref{eq:maxcorr_UB}.
\end{proof}
For example, functions of the form $\Phi_p(u) = \frac{u^p-1}{p-1}$ for $1 < p \le 2$ satisfy the conditions of the theorem with $\Psi_p(u) = \frac{u^{p-1}}{p-1}$ and $\Phi''_p(u) = pu^{p-2}$. This gives the bound
\begin{align}\label{eq:p_divergence}
	\eta_{\Phi_p}(\mu,K) \le \frac{2\mu_*^{p-2}}{p} S^2(\mu,K).
\end{align}
Note that $\Phi_2(u) = u^2-1$ induces the $\chi^2$-divergence, so $\eta_{\Phi_2}(\mu,K) = \eta_{\chi^2}(\mu,K) = S^2(\mu,K)$, and in that case the bound \eqref{eq:p_divergence} holds with equality. Moreover, as $p \searrow 1$, we have $\Ent_{\Phi_p}[U] \to \Ent[U]$, and in that limit \eqref{eq:p_divergence} becomes
\begin{align}\label{eq:1_divergence}
	\eta(\mu,K) \le \frac{2}{\mu_*}S^2(\mu,K).
\end{align}
Of course, the bound \eqref{eq:p_divergence} is nontrivial only if $S^2(\mu,K) < \frac{p}{2\mu_*^{p-2}}$; similarly, the bound \eqref{eq:1_divergence} is nontrivial only if $S^2(\mu,K) < \frac{\mu_*}{2}$. As recently shown by Makur and Zheng \cite{MakurZheng_contraction}, the constant $2$ in \eqref{eq:1_divergence} can be reduced to $1$, but it is not clear how to extend their techniques to $\Phi(u) \neq u \log u$.

\subsection{Upper bounds for operator convex $\Phi$}
\label{ssec:opconv}

Theorem~\ref{thm:maxcorr_UB} gives an upper bound on the SDPI constant $\eta_\Phi(\mu,K)$ in terms of the squared maximal correlation $S^2(\mu,K)$, but this bound has a multiplicative constant that depends on $\mu$. Given the lower bound of Theorem~\ref{thm:chi_2_lower_bound}, it is natural to ask whether there is a matching upper bound without such a multiplicative constant. A partial result in this direction was obtained by Choi et al.~\cite{Choi_Ruskai_Seneta}, who showed that the equality $\eta_\Phi(K) = \eta_{\chi^2}(K) = S^2(K)$ holds for all functions $\Phi \in \cF$ that are \textit{operator convex} (see below for definitions). In this section, we will derive a \textit{distribution-dependent} upper bound on $\eta_\Phi(\mu,K)$ that implies the result of Choi et al.

In preparation for this result, we first need some facts from matrix analysis \cite{Bhatia_Matrix_Analysis}. Let $H_n$ denote the space of all $n \times n$ Hermitian matrices, and let $H_n(I)$ denote the subset of $H_n$ consisting of all matrices whose eigenvalues lie in a given finite or infinite interval $I$ of the real line. Any function $\Phi : I \to \Reals$ can be extended to a matrix-valued function $\Phi : H_n(I) \to H_n$ as follows:
\begin{itemize}
	\item if $A \in H_n(I)$ is diagonal, i.e., $A = \diag(a_1,\ldots,a_n)$ for some $a_1,\ldots,a_n \in I$, then we let
	\begin{align*}
		\Phi(A) \deq \diag \left(\Phi(a_1),\ldots,\Phi(a_n)\right).
	\end{align*}
	\item if $A \in H_n(I)$ can be diagonalized as $A = U\Lambda U^*$, where $U$ is a unitary $n\times n$ matrix and $\Lambda \in H_n(I)$ is diagonal, then we let
	\begin{align*}
		\Phi(A) \deq U \Phi(\Lambda) U^*.
	\end{align*}
\end{itemize}
We introduce the following partial order on $H_n$: given any two $A,B \in H_n$, we write $A \preceq B$ if $B - A$ is positive semidefinite. We say that a  function $\Phi : I \to \Reals$ is \textit{$n$-convex} if
\begin{align*}
	\Phi(\lambda A + (1-\lambda)B) \preceq \lambda \Phi(A) + (1-\lambda) \Phi(B)
\end{align*}
for all $A,B \in H_n(I)$ and all $\lambda \in [0,1]$. If $\Phi$ is $n$-convex for all $n \in \Naturals$, then we say that it is \textit{operator convex}. By definition, any operator convex function is \textit{a fortiori} convex in the ordinary sense, but the converse is generally not true. We are particularly interested in functions $\Phi : \Reals^+ \to \Reals$ that are operator convex; here are some examples and counterexamples \cite[Ch.~V]{Bhatia_Matrix_Analysis}:
\begin{itemize}
	\item $\Phi(u) = u \log u$ is operator convex;
	\item $\Phi(u) = u^p$ is operator convex if and only if $p \in [-1,0] \cup [1,2]$.
	\item $\Phi(u) = -u^p$ is operator convex for $0 \le p \le 1$.
\end{itemize}
In general, it is not easy to determine whether a given function is operator convex. However, there is a deep result known as \textit{Loewner's theorem} \cite{Hansen_Loewner_thm}, which shows that operator convex functions possess very special integral representations: 

\begin{theorem}\label{thm:Loewner} A function $\Phi : \Reals^+ \to \Reals$ with $\Phi(0) = 0$ is operator convex if and only if there exist some constants $\alpha \in \Reals,\beta \ge 0$ and a positive measure $\upsilon$ on $\Reals^+$ satisfying $\int^\infty_0 (1+t^2)^{-1}\upsilon(\d t) < \infty$, such that
	\begin{align}\label{eq:Loewner}
		\Phi(u) = \alpha u + \beta u^2 + \int^\infty_0 \left( \frac{tu}{1+t^2} - \frac{u}{u+t}\right) \upsilon(\d t).
	\end{align}
\end{theorem}
\noindent For example, the operator convex function $\Phi(u)=u\log u$ can be represented in the form \eqref{eq:Loewner} with $\alpha = \beta = 0$ and with $\upsilon$ given by the restriction of the Lebesgue measure to $\Reals^+$ \cite[Example~V.4.18]{Bhatia_Matrix_Analysis}; the operator convex function $\Phi(u) = u^p$, $1 < p < 2$, can be represented in the form \eqref{eq:Loewner} with
$$
\alpha = \cos \frac{\pi p}{2}, \qquad \beta = 0, \qquad \upsilon(\d t) = \frac{\sin (\pi p)}{\pi} t^p \d t
$$
\cite[Example~V.4.19]{Bhatia_Matrix_Analysis}.

We also recall the definition of the Le Cam divergence with parameter $\lambda \in (0,1)$, cf.~Eq.~\eqref{eq:LC}:
\begin{align*}
	\LC_\lambda(\nu \| \mu) &= \lambda \bar{\lambda}\, \E_\mu \left[ \frac{(\d \nu/\d \mu-1)^2}{\lambda \d \nu/\d \mu +\bar{\lambda}}\right] = 1 - \E_\mu \left[ \frac{\d \nu/\d \mu}{\lambda \d \nu/\d \mu + \bar{\lambda}}\right],
\end{align*}
which is a $\Phi$-divergence with
\begin{align*}
	\Phi(u) = 1 - \frac{u}{\lambda u +\bar{\lambda}}.
\end{align*}
Note that $\LC_0(\cdot \| \cdot) = \LC_1(\cdot \| \cdot) = 0$. For $\lambda \in (0,1)$, consider the SDPI constant
\begin{align*}
	\eta_{\LC_\lambda}(\mu,K) = \sup_{\nu \neq \mu} \frac{\LC_\lambda(\nu K \| \mu K)}{\LC_\lambda(\nu\|\mu)}.
\end{align*}
Now we are in a position to state our result:

\begin{theorem}\label{thm:opconv_bounds} Suppose that $\Phi \in \cF$ is operator convex. Then
	\begin{align}\label{eq:opconv_bounds}
S^2(\mu,K) \le \eta_\Phi(\mu,K) \le \max\left(S^2(\mu,K),\sup_{0 < \lambda < 1} \eta_{\LC_\lambda}(\mu,K)\right).
	\end{align}
\end{theorem}
\begin{remark} {\em Since all explicit examples of functions in $\cC$ seem to be operator convex, it is tempting to think that all operator convex $\Phi$ are elements of the function class $\cC$ (cf.~Proposition~\ref{prop:LO}). However, this is not the case. For example, the function $\Phi(u) = (\sqrt{u}-1)^2$, which generates the Hellinger divergence, is operator convex. However, $1/\Phi''(u) = 2u^{3/2}$ is not concave, so $\Phi \not\in \cC$. \hfill$\diamond$}
\end{remark}

\begin{proof} By Loewner's theorem (Theorem~\ref{thm:Loewner}), $\Phi$ admits the integral representation \eqref{eq:Loewner}. Any $\Phi$ that can be represented in this form is infinitely differentiable and strictly convex at $u=1$. Therefore, $\eta_\Phi(\mu,K) \ge S^2(\mu,K)$ by Theorem~\ref{thm:chi_2_lower_bound}. This establishes first inequality in Eq.~\eqref{eq:opconv_bounds}.
	
Now we prove the second inequality in \eqref{eq:opconv_bounds}. First, let us rewrite \eqref{eq:Loewner} as
\begin{align*}
	\Phi(u) &= \beta u^2 - \int^\infty_0 \frac{u}{u+t}\upsilon(\d t) + A(u),
\end{align*}
where $A(u)$ is an affine function. A change of variables $\lambda = \frac{1}{t+1}$ gives
\begin{align}\label{eq:Loewner_transformed}
	\Phi(u) &= \beta u^2 - \int^1_0  \frac{\lambda u}{\lambda u + \bar{\lambda}}\Upsilon(\d\lambda) + A(u),
\end{align}
where $\Upsilon$ is some positive measure on $[0,1]$. Since any two elements of $\cF$ that differ by an affine function determine the same divergence, Eq.~\eqref{eq:Loewner_transformed} allows us to express $D_\Phi(\nu \| \mu)$ as
	\begin{align*}
		D_\Phi(\nu \| \mu) &= \beta \chi^2(\nu \| \mu) + \int^1_0 \lambda \LC_\lambda(\nu \| \mu) \Upsilon(\d\lambda)
	\end{align*}
The same holds for $\nu K$ and $\mu K$, so
\begin{align*}
	 D_\Phi(\nu K \| \mu K)
	 &= \beta \chi^2(\nu K \| \mu K) + \int^1_0 \lambda \LC_\lambda(\nu K \| \mu K) \Upsilon(\d\lambda) \\
	&\le \beta S^2(\mu,K)\chi^2(\nu \| \mu) + \int^1_0 \lambda \eta_{\LC_\lambda}(\mu,K)\LC_\lambda(\nu \| \mu)\Upsilon(\d\lambda) \\
	&\le \max\left(S^2(\mu,K),\sup_{0 < \lambda < 1} \eta_{\LC_\lambda}(\mu,K)\right) \cdot D_\Phi(\nu \| \mu).
\end{align*}
\end{proof}
We can now recover the result of Choi et al.~\cite{Choi_Ruskai_Seneta} as a corollary:
\begin{corollary}\label{cor:Choi} Suppose that $\Phi \in \cF$ is operator convex. Then
	\begin{align*}
		\eta_\Phi(K) = S^2(K)
	\end{align*}
	for any discrete channel $K$.
\end{corollary}
\begin{remark} {\em Since $\Phi(u) = u \log u$ is operator convex, this is a broad generalization of a result of Ahlswede and G\'acs \cite[Thm.~8]{Ahlswede_Gacs_hypercont}. It should be emphasized that Corollary~\ref{cor:Choi} does not mean that $\eta_\Phi(\mu,K) = S^2(\mu,K)$ for a given input distribution $\mu \in \PProb(\sX)$; however, this may be the case for specific choices of $\mu$ and $K$, as we show in the example after the proof.\hfill$\diamond$}
\end{remark}
\begin{proof} It suffices to show that
	\begin{align*}
		\sup_{0 < \lambda < 1}\eta_{\LC_\lambda}(\mu,K) \le S^2(K).
	\end{align*}
To that end, we first note that the Le Cam divergence $\LC_\lambda(\nu\|\mu)$ can be written as a convex combination of two $\chi^2$-divergences:
\begin{align*}
	\LC_\lambda(\nu \| \mu) = \lambda \chi^2(\nu \| \lambda \nu +\bar{\lambda} \mu) + \bar{\lambda} \chi^2(\mu \| \lambda \nu +\bar{\lambda}\mu).
\end{align*}
From this, it follows that
\begin{align*}
	\LC_\lambda(\nu K \| \mu K) 
	&= \lambda \chi^2(\nu K \| \lambda \nu K +\bar{\lambda} \mu K) + \bar{\lambda} \chi^2(\mu K \| \lambda \nu K +\bar{\lambda} \mu K) \\
	&\le  S^2(K) \left[\lambda\chi^2(\nu \| \lambda \nu + \bar{\lambda}\mu) + \bar{\lambda}\chi^2(\mu \| \lambda \nu + \bar{\lambda}\mu)\right] \\
	&= S^2(K) \LC_\lambda(\nu \| \mu).
\end{align*}
\end{proof}

\begin{example} {\em Let $\mu = \Bernoulli(1/2)$ and $K = \BSC(\eps)$. For any $q \neq 1/2$ and $\nu = \Bernoulli(q)$, we have $\mu K = \Bernoulli(1/2)$ and $\nu K = \Bernoulli(q \star \eps)$. Moreover,
	\begin{align*}
		\chi^2(\nu \| \mu) &= \chi^2(\Bernoulli(q)\|\Bernoulli(1/2)) = (1-2q)^2
	\end{align*}
	and
	\begin{align*}
		\chi^2(\nu K \| \mu K) &= \chi^2(\Bernoulli(q\star\eps)\|\Bernoulli(1/2)) =(1-2(q\star\eps))^2 =(1-2\eps)^2(1-2q)^2.
	\end{align*}
	Therefore,
	\begin{align}
		\eta_{\chi^2}(\mu,K) \equiv S^2(\mu,K) = (1-2\eps)^2. \label{eq:BSC_chi_2_ratio}
	\end{align}
	Moreover, for any $\lambda \in (0,1)$,
	\begin{align*}
		\LC_\lambda(\nu \| \mu) &= \LC_\lambda(\Bernoulli(q)\|\Bernoulli(1/2)) \\
		&= 1 - \frac{1}{2}\left(\frac{2q}{2\lambda q + \bar{\lambda}}+\frac{2\bar{q}}{2\lambda\bar{q}+\bar{\lambda}}\right) \\
		&= \frac{\lambda\bar{\lambda}(1-2q)^2}{1-\lambda^2(1-2q)^2} 
	\end{align*}
	and
	\begin{align*}
		\LC_\lambda(\nu K \| \mu K) &= \LC_\lambda(\Bernoulli(q\star\eps)\|\Bernoulli(1/2)) \\
		&= \frac{\lambda\bar{\lambda}(1-2(q\star\eps))^2}{1-\lambda^2(1-2(q\star\eps))^2} \\
		&= \frac{\lambda\bar{\lambda}(1-2\eps)^2(1-2q)^2}{1-\lambda^2(1-2\eps)^2(1-2q)^2}.
	\end{align*}
	Both of these divergences are invariant with respect to the transformation $q \mapsto 1-q$, so 
	\begin{align}\label{eq:BSC_LC_ratio}
		\eta_{\LC_\lambda}(\mu,K) &= \sup_{0 \le q < 1/2} \frac{(1-2\eps)^2(1-\lambda^2(1-2q)^2)}{1-(1-2\eps)^2\lambda^2(1-2q)^2} =(1-2\eps)^2,
	\end{align}
	where the supremum is achieved at $q=1/2$ (but not at any $q \neq 1/2$). (As an aside, it is not hard to show that the expression under the supremum in \eqref{eq:BSC_LC_ratio} is a concave function of $q$.) Comparing Eqs.~\eqref{eq:BSC_chi_2_ratio} and \eqref{eq:BSC_LC_ratio}, we see that
	\begin{align*}
		\eta_{\chi^2}(\mu,K) = \sup_{0 < \lambda < 1} \eta_{\LC_\lambda}(\mu,K) = (1-2\eps)^2.
	\end{align*}
	Therefore, by Theorem~\ref{thm:opconv_bounds},
	\begin{align*}
		\eta_\Phi(\Bernoulli(1/2),\BSC(\eps)) = (1-2\eps)^2
	\end{align*}
	for all operator convex $\Phi \in \cF$.
}
\end{example}

\subsection{Upper bounds via subgaussian concentration and information-transportation inequalities}
\label{ssec:info_transport_bounds}

Fix an admissible pair $(\mu,K) \in \PProb(\sX) \times \Chan(\sY|\sX)$, and let $(X,Y)$ be a random pair with probability law $\mu \otimes K$. We expect the SDPI constant $\eta(\mu,K)$ to be small if the channel output $Y$ of $K$ is nearly independent of the channel input $X \sim \mu$. In this section, we present upper bounds on $\eta(\mu,K)$ that capture this intuition in terms of the properties of the posterior likelihood ratio
\begin{align}\label{eq:PLR}
	a(x,y) \deq \frac{\d P_{X|Y=y}}{\d P_X}(x) = \frac{K^*(x|y)}{\mu(x)}.
\end{align}
Theorems \ref{thm:subgaussian} and \ref{thm:info_transport} quantify near-independence by looking at how tightly the random variable $a(X,y)$ concentrates around its expected value $1$ for each fixed $y$. Moreover, Theorem~\ref{thm:info_transport} shows a connection between SDPI for the relative entropy and \textit{information-transportation inequalities} introduced in the pioneering work of Marton \cite{Marton_blowup,Marton_dbar}.

First, we collect some preliminaries. A real-valued random variable $U$ is called \textit{subgaussian with parameter $v$} (or $v$-subgaussian) if $\E[e^{t(U-\E U)}] \le e^{vt^2/2}$ for all $t \in \Reals$ \cite[Sec.~2.3]{Boucheron_etal_concentration_book}. For any $v$-subgaussian random variable $U$ we have the tail estimate
$$
\PP(|U-\E U| \ge t) \le 2e^{-t^2/2v}, \qquad \forall t \in \Reals.
$$
To get the tightest such bound, we define the \textit{subgaussian constant}
\begin{align*}
	\sigma^2(U) \deq \inf \left\{ v \ge 0 : \E[e^{t(U-\E U)}] \le e^{vt^2/2}, t \in \Reals\right\}.
\end{align*}
With these definitions in place, we have the following theorem:
\begin{theorem}\label{thm:subgaussian} For each $y \in \sY$, let $\sigma^2(y) \deq \sigma^2\left(a(X,y)\right)$. Then
\begin{align}\label{eq:subgaussian_bound}
	\eta(\mu,K) \le 2\, \E[\sigma^2(Y)].
\end{align}
\end{theorem}
\begin{proof} Fix any $\nu \in \Prob(\sX)$ and let $f = {\d\nu}/{\d\mu}$. Observe that $\E[a(X,y)]=\E[f(X)]=1$. Then
	\begin{align}
		D(\nu K \| \mu K) &=
		\Ent[K^*f(Y)] \nonumber\\
		&\le \Var[K^*f(Y)] \nonumber\\
		&= \sum_{y \in \sY} \mu K(y) \left( K^*f(y) - 1\right)^2 \nonumber\\
		&= \sum_{y \in \sY} \mu K(y) \left|\sum_{x \in \sX} \mu(x)\left[a(x,y)f(x)-1\right]\right|^2 \nonumber\\
		&= \sum_{y \in \sY} \mu K(y) \left| \Cov\left( a(X,y), f(X)\right)\right|^2, \label{eq:cov_est_0}
	\end{align}
	where the inequality is by Lemma~\ref{lm:f_entropy_UB} in Appendix~\ref{app:lemmas}. Next, we make use of the fact that
	\begin{align}\label{eq:Donsker_Varadhan}
		\Ent[U] &\ge \E[UZ] - \E[U] \log \E[e^Z]
	\end{align}
for any random variable $Z$ jointly distributed with $U$ and satisfying $\E[e^Z] < \infty$ (see, e.g., \cite[Thm.~4.13]{Boucheron_etal_concentration_book}; in fact, this bound holds with equality for $Z = \log U$). If we fix an arbitrary $y \in \sY$ and then use \eqref{eq:Donsker_Varadhan}  with $U = f(X)$ and $Z = \pm t\left(a(X,y)-1\right)$ for some $t > 0$, we get
\begin{align*}
	\left|\Cov\left(a(X,y),f(X)\right)\right| 
	&\le \frac{1}{t} \left(\log \E[e^{t\left(a(X,y)-1\right)}] + \Ent[f(X)]\right)  \\
	&\le \frac{\sigma^2(y)t}{2} + \frac{\Ent[f(X)]}{t}.
\end{align*}
Since this holds for an arbitrary $t$, we have
\begin{align*}
	\left|\Cov\left(a(X,y),f(X)\right)\right| \le \inf_{t > 0} \left\{ \frac{\sigma^2(y)t}{2} + \frac{\Ent[f(X)]}{t} \right\} = \sqrt{2 \sigma^2(y) \Ent[f(X)]}.
\end{align*}
Using this estimate in \eqref{eq:cov_est_0}, we get \eqref{eq:subgaussian_bound}.
\end{proof}

In order to apply Theorem~\ref{thm:subgaussian}, we need to compute or upper-bound the subgaussian constant $\sigma^2(a(X,y))$ for each $y \in \sY$. In some situtations, it is possible to derive exact expressions for subgaussian constants (as we show in the examples below); when the function $x \mapsto a(x,y)$ is Lipschitz for each $y \in \sY$, one can derive upper bounds using \textit{information-transportation inequalities} introduced in the pioneering work of Marton \cite{Marton_blowup,Marton_dbar} (see, e.g., the text of Villani \cite{Villani_topics}). If we endow the input alphabet $\sX$ with a metric $d$, then we can define the {\em $L^1$ Wasserstein distance} (or \textit{optimal transportation distance}) on $\Prob(\sX)$ by
\begin{align*}
	W_{1}(\mu,\nu) \deq \inf \left\{ \E[d(X,\bar{X})] : P_{X\bar{X}} \in \Prob(\sX \times \sX), P_X = \mu, P_{\bar{X}} = \nu\right\}
\end{align*}
For example, for the trivial metric $d(x,x') = \1\{x \neq x'\}$ we recover the total variation distance: $W_1(\mu,\nu) = \| \mu - \nu \|_\TV$. Given a function $f : \sX \to \Reals$, denote by
\begin{align*}
	\delta(f) \deq \sup_{x,x' \in \sX \atop x \neq x'} \frac{|f(x)-f(x')|}{d(x,x')}
\end{align*}
the \textit{oscillation} (or the \textit{Lipschitz norm}) of $f$ w.r.t.\ the metric $d$.

\begin{theorem}\label{thm:info_transport} Fix an admissible pair $(\mu,K) \in \PProb(\sX) \times \Chan(\sY|\sX)$. Suppose that $\mu$ satisfies an {\em information-transportation inequality} with constant $c > 0$, i.e.,
	\begin{align}\label{eq:info_transport}
		W_1(\nu,\mu) \le \sqrt{2 c\, D(\nu \| \mu)}, \qquad \forall \nu \neq \mu.
	\end{align}
Then
\begin{align}\label{eq:info_transport_bound}
	\eta(\mu,K) \le 2 c\, \E \left[ \delta^2\left( a(\cdot,Y)\right)\right].
\end{align}
\end{theorem}

\begin{proof} By a result of Bobkov and G\"otze \cite{Bobkov_Goetze}, a probability measure $\mu \in \cP(\sX)$ satisfies \eqref{eq:info_transport} if and only if
\begin{align}\label{eq:BG_subgauss}
	\E_\mu\left[ e^{t \left(f(X) - \E f(X)\right)}\right] \le e^{\frac{ct^2}{2}}, \qquad t \in \Reals
\end{align}
for every $f \in \Func(\sX)$ with $\delta(f) \le 1$. In particular, if \eqref{eq:info_transport} holds, then, by rescaling \eqref{eq:BG_subgauss}, we get
\begin{align*}
	\E_\mu\left[ e^{t \left(a(X,y) - 1\right)}\right] \le \exp\left(\frac{c \delta^2\left(a(\cdot,y)\right)t^2}{2}\right).
\end{align*}
This implies that $\sigma^2\big(a(X,y)\big) \le c\delta^2\left(a(\cdot,y)\right)$. Substituting this into \eqref{eq:subgaussian_bound}, we get \eqref{eq:info_transport_bound}.
\end{proof}

\begin{example}[Binary symmetric channels with asymmetric inputs] {\em Let $\sX = \sY = \{0,1\}$, $\mu = \Bernoulli(p)$, $K = \BSC(\eps)$. We take the trivial metric $d(x,x') = \1\{x \neq x'\}$. In this case, Theorems \ref{thm:subgaussian} and \ref{thm:info_transport}  give the same bound. Indeed, by a result of Ordentlich and Weinberger \cite{Ordentlich_Weinberger_Pinsker}, $\mu = \Bernoulli(p)$ satisfies an information-transportation inequality
	\begin{align}\label{eq:DD_Pinsker}
		\| \nu - \mu \|_\TV \le \sqrt{2c(p)\, D(\nu \| \mu)}, \qquad c(p) \deq \frac{p-\bar{p}}{2(\log p - \log \bar{p})},
	\end{align}
and the constant in front of the relative entropy is optimal, i.e., 
\begin{align*}
 \inf_\nu \frac{D(\nu \| \mu)}{\| \nu - \mu \|^2_\TV} = \frac{1}{2c(p)}.
\end{align*}
[The inequality \eqref{eq:DD_Pinsker} is a distribution-dependent refinement of Pinsker's inequality, where we fix $\mu$ and vary only $\nu$.] A simple calculation gives
\begin{align*}
	\delta\left(a(\cdot,0)\right) & = \left|\frac{1-2\eps}{1-\eps \star p}\right|, \qquad \delta\left(a(\cdot,1)\right) = \left|\frac{1-2\eps}{\eps \star p}\right|,
\end{align*}
where $\eps \star p = \eps \bar{p} + \bar{\eps} p$. Therefore, applying Theorem~\ref{thm:info_transport}, we get the bound
\begin{align}\label{eq:BSC_info_transport_bound}
	\eta\left(\Bernoulli(p),\BSC(\eps)\right) \le  \frac{2c(p)(1-2\eps)^2}{(1-\eps \star p)(\eps \star p)},
\end{align}
This bound is, unfortunately, loose. Indeed, if we take the limit $p \searrow 1/2$, then we get
\begin{align}\label{eq:FC_BSC_info_transport_bound}
\eta\left(\Bernoulli(1/2),\BSC(\eps)\right) \le 2(1-2\eps)^2.
\end{align}
which is off by a factor of $2$, but still tighter than the Dobrushin contraction bound $|1-2\eps|$  (Theorem~\ref{thm:Markov_contraction_bound}) in the range $1/4 < \eps < 3/4$. Figure~\ref{fig:BSC_max_eta} shows a plot of the maximum value of the right-hand side of \eqref{eq:BSC_info_transport_bound} over $p$ for each fixed value of the crossover probability $\eps$; from this, we see that the bound is nontrivial (i.e., takes values strictly smaller than $1$) for $\eps \gtrsim 0.156$.

\begin{figure}[htb]
	\centerline{\includegraphics[width=0.5\textwidth]{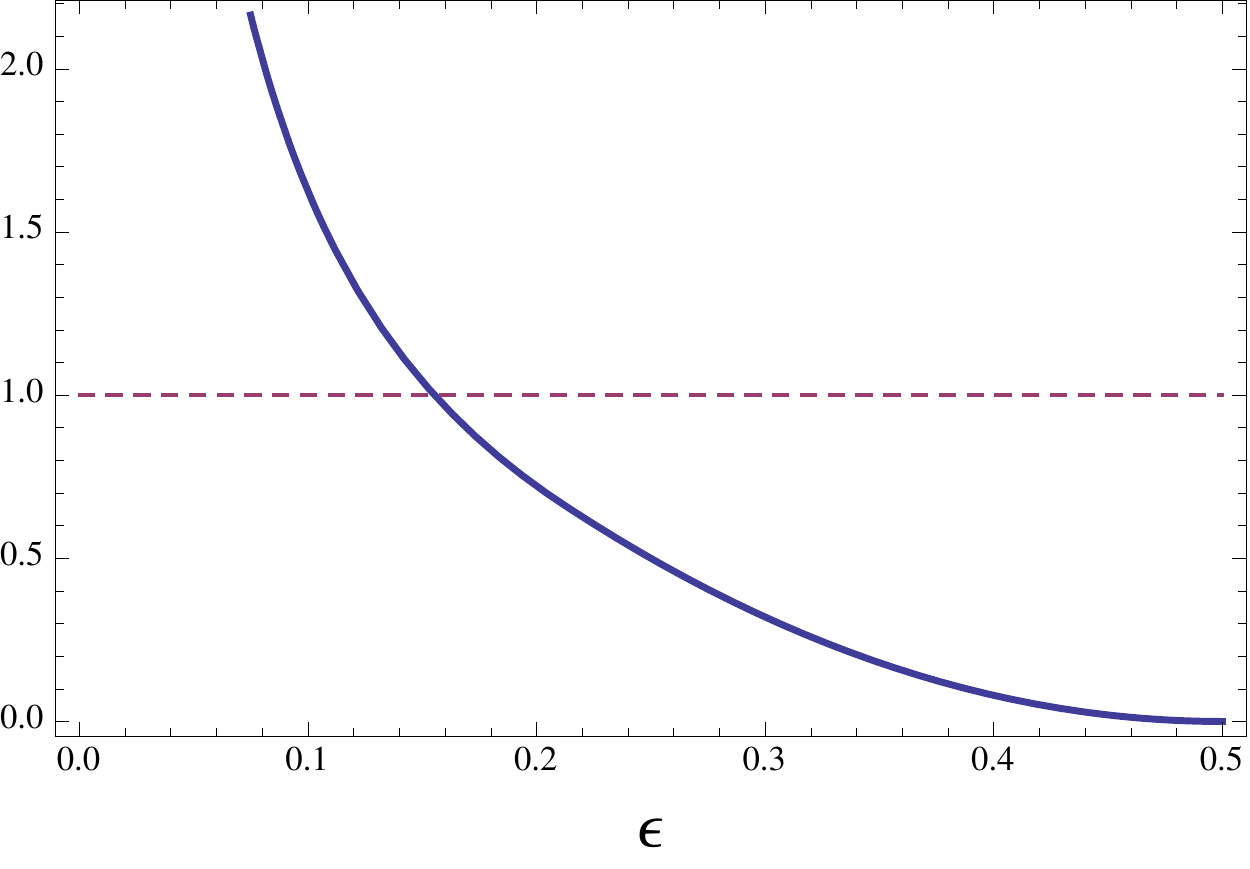}}
	\caption{\label{fig:BSC_max_eta} Maximum value of the right-hand side of \eqref{eq:BSC_info_transport_bound} over $p \in [0,1]$ for each fixed $\eps$.}
\end{figure}

In order to apply Theorem~\ref{thm:subgaussian}, we need to know the subgaussian constants of $a(X,y)$, $y \in \{0,1\}$. By a result of Bobkov et al.~\cite{Bobkov_Houdre_Tetali_subgaussian}, for any function $f : \{0,1\} \to \Reals$ and for $X \sim \Bernoulli(p)$ we have
\begin{align}\label{eq:BHT}
	2\sigma^2\left(f(X)\right) &= 2c(p) \left| f(0) - f(1)\right|^2.
\end{align}
Applying \eqref{eq:BHT} to $f=a(\cdot,0)$ and $a(\cdot,1)$, we get
\begin{align*}
	2\sigma^2(0) = 2c(p) \left|\frac{1-2\eps}{1-\eps \star p}\right|^2,  \qquad
	2\sigma^2(1) = 2c(p) \left|\frac{1-2\eps}{\eps \star p}\right|^2,
\end{align*}
and indeed Theorem~\ref{thm:subgaussian} gives the same bound \eqref{eq:BSC_info_transport_bound}.
}
\end{example}

\begin{example}[Binary input channels]{\em Let $\sX = \{0,1\}$ with $\mu = \Bernoulli(p)$, and consider an arbitrary channel $K \in \Chan(\sY|\sX)$ with a finite (not necessarily binary) output alphabet $\sY$. Then
	\begin{align*}
		a(x,y) = \frac{K^*(x|y)}{\mu(x)} = \frac{K(y|x)}{\mu K(y)},
	\end{align*}
where $\mu K(y) = \bar{p} K(y|0) + p K(y|1)$. If we again take $d$ to be the trivial metric, then the same analysis as in the previous example can be used to show that
\begin{align*}
	2\sigma^2\left(y\right) &= 2c(p) \frac{|K(y|0) - K(y|1)|^2}{\mu K(y)^2},
\end{align*}
and Theorem~\ref{thm:subgaussian} gives the bound
\begin{align*}
	\eta\left(\Bernoulli(p),K\right) &\le 2c(p) \sum_{y \in \sY} \frac{|K(y|0) - K(y|1)|^2}{\bar{p} K(y|0) + p K(y|1)}.
\end{align*}
}
\end{example}

\begin{example}[Random walk on a graph] {\em Consider a connected undirected graph $G = (\sV,\sE)$ without self-loops or multiple edges, and let $\sX = \sY = \sV$. If the vertices $x$ and $y$ are connected by an edge, we shall write $x \leftrightarrow y$; the degree of a vertex $x$ is defined as $\deg_G(x) \deq |\left\{ y \in \sV: x \leftrightarrow y \right\}|$. Define a probability measure $\mu = \mu_G \in \cP(\sV)$ by
	\begin{align*}
		\mu_G(x) \deq \frac{\deg_G(x)}{2|\sE|}, \qquad x \in \sV.
	\end{align*}
Fix a parameter $\eps \in (0,1)$, and consider a channel $K^{(\eps)}_G$ with
	\begin{align}\label{eq:GRW}
		K^{(\eps)}_G(y|x) &= \begin{cases}
		\bar{\eps}, & \text{if } x=y \\
		\dfrac{\eps}{\deg_G(x)}, & \text{if } x \leftrightarrow y \\
		0, & \text{otherwise}
	\end{cases}.
	\end{align}
Again, let $d$ be the trivial metric, $d(x,x') = \1\{x\neq x'\}$. Then $W_1(\nu,\mu) = \| \nu - \mu \|_\TV$, and we can take $c = 1/4$ in \eqref{eq:info_transport}, which is then just Pinsker's inequality. It is not hard to show that $K^{(\eps)}_G$ is \textit{reversible} w.r.t.\ $\mu_G$, i.e.,
\begin{align*}
	\mu_G(x)K^{(\eps)}_G(y|x) = \mu_G(y)K^{(\eps)}_G(x|y), \qquad \forall x,y \in \sV.
\end{align*}
Therefore, $\mu_G K^{(\eps)}_G = \mu_G$, so the posterior likelihood ratio is given by
\begin{align*}
	a(x,y) = \frac{K^{(\eps)}_G(y|x)}{\mu_G(y)} = \frac{2|\sE|}{\deg_G(y)} K^{(\eps)}_G(y|x).
\end{align*}
Now, from the definition \eqref{eq:GRW} of $K^{(\eps)}_G$ it follows that
\begin{align*}
	\left|K^{(\eps)}_G(y|x) - K^{(\eps)}_G(y|x')\right| &= \begin{cases}
	\left|\bar{\eps} - \frac{\eps}{\deg_G(x')}\1\{x' \leftrightarrow y\}\right|, & \text{if $x=y$} \\
	\left|\bar{\eps} - \frac{\eps}{\deg_G(x)}\1\{x \leftrightarrow y\}\right|, & \text{if $x'=y$} \\
	\frac{\eps}{\deg_G(x)}, & \text{if $x \leftrightarrow y,\, x' \not\leftrightarrow y$} \\
	\frac{\eps}{\deg_G(x')}, & \text{if $x \not\leftrightarrow y,\, x' \leftrightarrow y$} \\
	\eps\left| \frac{1}{\deg_G(x)} - \frac{1}{\deg_G(x')}\right|, &\text{if $x \leftrightarrow y,\, x' \leftrightarrow y$} \\
	0, &\text{if $x \not\leftrightarrow y,\, x' \not\leftrightarrow y$}
\end{cases}
\end{align*}
where $x \not\leftrightarrow y$ means that $x$ and $y$ are not connected by an edge and that $x \neq y$. Therefore,
\begin{align*}
	\delta^2\left(a(\cdot,y)\right) &= \frac{4|\sE|^2}{\deg_G(y)}\max_{x,x' \in \sV} \left| K^{(\eps)}_G(y|x) - K^{(\eps)}_G(y|x')\right|^2 \\
	&= \frac{4|\sE|^2}{\deg_G(y)^2} \Big(\Delta_0(y,\eps) \vee \Delta_1(y,\eps) \vee \Delta_2(y,\eps)\Big),
\end{align*}
where
\begin{subequations}\label{eq:Deltas}
	\begin{align}
	\Delta_0(y,\eps) &\deq \max_{x \in \sV\backslash\{y\}} \left(\frac{\eps}{\deg_G(x)}\right)^2 \1\{\deg_G(y) < |\sV| - 1 \} \\
	\Delta_1(y,\eps) &\deq \max_{x \in \sV\backslash\{y\}} \left| \bar{\eps} - \frac{\eps}{\deg_G(x)}\1\{x \leftrightarrow y\}\right|^2 \\
	\Delta_2(y,\eps) &\deq \max_{x,x' \in \sV\backslash\{y\}} \eps^2\left|\frac{1}{\deg_G(x)} - \frac{1}{\deg_G(x')}\right|^2 \1\{x \leftrightarrow y,\, x' \leftrightarrow y\}.
\end{align}
\end{subequations}
Theorem~\ref{thm:info_transport} then gives the bound
\begin{align}\label{eq:GRW_eta_bound}
	\eta\left(\mu_G,K^{(\eps)}_G\right) \le |\sE| \sum_{y \in \sV} \frac{\Delta_0(y,\eps) \vee \Delta_1(y,\eps) \vee \Delta_2(y,\eps)}{\deg_G(y)}
\end{align}
(note that $\deg_G(y) > 0$ for each $y$, since $G$ is connected).

For example, if $G$ is a complete graph, then $\Delta_0(y,\eps) = \Delta_2(y,\eps) = 0$ for all $y$, while
$$
\Delta_1(y,\eps) = \left(1 - \frac{|\sV|}{|\sV|-1}\eps\right)^2,  \qquad y \in \sV
$$
so we get the bound
\begin{align}\label{eq:complete_graph}
	\eta\left(\mu_G,K^{(\eps)}_G\right) \le \frac{|\sV|^2}{2} \left(1 - \frac{|\sV|}{|\sV|-1}\eps\right)^2,
\end{align}
which is nontrivial (i.e., strictly smaller than unity) in the range
\begin{align*}
	\frac{\left(|\sV|-1\right)\left(|\sV|-\sqrt{2}\right)}{|\sV|^2} < \eps <\frac{\left(|\sV|-1\right)\left(|\sV|+\sqrt{2}\right)}{|\sV|^2}.
\end{align*}
For the complete graph on the two-point set $\sV = \{0,1\}$, the channel $K^{(\eps)}_G$ is just $\BSC(\eps)$, and the bound \eqref{eq:complete_graph} reduces to \eqref{eq:FC_BSC_info_transport_bound}.

As another example, let $G$ be the path graph on the ternary vertex set $\sV =\{0,1,2\}$, i.e., $\sE = \big\{ \{0,1\}, \{1,2\} \big\}$. Then $\mu_G(0) = \mu_G(2) = 1/4$ and $\mu_G(1) = 1/2$. From \eqref{eq:Deltas}, we get
\begin{align*}
	\Delta_0(y,\eps) &= \begin{cases}
	\eps^2, & y \in \{0,2\} \\
	0, & y = 1
\end{cases} \\
\Delta_1(y,\eps) &= \begin{cases}
\left(1-\frac{3\eps}{2}\right)^2, & y \in \{0,2\} \\
(1-2\eps)^2, & y = 1
\end{cases} \\
\Delta_2(y,\eps) &= 0, \qquad y \in \{0,1,2\}.
\end{align*}
Substituting this into \eqref{eq:GRW_eta_bound}, we get
\begin{align}\label{eq:2path_eta}
	\eta\left(\mu_G,K^{(\eps)}_G\right) \le \begin{cases}
	13 \eps^2 - 16 \eps + 5, & 0 \le \eps \le 0.4 \\
	8 \eps^2 - 4 \eps + 1, & 0.4 \le \eps \le 1
\end{cases}.
\end{align}
\begin{figure}[htb]
	\centerline{\includegraphics[width=0.5\textwidth]{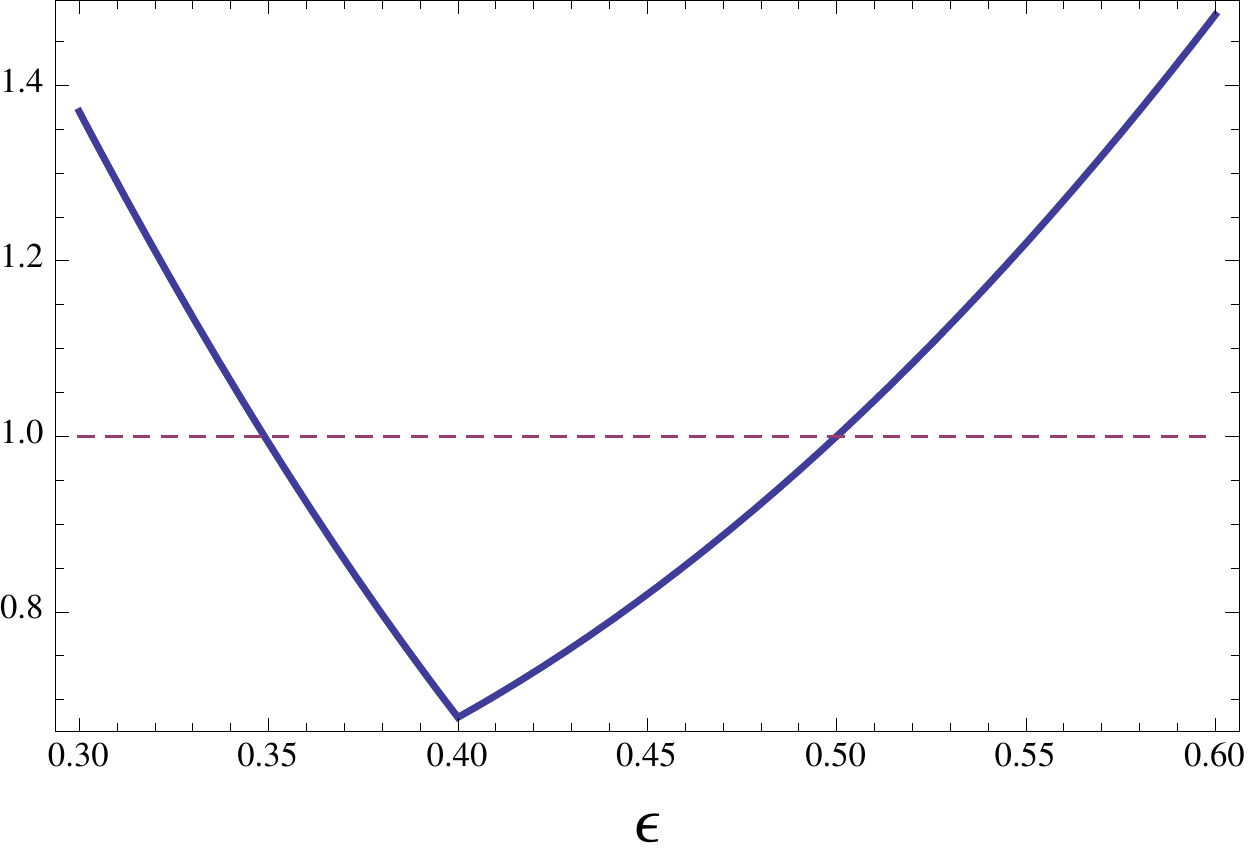}}
	\caption{\label{fig:2path} The bound of Eq.~\eqref{eq:2path_eta} as a function of the noise parameter $\eps$.}
\end{figure}
This bound, plotted in Figure~\ref{fig:2path}, is nontrivial only in the range $\frac{8-2\sqrt{3}}{13} < \eps < \frac{1}{2}$.
}
\end{example}

\begin{example}[General discrete channel] {\em Consider arbitrary finite alphabets $\sX$ and $\sY$, together with an admissible pair $(\mu,K) \in \PProb(\sX) \times \Chan(\sY|\sX)$. If we endow $\sX$ with the trivial metric $d(x,x') = \1\{x \neq x'\}$, then $\mu$ will satisfy the information-transportation inequality \eqref{eq:info_transport} for $W_1(\nu,\mu) = \| \nu - \mu \|_\TV$ with optimal ($\mu$-dependent) constant $c(\beta_\mu)$, where the function $c(\cdot)$ is defined in \eqref{eq:DD_Pinsker}, and
	\begin{align*}
		\beta_\mu \deq \min \left\{ \mu(\sA) : \sA \subseteq \sX,\, \mu(\sA) \ge 1/2 \right\}
	\end{align*}
is a measure of ``imbalance'' of $\mu$ --- in particular, when $\mu$ is the uniform distribution on $\sX$ and $|\sX|$ is even, $\beta_\mu = 1/2$. Again, this is just the distribution-dependent refinement of Pinsker's inequality \cite{Ordentlich_Weinberger_Pinsker}. Then
\begin{align*}
	\delta^2\left(a(\cdot,y)\right) &= \max_{x,x' \in \sX} \left| \frac{K^*(x|y)}{\mu(x)} - \frac{K^*(x'|y)}{\mu(x')}\right|^2 \\
	&= \frac{1}{\mu K(y)^2}\max_{x,x' \in \sX} \left| K(y|x) - K(y|x')\right|^2 \\
	&= \frac{1}{\mu K(y)^2} \delta^2\left( K(y|\cdot)\right),
\end{align*}
so Theorem~\ref{thm:info_transport} gives the bound
\begin{align}\label{eq:eta_general_channel}
	\eta(\mu,K) \le 2c(\beta_\mu) \sum_{y \in \sY} \frac{\delta^2\left(K(y|\cdot)\right)}{\mu K(y)}.
\end{align}
}

\end{example}

In general, the bounds of Theorems \ref{thm:subgaussian} and \ref{thm:info_transport} are nontrivial only for channels that are ``sufficiently noisy,'' in the sense that the posterior likelihood ratio \eqref{eq:PLR} is nearly constant as a function of the input symbol $x$ for any fixed output symbol $y$. In particular, the function $x \mapsto a(x,y)$ is constant for each $y \in \sY$ if and only if the output of $K$ is independent of the input, i.e., if $\eta(\mu,K) = 0$. However, these bounds may be useful for capturing the \textit{scaling} of the SDPI constant $\eta(\mu,K)$ with various parameters of the problem. To the best of our knowledge, the first bound on $\eta(\mu,K)$ in terms of a certain concentration property of the posterior likelihood ratio is due to Weitz \cite{Weitz_thesis} (see also \cite{Martinelli_etal_Glauber_trees}), and can be stated in our notation as follows: 
\begin{align}\label{eq:Weitz_bound}
	\eta(\mu,K) \le \left(\frac{c}{(\mu K)_*}\right)^2 \E[\tau(Y)],
\end{align}
where $c > 0$ is some numerical constant, $(\mu K)_* = \min_{y \in \sY}\mu K(y)$, and
\begin{align*}
	\tau(y) \deq \inf \left\{ t \ge 0 : \PP\left(\left|a(X,y)-1\right| > t\right) \le e^{-2/t}\right\}.
\end{align*}
Since the function $t \mapsto e^{-2/t}$ is increasing, converges to $1$ as $t \to \infty$, and to $0$ as $t \searrow 0$, the quantity $\tau(y)$ should be very close to zero for the bound \eqref{eq:Weitz_bound} to be nontrivial. In contrast to the bounds of Theorems~\ref{thm:subgaussian} and \ref{thm:info_transport}, which involve quantities pertaining to \textit{large deviations} of $a(X,y)$ from its mean, Weitz's bound is in terms of a quantity that has to do with \textit{small deviations} of $a(X,y)$ from its mean.

\subsection{Tensorization}

So far, we have considered the case of a single channel. However, many problems in information theory involve multiple uses of the same channel (or, more generally, transmission of correlated data over a memoryless channel with time-varying transition probabilities). In this context, it is of interest to determine whether the resulting ``super-channel'' inherits any SDPI-type behavior from the constituent channels. 

In precise terms, let $(\mu_1,K_1),\ldots,(\mu_n,K_n)$ be $n$ admissible pairs, where, for each $i$, $\mu_i \in \PProb(\sX_i)$ and $K_i \in \Chan(\sY_i|\sX_i)$ for some alphabets $\sX_i,\sY_i$. Fix some $\Phi \in \cF$, a product distribution $\mu = \mu_1 \otimes \ldots \otimes \mu_n \in \PProb(\sX_1 \times \ldots \times \sX_n)$, and a product channel $K = K_1 \otimes \ldots \otimes K_n \in \Chan(\sY_1 \times \ldots \times \sY_n|\sX_1 \times \ldots \times \sX_n)$. We say that the SDPI constant $\eta_\Phi(\mu,K)$ \textit{tensorizes} if
\begin{align*}
	\eta_\Phi(\mu,K) = \max_{1 \le i \le n} \eta_\Phi(\mu_i,K_i).
\end{align*}
For instance, Witsenhausen \cite{Witsenhausen_correlation} showed that $\eta_{\chi^2}(\mu,K)$ tensorizes, while a recent paper by Anantharam et al.~\cite{Anantharam_etal_HGR} presents two different proofs of the tensorization property of $\eta(\mu,K)$. In each case, the proof relies on specific properties of the underlying $\Phi$ --- Witsenhausen exploits the connection between $\eta_{\chi^2}(\mu,K)$ and the eigenvalues of the linear operator $KK^* : L^2(\sX,\mu) \to L^2(\sX,\mu)$, whereas Anantharam et al.\ use the chain rule for the relative entropy. The question is, can one give a unified proof of tensorization for a broader class of functions $\Phi \in \cF$ that contains both $\Phi(u) = (u-1)^2$ and $\Phi(u) = u \log u$? As we show next, the answer is `yes' for all functions $\Phi$ whose $\Phi$-entropies are subadditive and homogeneous in the sense of Definition~\ref{def:gen_hom}.

\begin{theorem}[Tensorization]\label{thm:tensorization_1} Suppose that $\Phi \in \cF$ induces a subadditive and homogeneous $\Phi$-entropy.  Consider any $n$ admissible pairs $(\mu_i,K_i) \in \PProb(\sX_i) \times \Chan(\sY_i|\sX_i)$. Then
	\begin{align*}
		\eta_\Phi(\mu_1 \otimes \ldots \otimes  \mu_n, K_1 \otimes \ldots \otimes K_n) = \max_{1 \le i \le n} \eta_\Phi(\mu_i,K_i).
	\end{align*}
\end{theorem}
\begin{proof} For the sake of brevity, let $\eta = \eta_\Phi(\mu_1 \otimes \ldots \otimes  \mu_n, K_1 \otimes \ldots \otimes K_n)$, $\eta_i = \eta_\Phi(\mu_i,K_i)$, $\mu = \mu_1 \otimes \ldots \otimes \mu_n$, and $K = K_1 \otimes \ldots \otimes K_n$.
	
	To show that $\eta \ge \eta_i$ for all $i$, take $\nu \in \Prob(\sX_1 \times \ldots \times \sX_n)$ of the form $\mu_1 \otimes \ldots \otimes \mu_{i-1} \otimes \nu_i \otimes \mu_{i+1} \otimes \ldots \otimes \mu_n$ for some $\nu_i \in \Prob(\sX_i) \backslash \{\mu_i\}$. Then
	\begin{align*}
		& D_\Phi(\nu \| \mu) = D_\Phi(\nu_i \| \mu_i), \\
		& D_\Phi(\nu K \| \mu K) = D_\Phi(\nu_iK_i \| \mu_iK_i).
	\end{align*}
Taking the supremum of $\frac{D_\Phi(\nu K \| \mu K)}{D_\Phi(\nu \| \mu)}$ over all such $\nu$, we conclude that $\eta \ge \eta_i$.

For the reverse inequality $\eta \le \max_{1 \le i \le n} \eta_i$, it suffices to consider the case $n=2$; the general case will follow by induction. Thus, let us fix two admissible pairs $(\nu_i,K_i) \in \PProb(\sX_i) \times \Chan(\sY_i|\sX_i)$, $i=1,2$, and an arbitrary nonconstant function $f \in \PFunc(\sX_1 \times \sX_2)$. Let $(X_1,X_2,Y_1, Y_2) \in \sX_1 \times \sX_2 \times \sY_1 \times \sY_2$ be a random tuple, such that
\begin{align*}
	P_{X_1X_2} = \mu_1 \otimes \mu_2, \qquad P_{Y_1Y_2|X_1X_2} = K_1 \otimes K_2.
\end{align*}
Then, from \eqref{eq:total_entropy},
\begin{align*}
	\Ent_\Phi\big[K^*f(Y_1,Y_2)\big] = \E\left[ \Ent_\Phi\big[K^*f(Y_1,Y_2)\big|Y_1\big]\right] + \Ent_\Phi\big[\E[K^*f(Y_1,Y_2)|Y_1]\big].
\end{align*}
Define the functions $f_1 \in \PFunc(\sY_1 \times \sX_2)$ and $f_2 \in \PFunc(\sX_1 \times \sY_2)$ by
\begin{align*}
	f_1(y_1,x_2) &= \sum_{x_1 \in \sX_1} P_{X_1|Y_1}(x_1|y_1)f(x_1,x_2) \\
	f_2(x_1,y_2) &= \sum_{x_2 \in \sX_2} P_{X_2|Y_2}(x_2|y_2)f(x_1,x_2),
\end{align*}
which can be written more succinctly as $f_1 = (K^*_1 \otimes \id_2) f$ and $f_2 = (\id_1 \otimes K^*_2)f$, where $\id_1$ and $\id_2$ are the identity mappings on $\Func(\sX_1)$ and $\Func(\sX_2)$.  Since $Y_1$ and $Y_2$ are independent, we can write
\begin{align*}
	\Ent_\Phi\big[K^*f(Y_1,Y_2)\big|Y_1=y_1\big] & = \E[\Phi(K^*_2 f_1(y_1,Y_2))] - \Phi(\E[K^*_2 f_1(y_1,Y_2)]) \\
& = \Ent_\Phi\big[K^*_2 f_1(y_1,Y_2)\big] \\
& \le \eta_2 \Ent_\Phi\big[f_1(y_1,X_2)\big] \\
& \le \eta_2 \sum_{x_1 \in \sX_1} P_{X_1|Y_1}(x_1|y_1) \Ent_\Phi\big[f(x_1,X_2)\big],
\end{align*}
where the first inequality uses \eqref{eq:functional_SDPI_gen_hom}, while the second inequality follows from the definition of $f_1$ and from the convexity property \eqref{eq:phi_convexity}, which is equivalent to the assumed subadditivity of $\Ent_\Phi[\cdot]$. Therefore,
\begin{align*}
	\E\left[ \Ent_\Phi\big[K^*f(Y_1,Y_2)\big|Y_1\big]\right] &= \sum_{y_1 \in \sY_1} P_{Y_1}(y_1) \Ent_\Phi\big[K^*f(Y_1,Y_2)\big|Y_1=y_1\big] \\
	&\le \eta_2 \sum_{y_1 \in \sY_1}P_{Y_1}(y_1)\sum_{x_1 \in \sX_1}P_{X_1|Y_1}(x_1|y_1) \Ent_\Phi\big[f(x_1,X_2)\big] \\
	&= \eta_2 \sum_{x_1 \in \sX_1} P_{X_1}(x_1) \Ent_\Phi\big[f(x_1,X_2)\big] \\
	&= \eta_2\, \E\left[\Ent_\Phi\big[f(X_1,X_2)\big|X_1\big]\right].
\end{align*}
Next, let $g_2(x_1) = \E[f_2(x_1,Y_2)] = \E[f_2(X_1,Y_2)|X_1=x_1]$. Then
	\begin{align*}
		\Ent_\Phi\left[\E[K^*f(Y_1,Y_2)|Y_1]\right] 
	&= \Ent_\Phi \left[ K^*_1 g_2(Y_1)\right] \\
	&\le \eta_1 \Ent_\Phi\left[ g_2(X_1)\right] \\
	&= \eta_1 \Ent_\Phi\left[ \E[f(X_1,X_2)|X_1]\right],
	\end{align*}
	where the first line follows from the fact that $(X_1,Y_1)$ and $(X_2,Y_2)$ are independent and from definitions,
	whereas in the last line we have used the fact that
	\begin{align*}
		g_2(x_1) &= \E[f_2(x_1,Y_2)] \\
		&=\sum_{y_2 \in \sY_2} P_{Y_2}(y_2) f_2(x_1,y_2) \\
		&= \sum_{y_2 \in \sY_2} P_{Y_2}(y_2)\sum_{x_2 \in \sX_2} P_{X_2|Y_2}(x_2|y_2) f(x_1,x_2) \\
		&= \sum_{x_2 \in \sX_2} P_{X_2}(x_2) f(x_1,x_2) \\
		&= \E[f(X_1,X_2)|X_1=x_1].
	\end{align*}
Combining everything, we can write
\begin{align*}
		 \Ent_\Phi\big[K^*f(Y_1,Y_2)\big] & \le \eta_2 \cdot \E\left[ \Ent_\Phi\big[f(X_1,X_2)\big|X_1\big]\right]  + \eta_1 \cdot \Ent_\Phi \left[ \E[f(X_1,X_2)|X_1]\right] \\
		&  \le \max_{i=1,2} \eta_i \cdot \Bigg\{ \E\left[ \Ent_\Phi\big[f(X_1,X_2)\big|X_1\big]\right]  + \Ent_\Phi \left[ \E[f(X_1,X_2)|X_1]\right] \Bigg\} \\
		&  = \max_{i=1,2} \eta_i \cdot \Ent_\Phi\left[ f(X_1,X_2) \right],
\end{align*}
where in the last step we have used the law of total entropy \eqref{eq:total_entropy}. Since $f$ was arbitrary, we obtain the bound $\eta \le \max(\eta_1,\eta_2)$.
\end{proof}

\subsection{Mixtures of local channels}

Another situation that often arises in stochastic simulation and machine learning is as follows: Fix $n$ channels $K_i \in \Chan(\sX_i|\sX_i)$, $1 \le i \le n$, and a probability distribution $p = (p_i)^n_{i=1}$ on the set $\{1,\ldots,n\}$. Given an input block $x^n = (x_1,\ldots,x_n) \in \sX_1 \times \ldots \times \sX_n$, a random output block $Y^n = (Y_1,\ldots,Y_n) \in \sX_1 \times \ldots \times \sX_n$ is generated as follows: 
\begin{enumerate}
	\item a random index $J \in \{1,\ldots,n\}$ is drawn according to $p$;
	\item $Y_J$ is drawn according to $K_J(\cdot|x_J)$;
	\item $Y^{\backslash J} = x^{\backslash J}$.
\end{enumerate}
The overall stochastic transformation is described by the Markov kernel
\begin{align*}
	K \deq \sum^n_{i=1} p_i \left( \id_1 \otimes \ldots \otimes \id_{i-1} \otimes K_i \otimes \id_{i+1} \otimes \ldots \otimes \id_n\right),
\end{align*}
where, for each $i$, $\id_i$ is the idenitity mapping on $\Func(\sX_i)$. Now let us also fix $n$ probability distributions $\mu_i \in \Prob(\sX_i)$, $1 \le i \le n$. The question is: how does the SDPI constant $\eta_\Phi(\mu_1 \otimes \ldots \otimes, \mu_n, K)$ for some $\Phi \in \cF$ depend on $p$ and on the individual SDPI constants $\eta_\Phi(\mu_i,K_i)$?

\begin{theorem}\label{thm:tensorization_2} Under the same conditions as in Theorem~\ref{thm:tensorization_1},
	\begin{align}\label{eq:tensorization_2}
		1-\eta_\Phi(\mu_1 \otimes \ldots \otimes, \mu_n, K) \ge \min_{1 \le i \le n} p_i \left(1-\eta_\Phi(\mu_i,K_i)\right).
	\end{align}
\end{theorem}
\begin{proof} Once again, it suffices to consider the case $n=2$. Thus, we fix two admissible pairs $(\mu_i,K_i) \in \Prob(\sX_i) \times \Chan(\sX_i|\sX_i)$, $i \in \{1,2\}$ and a parameter $p \in [0,1]$, and consider the channel
	\begin{align*}
		K = p (K_1 \otimes \id_2) + \bar{p} (\id_1 \otimes K_2).
	\end{align*}
Let $\mu = \mu_1 \otimes \mu_2$ denote the reference input distribution. We need to show that
\begin{align}\label{eq:tensorization_2_goal}
	1-\eta_\Phi(\mu,K) \ge \min\Big(p(1-\eta_\Phi(\mu_1,K_1)), \bar{p}(1-\eta_\Phi(\mu_2,K_2))\Big).
\end{align}
As in the proof of Theorem~\ref{thm:tensorization_1}, we adopt the shorthand notation $\eta_i = \eta_\Phi(\mu_i,K_i)$ and
\begin{align*}
	\eta = \eta_\Phi(\mu,K) = \eta_\Phi\left(\mu_1 \otimes \mu_2, K\right).
\end{align*}
Let $(X_1,Y_1,X_2,Y_2)$ be a random tuple with $(X_1,X_2) \sim \mu$ and $P_{Y_1,Y_2|X_1,X_2} = K$. Also, define the Radon--Nikodym derivatives
\begin{align}
	g_1(y_1,y_2) & \deq \frac{  \d(\mu_1 K_1 \otimes \mu_2)}{  \d(\mu K)}(y_1,y_2) = \frac{\mu_1 K_1(y_1)\mu_2(y_2)}{p \mu_1 K_1(y_1)\mu_2(y_2) + \bar{p}\mu_1(y_1) \mu_2 K_2(y_2)} \label{eq:g_1_RN}
\end{align}
and
\begin{align}
	g_2(y_1,y_2) & \deq \frac{  \d(\mu_1 \otimes \mu_2K_2)}{  \d(\mu K)}(y_1,y_2) = \frac{\mu_1 (y_1)\mu_2K_2(y_2)}{p \mu_1 (y_1)\mu_2K_2(y_2) + \bar{p}\mu_1(y_1) \mu_2 K_2(y_2)} \label{eq:g_2_RN}.
\end{align}
A simple calculation shows that
	\begin{align*}
		P_{X_1,X_2|Y_1,Y_2}(\cdot|y_1,y_2)
		&= K^*(\cdot|y_1,y_2) \\
		&= p g_1(y_1,y_2) (K^*_1 \otimes \id_2)(\cdot|y_1,y_2) + \bar{p} g_2(y_1,y_2) (\id_1 \otimes K^*_2)(\cdot|y_1,y_2).
	\end{align*}
Now consider an arbitrary nonconstant function $f \in \PFunc(\sX_1 \times \sX_2)$. Then
\begin{align*}
&	\Ent_\Phi\big[f(X_1,X_2)\big] - \Ent_\Phi\big[K^*f(Y_1,Y_2)\big] \nonumber\\
&\quad= \E\left[ \Ent_\Phi\big[f(X_1,X_2)\big|Y_1,Y_2\big]\right] \\
&\quad = \sum_{y_1 \in \sY_1}\sum_{y_2 \in \sY_2} P_{Y_1,Y_2} (y_1, y_2) \Bigg[ \sum_{x_1 \in \sX_1}\sum_{x_2 \in \sX_2} P_{X_1,X_2|Y_1,Y_2}( x_1, x_2|y_1,y_2) \Phi\big(f(x_1,x_2)\big) \nonumber \\
&\qquad \qquad \qquad - \Phi\Big( \sum_{x_1 \in \sX_1}\sum_{x_2 \in \sX_2} P_{X_1,X_2|Y_1,Y_2}(  x_1,  x_2|y_1,y_2) f(x_1,x_2)\Big) \Bigg] \\
&\quad = \sum_{y_1 \in \sY_1}\sum_{y_2 \in \sY_2} \mu K(y_1,y_2) \Bigg[ p g_1(y_1,y_2) \sum_{x_1 \in \sX_1}\sum_{x_2 \in \sX_2} K^*_1 \otimes \id_2 (x_1, x_2 | y_1, y_2)  \Phi\big(f(x_1,x_2)\big) \nonumber \\
& \qquad \qquad \qquad \qquad \qquad + \bar{p} g_2(y_1,y_2)\sum_{x_1 \in \sX_1}\sum_{x_2 \in \sX_2} \id_1 \otimes K^*_2 (  x_1,   x_2 | y_1, y_2)  \Phi\big(f(x_1,x_2)\big) \Bigg] \nonumber \\
& \qquad \qquad - \sum_{y_1 \in \sY_1} \sum_{y_2 \in \sY_2} \mu K(  y_1,y_2) \Phi \Bigg( p g_1(y_1,y_2) \sum_{x_1 \in \sX_1}\sum_{x_2 \in \sX_2} K^*_1 \otimes \id_2 (  x_1,  x_2|y_1,y_2) f(x_1,x_2) \nonumber\\
& \qquad \qquad \qquad \qquad \qquad + \bar{p} g_2(y_1,y_2) \sum_{x_1 \in \sX_1}\sum_{x_2 \in \sX_2} \id_1 \otimes K^*_2 (  x_1,  x_2|y_1,y_2) f(x_1,x_2) \Bigg).
\end{align*}
From this, using the fact that $\Phi$ is convex and that $p g_1 + \bar{p} g_2 = 1$, we get
\begin{align}
	&	\Ent_\Phi\big[f(X_1,X_2)\big] - \Ent_\Phi\big[K^*f(Y_1,Y_2)\big] \nonumber\\
	& \ge p \sum_{y_1 \in \sY_1}\sum_{y_2 \in \sY_2} \mu K(  y_1,  y_2) g_1(y_1,y_2) \Bigg[ \sum_{x_1 \in \sX_1}\sum_{x_2 \in \sX_2} (K^*_1 \otimes \id_2)(  x_1,  x_2|y_1,y_2) \Phi\big(f(x_1,x_2)\big) \nonumber \\
	&\qquad \qquad \qquad \qquad - \Phi \Big( \sum_{x_1 \in \sX_1}\sum_{x_2 \in \sX_2} (K^*_1 \otimes \id_2) (  x_1,  x_2|y_1,y_2) f(x_1,x_2)\Big)\Bigg] \nonumber \\
	&\qquad + \bar{p}\sum_{y_1 \in \sY_1}\sum_{y_2 \in \sY_2} \mu K(  y_1,  y_2) g_2(y_1,y_2) \Bigg[ \sum_{x_1 \in \sX_1}\sum_{x_2 \in \sX_2} (\id_1 \otimes K^*_2)(  x_1,  x_2|y_1,y_2) \Phi\big(f(x_1,x_2)\big) \nonumber \\
	&\qquad \qquad \qquad \qquad - \Phi \Big( \sum_{x_1 \in \sX_1}\sum_{x_2 \in \sX_2} (\id_1 \otimes K^*_2) (  x_1,  x_2|y_1,y_2) f(x_1,x_2)\Big)\Bigg] \nonumber \\
	&= p   \sum_{y_2 \in \sY_2} \mu_2(  y_2) \sum_{y_1 \in \sY_1} \mu_1 K_1 (  y_1) \Bigg[ \sum_{x_1 \in \sX_1} K^*_1(  x_1|y_1) \Phi\big(f(x_1,y_2)\big)  - \Phi \Big( \sum_{x_1 \in \sX_1} K^*_1  (  x_1|y_1) f(x_1,y_2)\Big)\Bigg] \nonumber \\ 
	&\qquad + \bar{p}    \sum_{y_1 \in \sY_1} \mu_1(  y_1) \sum_{y_2 \in \sY_2} \mu_2 K_2 (  y_2) \Bigg[ \sum_{x_2 \in \sX_2} K^*_2(  x_2|y_2)\Phi\big(f(y_1,x_2)\big)  - \Phi \Big( \sum_{x_2 \in \sX_2} K^*_2  (  x_2|y_2) f(y_1,x_2)\Big)\Bigg], \label{eq:tensorization_2_convexity_bound}
\end{align}
where in the last step we have used the definitions \eqref{eq:g_1_RN} and \eqref{eq:g_2_RN} of $g_1$ and $g_2$. Now consider a random tuple $(X_1,X_2,U,V)$, such that
\begin{enumerate}
	\item $U \longrightarrow X_1 \longrightarrow X_2 \longrightarrow V$ is a Markov chain;
	\item $P_{X_1X_2} = \mu = \mu_1 \otimes \mu_2$;
	\item $P_{U|X_1} = K_1$;
	\item $P_{V|X_2} = K_2$.
\end{enumerate}
Using these definitions in \eqref{eq:tensorization_2_convexity_bound} gives
\begin{align}
		&	\Ent_\Phi\big[f(X_1,X_2)\big] - \Ent_\Phi\big[K^*f(Y_1,Y_2)\big]\nonumber\\
		& \qquad \ge p \sum P_{X_2}\left(  x_2\right) \E\left[ \Ent_\Phi \big[f(X_1,x_2)\big|U\big]\right] + \bar{p} \sum P_{X_1}\left(  x_1\right) \E\left[ \Ent_\Phi\big[f(x_1,X_2)\big|V\big]\right] \nonumber\\
		& \qquad \ge p (1-\eta_1)  \sum_{x_2 \in \sX_2} P_{X_2}\left(  x_2\right) \Ent_\Phi\big( f(X_1,x_2)\big) + \bar{p} (1-\eta_2) \sum_{x_1 \in \sX_1} P_{X_1}\left(  x_1\right) \Ent_\Phi\big[ f(x_1,X_2)\big] \nonumber\\
		& \qquad \ge \min\Big(p (1-\eta_1), \bar{p}(1-\eta_2)\Big) \Big\{ \E\left[\Ent_\Phi\big[f(X_1,X_2)\big| X_2\big]\right] + \E\left[\Ent_\Phi\big[f(X_1,X_2)\big| X_1\big]\right]\Big\} \label{eq:tensorization_2_1}\\
		& \qquad \ge \min\Big(p (1-\eta_1), \bar{p}(1-\eta_2)\Big) \Ent_\Phi \big[f(X_1,X_2)\big], \label{eq:tensorization_2_0}
\end{align}
where \eqref{eq:tensorization_2_1} is by the independence of $X_1$ and $X_2$, while \eqref{eq:tensorization_2_0} is by the assumed subadditivity of $\Ent_\Phi[\cdot]$. Since $f$ was arbitrary, we see that the inequality \eqref{eq:tensorization_2_goal} indeed holds.
\end{proof}

\begin{example} {\em Let $\sX_1 = \ldots = \sX_n = \{0,1\}$, $\mu_1 = \ldots = \mu_n = \Bernoulli(1/2)$, and $K_1 = \ldots = K_n = \BSC(\eps)$. Take $p$ to be the uniform distribution on $\{1,\ldots,n\}$. Then $K$ acts as follows: Given an $n$-bit input string $x^n = (x_1,\ldots,x_n)$, we pick one of the bits uniformly at random and flip it with probability $\eps$; the remaining bits stay the same. Then
	\begin{align*}
		\eta\left(\Bernoulli(1/2)^{\otimes n}, K\right) \le 1 - \frac{1-(1-2\eps)^2}{n} = 1 - \frac{4\eps\bar{\eps}}{n}.
	\end{align*}
In particular, when $\eps=1/2$, we get the upper bound of $1-1/n$.

We can also consider flipping bits in blocks: Let $\cB = \{B_m\}^k_{m=1}$ be a disjoint partition of the set $\{1,\ldots,n\}$ into $k$ blocks. We pick a block uniformly at random, and then independently flip each bit in that block with probability $\eps$. Denoting the resulting channel by $K_\cB$, we have
\begin{align}\label{eq:blocks_0}
	\eta\left(\Bernoulli(1/2)^{\otimes n}, K_{\cB}\right) \le 1-\frac{4\eps\bar{\eps}}{k}.
\end{align} 
To prove this, let $\mu^{(m)} = \bigotimes_{i \in B_m} \mu_i$ and $K^{(m)} = \bigotimes_{i \in B_m} K_i$. Then $\mu = \mu_1 \otimes \ldots \otimes \mu_n = \mu^{(1)} \otimes \ldots \otimes \mu^{(k)}$, and by Theorem~\ref{thm:tensorization_2} we have
\begin{align}\label{eq:blocks}
	\eta\left(\Bernoulli(1/2)^{\otimes n}, K_{\cB}\right) \le 1 - \frac{1}{k}\min_{1 \le m \le k} \left(1-\eta\big(\mu^{(m)},K^{(m)}\big)\right).
\end{align}
Since each $\mu^{(m)}$ is a product measure and each $K^{(m)}$ is a tensor product of BSCs, Theorem~\ref{thm:tensorization_1} gives 
\begin{align*}
	\eta\big(\mu^{(m)},K^{(m)}\big) = \max_{i \in B_m} \eta\left(\mu_i, K_i\right) = \eta\left(\Bernoulli(1/2), \BSC(\eps)\right) = (1-2\eps)^2.
\end{align*}
Substituting this into \eqref{eq:blocks}, we get \eqref{eq:blocks_0}. For $k=1$, $K_{\cB} \equiv \BSC(\eps)^{\otimes n}$, which has $\eta=(1-2\eps)^2$ by Theorem~\ref{thm:tensorization_1}. The bound of Eq.~\eqref{eq:blocks_0} is then achieved with equality. 
}
\end{example}

\subsection{Comparison of SDPI constants}

The following theorem shows that an upper bound on an SDPI constant for one source-channel pair can be converted into an upper bound for another such pair via a change-of-measure argument:

\begin{theorem}\label{thm:SDPI_comparison} Let $(\mu,K), (\bar{\mu},\bar{K}) \in \PProb(\sX) \times \Chan(\sY|\sX)$ be two admissible pairs.  Then, for any $\Phi \in \cF$ that satisfies the homogeneity condition \eqref{eq:gen_hom},
	\begin{align}\label{eq:SDPI_comparison}
		\eta_\Phi(\mu,K) \le 1 - \frac{a}{A}\left(1 - \eta_\Phi(\bar{\mu},\bar{K})\right),
	\end{align}
where 
		\begin{align*}
			A \deq \max_{(x,y) \in \sX \times \sY}\frac{\bar{\mu} \otimes \bar{K}(x,y)}{\mu \otimes K(x,y)}  \qquad \text{and} \qquad a \deq \min_{x \in \sX} \frac{\bar{\mu}(x)}{\mu(x)}.
		\end{align*}
\end{theorem}
\begin{remark}{\em It is easy to see that $0 < a \le A$. Indeed, the first inequality holds since $\mu,\bar{\mu} \in \PProb(\sX)$. For the second, by definition of $a$ and $A$, for every $x \in \sX$ we have
	\begin{align*}
		a \mu(x) \le \bar{\mu}(x) = \sum_{y \in \sY}\bar{\mu} \otimes \bar{K} (x,y) \le A \sum_{y \in \sY} \mu \otimes K (x,y) = A \mu (x).
	\end{align*}}
\end{remark}
\begin{proof} Consider random pairs $(X,Y)$ and $(\bar{X},\bar{Y})$ with respective probability laws $\mu \otimes K$ and $\bar{\mu} \otimes \bar{K}$. Using Eq.~\eqref{eq:functional_SDPI_gen_hom} and the law of total entropy Eq.~\eqref{eq:total_entropy},
	\begin{align}
		\eta_\Phi(\mu,K) &= \sup\left\{ \frac{\Ent_\Phi \left[\E[ f(X)|Y]\right]}{\Ent_\Phi\left[f(X)\right]} : f \in \PFunc(\sX)\, f \neq {\rm const} \right\} \\
		&= \sup \left\{ \frac{\Ent_\Phi [f(X)] - \E\left[\Ent_\Phi\big[f(X)\big|Y\big]\right]}{\Ent_\Phi[f(X)]} : f \in \PFunc(\sX),\, f \neq {\rm const}\right\} \nonumber \\
		&= 1 - \inf\left\{ \frac{\E\left[\Ent_\Phi\big[f(X)\big|Y\big]\right]}{\Ent_\Phi[f(X)]} : f \in \PFunc(\sX),\, f \neq {\rm const} \right\}. \label{eq:comparison_1}
	\end{align}
Using Lemma~\ref{lm:var_cond_ent} in Appendix~\ref{app:lemmas}, we can write
\begin{align*}
 \E\left[\Ent_\Phi\big[f(X)\big|Y\big]\right] &= \inf_{\xi \in \PFunc(\sY)} \E\left[\Phi(f(X))-\Phi(\xi(Y))-(f(X)-\xi(Y))\Phi'(\xi(Y))\right]\\
 &= \inf_{\xi \in \PFunc(\sY)} \E\left[\frac{\d(\mu \otimes K)}{\d(\bar{\mu} \otimes \bar{K})}(\bar{X},\bar{Y})\left( \Phi(f(\bar{X}))-\Phi(\xi(\bar{Y}))-(f(\bar{X})-\xi(\bar{Y}))\Phi'(\xi(\bar{Y})) \right)\right] \\
 &\ge \frac{1}{A} \inf_{\xi \in \PFunc(\sY)}\E\left[ \Phi(f(\bar{X}))-\Phi(\xi(\bar{Y}))-(f(\bar{X})-\xi(\bar{Y}))\Phi'(\xi(\bar{Y})) \right] \\
 &= \frac{1}{A}	\E\left[\Ent_\Phi\big[f(\bar{X})\big|\bar{Y}\big]\right],
\end{align*}
where the inequality follows from the definition of $A$ and from the convexity of $\Phi$. An analogous argument gives the inequality
\begin{align*}
	\Ent_\Phi\big[f(X)\big] \le \frac{1}{a} \Ent_\Phi\big[f(\bar{X})\big].
\end{align*}
Using these estimates in \eqref{eq:comparison_1}, we get
\begin{align*}
	\eta_\Phi(\mu,K) &\le 1 - \frac{a}{A}\inf \left\{ \frac{\E\left[\Ent_\Phi\big[f(\bar{X})\big|\bar{Y}\big]\right]}{\Ent_\Phi[f(\bar{X})]} : f \in \PFunc(\sX),\, f \neq {\rm const} \right\} \\
	&= 1 - \frac{a}{A}\left(1 - \eta_\Phi(\bar{\mu},\bar{K}) \right).
\end{align*}
\end{proof}
\begin{corollary}\label{cor:SDPI_comparison} If two channels $K,\bar{K} \in \Chan(\sY|\sX)$ are such that $\bar{K}(y|x) \le A K(y|x)$ for all $(x,y) \in \sX \times \sY$, then 
	\begin{align*}
		\eta_\Phi(\mu,K) \le 1 - \frac{1}{A}\left(1-\eta_\Phi(\mu,\bar{K})\right)
	\end{align*}
	for any $\Phi \in \cF$ satisfying \eqref{eq:gen_hom} and any $\mu \in \PProb(\sX)$.
\end{corollary}

\subsection{Extremal functions}
\label{ssec:extremal}

In this section, we will characterize the extremal functions $f \in \PFunc(\sX)$ that attain the infimum in \eqref{eq:functional_SDPI}. In particular, we will prove that, for any sufficiently smooth $\Phi$, these functions are solutions of the variational equation
\begin{align}\label{eq:SDP_var_0}
	\E\left[ \Phi'\left(\E[f(\bar{X})|Y]\right)\Big|X\right] - \Phi'(1) = \eta\left( \Phi' \left(f(X)\right)-\Phi'(1)\right),
\end{align}
with $\eta = \eta_\Phi(\mu,K)$ under the constraint $f \in \PFunc(\sX)$ and $\E[f(X)]=1$. Here, the random triple $(X,\bar{X},Y)$ is such that $X \to Y \to \bar{X}$ is a Markov chain, and both $(X,Y)$ and $(\bar{X},Y)$ have law $\mu \otimes K$. Written more compactly, \eqref{eq:SDP_var_0} takes the form
\begin{align}\label{eq:SDP_var}
	K (\Phi' \circ K^*f) - \Phi'(1) = \eta \left( \Phi' \circ f - \Phi'(1)\right).
\end{align}
\begin{theorem}\label{thm:SDP_var} Suppose $\Phi \in \cF$ has the following properties:
	\begin{enumerate}[(a)]
		\item It is three times differentiable with $\Phi''(1) > 0$.
		\item The $\Phi$-entropy functional is homogeneous in the sense of Definition~\ref{def:gen_hom}.
		\item There exists a constant $c > 0$, such that
		\begin{align}\label{eq:entropy_via_derivative}
			\Ent_\Phi[U] = c\, \E\left[U\left(\Phi'(U)-\Phi'(1)\right)\right]
		\end{align}
		for any nonnegative-valued random variable $U$ with $\E U = 1$.
	\end{enumerate}
	Then either $\eta_\Phi(\mu,K) = S^2(\mu,K)$, or there exists a nonconstant function $f \in \PFunc(\sX)$, such that \eqref{eq:SDP_var} holds with $\eta = \eta_\Phi(\mu,K)$. Moreover, $\eta_\Phi(\mu,K)$ is the smallest constant $\eta > 0$, for which \eqref{eq:SDP_var} has a solution among nonconstant functions in $\PFunc(\sX)$.
\end{theorem}

\begin{remark} {\em The functions $\Phi(u) = u \log u$ and $\Phi(u) = \frac{u^p-1}{p-1}$, $1 < p \le 2$, satisfy the condition \eqref{eq:entropy_via_derivative} (details are provided in the examples after the proof). On the other hand, the function $\Phi(u) = -\log u$ cannot satisfy \eqref{eq:entropy_via_derivative} for any choice of $c$, since $\Ent_\Phi[U] = -\E \log U$, while $\E[U(\Phi'(U)-\Phi'(1))] = \E U - 1 = 0$ for any nonnegative-valued $U$ with $\E U = 1$.\hfill$\diamond$}
\end{remark}

\begin{remark} {\em We emphasize that the variational equation \eqref{eq:SDP_var} may have multiple solutions, not all of which are actually extremal. In general, it is not easy to obtain explicit closed-form expressions for the extremal solutions of \eqref{eq:SDP_var}.\hfill$\diamond$}
\end{remark}

\begin{proof} Suppose that $\eta_\Phi(\mu,K) > S^2(\mu,K)$, for otherwise there is nothing to prove. We seek to minimize the functional
	\begin{align*}
		W(f) \deq \frac{\Ent_\Phi[f(X)]-\Ent_\Phi[K^* f(Y)]}{\Ent_\Phi[f(X)]} = \frac{\E\left[\Ent_\Phi[f(X)|Y]\right]}{\Ent_\Phi[f(X)]}
	\end{align*}
	over all $f \in \PFunc(\sX)$. By homogeneity, $W(c f) = W(f)$ for all $c > 0$, so without loss of generality we can restrict the minimization to $f \in \cM \deq \left\{ f \in \PFunc(\sX) : E[f(X)]=1 \right\}$. For $\eps > 0$, define the set
	$$
	\cM_\eps \deq \left\{ f \in \PFunc(\sX): \E[f(X)]=1 \text{ and } \| f - 1 \|_\infty < \eps \right\},
	$$
so $\cM = \cM_\infty$. From the Taylor expansion 
\begin{align*}
	\Phi(1+u) = \Phi(1) + \Phi'(1)u + \frac{\Phi''(1)}{2}u^2 + O\left(u^3\right),
\end{align*}
we have, for every $f \in \cM_\eps$,
\begin{align*}
	\Ent_\Phi[f(X)] = \frac{\Phi''(1)}{2} \Var[f(X)] + O(\eps^3)
\end{align*}
and
\begin{align*}
	\E\left[\Ent_\Phi[f(X)|Y]\right] = \frac{\Phi''(1)}{2}\left(\Var[f(X)]-\Var[\E[f(X)|Y]]\right) + O\left(\eps^3\right).
\end{align*}
Therefore,
\begin{align*}
	\inf_{f \in \cM_\eps} W(f) &= \inf_{f \in \cM_\eps} \frac{\Var[f(X)]-\Var[\E[f(X)|Y|]]+O(\eps^3)}{\Var[f(X)]+O(\eps^3)},
\end{align*}
which implies that, for any $\delta \in (0,1)$ there exists some $\eps_0 = \eps_0(\delta)$, such that
\begin{align*}
	\inf_{f \in \cM_{\eps_0}} W(f) &\ge \inf_{g \in \cM} \frac{\Var[g(X)]-\Var[\E[g(X)]|Y|]}{\Var[g(X)]} + \delta \\
	&= 1 - S^2(\mu,K) + \delta \\
	&> 1 - \eta_\Phi(\mu,K) + \delta.
\end{align*}
On the other hand, since $\inf_{f \in \cM} W(f) = 1-\eta_\Phi(\mu,K)$, any $f$ that minimizes $W$, if it exists, must lie in $\cM \backslash \cM_{\eps_0}$, i.e., it must be nonconstant. It remains to show the existence of such a minimizing $f$. Since any $f \in \cM$ satisfies $\| f \|_\infty \le 1/\mu_*$, where $\mu_*$ is the smallest (positive) mass of $\mu$, the set $\cM \backslash \cM_{\eps_0}$ is a closed and bounded subset of a finite-dimensional linear space, hence compact. The denominator of $W(f)$ is positive for all $f \in \cM \backslash \cM_{\eps_0}$, so $W$ is a continuous functional on the compact set $\cM \backslash \cM_{\eps_0}$ and thus attains its infimum on some nonconstant $f \in \cM$.

Now, let $f$ be such a minimizing function. We use a variational argument following Bobkov and Tetali \cite[Sec.~6]{Bobkov_Tetali_logsob}. Given an arbitrary $g \in \Func(\sX)$, the perturbed function $f + \eps g$ is nonnegative for all sufficiently small $\eps > 0$. Consequently, by definition of $\eta_\Phi$,
\begin{align}\label{eq:1st_variation}
	(1-\eta_\Phi)\Ent_\Phi\left[f(X)+\eps g(X)\right] \le \E\left[\Ent_\Phi[f(X)+\eps g(X)|Y]\right].
\end{align}
Applying the Taylor expansion
\begin{align*}
	\Ent_\Phi[U + \eps V] &= \E[\Phi(U+\eps V)] - \Phi(\E U + \eps \E V) \\
	&= \Ent_\Phi[U] + \eps \E\left[ \big(\Phi'(U) - \Phi'(\E U)\big) V \right] + O(\eps^2)
\end{align*}
to $U = f(X)$, $V = g(X)$ first for $X \sim \mu$ and then for $X \sim K^*(\cdot|y)$, $y \in \sY$, and then using the fact that $(1-\eta_\Phi) \Ent_\Phi[f(X)] = \E\left[\Ent_\Phi[f(X)|Y]\right]$ by the extremality of $f$, we have
\begin{align}
& (1-\eta_\Phi) \Ent_\Phi[f(X)+\eps g(X)] - \E\left[\Ent_\Phi[f(X)+\eps g(X)|Y]\right] \nonumber\\
& \qquad =  \eps\,\E\Big[ (1-\eta_\Phi)\big(\Phi'(f(X)) - \Phi'(1)\big) g(X) - \E\big[\big(\Phi'(f(X)) - \Phi'(\E[f(X)|Y])\big) g(X) |Y\big]\Big]
+ o(\eps) \nonumber\\
& \qquad = \eps\,\E\Big[ \Phi'(\E[f(X)|Y]) \E[g(X)|Y] - \eta_\Phi\big(\Phi'(f(X)) - \Phi'(1)\big)g(X) \Big] + o(\eps) \nonumber\\
& \qquad = \eps\,\E\Big[ \Big(\E[\Phi'(\E[f(\bar{X})|Y])|X] - \eta_\Phi \Phi'(f(X) - (1-\eta_\Phi) \Phi'(1)\big) g(X)  \Big)\Big] + o(\eps), \label{eq:1st_variation_2}
\end{align}
where in the last line $\bar{X}$ is an independent and identically distributed copy of $X$ given $Y$, and we have used the fact that $\E[\xi(Y)\E[\gamma(X)|Y]] = \E[\E[\xi(Y)|X]\gamma(X)]$ for any pair $\gamma \in \Func(\sX), \xi \in \Func(\sY)$. Now, by \eqref{eq:1st_variation}, the leftmost quantity in \eqref{eq:1st_variation_2} is nonpositive, whereas the rightmost quantity will be nonpositive for all sufficiently small $\eps > 0$ if and only if
\begin{align*}
	\E\big[ \big(\E[\Phi'(\E[f(\bar{X})|Y])|X] - \eta_\Phi\big(\Phi'(f(X)) - (1-\eta_\Phi) \Phi'(1)\big) g(X)  \big)\big] = 0.
\end{align*}
Since $g$ is arbitrary and $\mu \in \PProb(\sX)$, the minimizing function $f \in \PFunc(\sX)$ with $\E[f(X)]=1$ must satisfy
\begin{align*}
	\E[\Phi'(\E[f(\bar{X})|Y])|X] - \Phi'(1) = \eta_\Phi\big(\Phi'(f(X)) - \Phi'(1)\big),
\end{align*}
which is precisely \eqref{eq:SDP_var}.

It remains to show minimality. To that end, let $\tilde{\eta} > 0$ be another constant such that there exists some function $\tilde{f} \in \PFunc(\sX)$ with $\E[\tilde{f}(X)]=1$ satisfying
\begin{align}\label{eq:SDP_var_alt}
		\E\big[ \Phi'\big(\E[\tilde{f}(\bar{X})|Y]\big)\big|X\big] - \Phi'(1) = \tilde{\eta}\big( \Phi' \big(\tilde{f}(X)\big)-\Phi'(1)\big)
\end{align}
Multiplying both sides of \eqref{eq:SDP_var_alt} by $\tilde{f}(X)$, taking expectations, and using \eqref{eq:entropy_via_derivative}, we get
\begin{align*}
	\Ent_\Phi[\E[\tilde{f}(X)|Y]] = \tilde{\eta} \Ent_\Phi[\tilde{f}(X)].
\end{align*}
By definition of $\eta_\Phi(\mu,K)$, we must have $\tilde{\eta} \le \eta_\Phi(\mu,K)$.
\end{proof}
The proof of the theorem shows that if $\eta_\Phi(\mu,K) > S^2(\mu,K)$, then Eq.~\eqref{eq:SDP_var} admits a nontrivial (i.e., nonconstant) solution. The contrapositive of this statement gives:

\begin{corollary}\label{cor:SDP_trivial} If the infimum in \eqref{eq:functional_SDPI} is not achieved, i.e., if the SDPI \eqref{eq:f_SDPI} is strict unless $f \equiv 1$, then $\eta_\Phi(\mu,K) = S^2(\mu,K)$.
\end{corollary}
\begin{remark} {\em Equivalently, $\eta_\Phi(\mu,K)=S^2(\mu,K)$ if for an arbitrary $\gamma > 0$ the only solution to Eq.~\eqref{eq:SDP_var} among $f \in \PFunc(\sX)$ with $\E[f(X)]=1$ is the trivial solution $f \equiv 1$.\hfill$\diamond$}
\end{remark}

Here are a couple of specific examples:
\begin{itemize}
	\item For $\Phi(u) = u \log u$, we have $\Phi'(u) = \log u + 1$, and $\Phi$ satisfies the conditions (a)--(c) of Theorem~\ref{thm:SDP_var}. In particular, Eq.~\eqref{eq:entropy_via_derivative} holds with $c = 1$. The variational equation \eqref{eq:SDP_var} becomes
	\begin{align*}
		K(\log K^* f) = \eta \log f, \qquad \eta = \eta(\mu,K).
	\end{align*}
	\item For $\Phi(u) = \frac{u^p-1}{p-1}$, $1 < p \le 2$, we have $\Phi'(u) = \frac{pu^{p-1}}{p-1}$, and $\Phi$ satisfies the conditions (a)--(c). In this case, \eqref{eq:entropy_via_derivative} holds with $c=p$. The variational equation takes the form
	\begin{align*}
	K \left((K^*f)^{p-1}\right) = \eta_p f^{p-1} + 1-\eta_p, \qquad \eta_p = \eta_{\Phi_p}(\mu,K).
	\end{align*}
\end{itemize}

\section{Connections with $\Phi$-Sobolev inequalities}
\label{sec:phisob}

\subsection{General framework}

Strong data processing inequalities for a pair $(\mu,K)$ can be interpreted in terms of the effect of the adjoint channel $K^*$ on the $\Phi$-entropies of suitably normalized nonnegative functions of the input, see Proposition~\ref{prop:functional_SDPI}. In this section, we show that there is a close relationship between SDPIs and another class of functional inequalities --- the so-called \textit{$\Phi$-Sobolev inequalities} \cite{Chafai_entropy,Boucheron_etal_concentration_book} that relate the $\Phi$-entropy $\Ent_\Phi[f(X)]$ of an arbitrary function of the input $X \sim \mu$ to some measure of correlation between $f(X)$ and the output $Y \sim \mu K$.

We will measure correlation in the following way. For any triple $(U,V,Z)$ of jointly distributed random variables, where $U,V$ are real-valued, we define
\begin{align}\label{eq:Dirichlet}
	\cE(U,V|Z)  &\deq \E\left[\big(U-\E[U|Z]\big)\big(V-\E[V|Z]\big)\right].
\end{align}
This quantity has an estimation-theoretic interpretation: since $e(U|Z) \deq U-\E[U|Z]$ is the error of a minimum mean-square error (MMSE) estimator of $U$ given $Z$, and $\E[e(U|Z)] = 0$, $\cE(U,V|Z)$ is the covariance of $e(U|Z)$ and $e(V|Z)$: \begin{align*}
\cE(U,V|Z) = \Cov\left[e(U|Z),e(V|Z)\right].
\end{align*}
In particular, 
\begin{align*}
\cE(U,U|Z) = \E\Big[\big(U - \E[U|Z]\big)^2\Big] \equiv \MMSE(U|Z),
\end{align*}
the MMSE achievable in estimating $U$ from $Z$. We pause to record a few key properties of $\cE$ (see Appendix~\ref{app:E_proof} for the proof):
\begin{proposition}\label{prop:E} The functional $\cE$ defined in \eqref{eq:Dirichlet} has the following properties:
	\begin{enumerate}
		\item Symmetry -- $\cE(U,V|Y) = \cE(V,U|Y)$.
		\item Linearity -- $\cE(a U + b U',V|Y) = a\,\cE(U,V|Y) + b\,\cE(U',V|Y)$ for any constants $a,b \in \Reals$.
		\item Degeneracy -- If $U$ is constant a.s., then $\cE(U,V|Y)=0$.
		\item Representation in terms of an exchangeable pair -- Let $(X,Y) \in \sX \times \sY$ be a random pair with $P_X = \mu \in \Prob(\sX)$ and $P_{Y|X} = K \in \Chan(\sY|\sX)$, where $(\mu,K)$ is an admissible pair. Then for any two functions $f,g\in\Func(\sX)$,
		\begin{align}
			\cE(f(X),g(X)|Y)
			&= \frac{1}{2} \E\left[\big(f(X)-f(X')\big)\big(g(X)-g(X')\big)\right] \label{eq:exch_1}\\
			&= \E\left[ \big(f(X)-f(X')\big)_+\big(g(X)-g(X')\big)\right], \label{eq:exch_2}
		\end{align}
		where $(u)_+ \deq u \vee 0$, and $(X,X')$ is a pair of $\sX$-valued random variables with $P_X = \mu$ and $P_{X'|X} = K^* K$.
	\end{enumerate}
\end{proposition}
\begin{remark} {\em The terminology in Item~4 merits some discussion. It is not hard to show (and, in fact, we do show it in the proof of the proposition) that the joint distribution $P_{XX'}$ has the following symmetry property:
\begin{align}\label{eq:exchangeability}
	P_{XX'}(x,x') = P_{XX'}(x',x), \qquad \forall x,x' \in \sX.
\end{align}
In other words, the random variables $X$ and $X'$ form an \textit{exchangeable pair}.\hfill$\diamond$
}
	
\end{remark}

Generalizing the definition due to Chafa\"i \cite{Chafai_entropy}, we now introduce \textit{$\Phi$-Sobolev inequalities}:

\begin{definition} Consider an admissible pair $(\mu,K)$ and a random pair $(X,Y)$ with probability distribution $\mu \otimes K$. Fix a function $\Phi \in \cF$. We say that $(\mu,K)$ satisfies a $\Phi$-Sobolev inequality with constant $\alpha \ge 0$ if there exists some function $\Psi : \Reals^+ \to \Reals$, such that the inequality
	\begin{align}\label{eq:fsob}
		\Ent_\Phi[f(X)] \le \alpha\, \cE\big( f(X), \Psi \circ f(X) \big | Y \big)
	\end{align}
holds for all $f \in \PFunc(\sX)$.
\end{definition}
Now we are ready to state our main result that relates  SDPIs to $\Phi$-Sobolev inequalities:

\begin{theorem}\label{thm:SDPI_to_phiSob} Suppose that $\Phi \in \cF$ is such that $\Phi(0) < \infty$, the function $\Psi$ defined in \eqref{eq:first_diff_0} is concave, and the corresponding $\Phi$-entropy is homogeneous. Then $\eta_\Phi(\mu,K) \le c$ implies that the pair $(\mu,K)$ satisfies the $\Phi$-Sobolev inequality of the form \eqref{eq:fsob} with constant $\alpha = (1-c)^{-1}$.
\end{theorem}

\begin{proof} For any $u > 0$ we can write $\Phi(u) = u\Psi(u) + \Phi(0)$. Thus, for any real-valued random variable $U$ which is a.s.\ strictly positive and a jointly distributed random variable $Y$, we have
\begin{align}
	\Ent_\Phi[U|Y] &= \E[U\Psi(U)|Y] - \E[U|Y] \Psi(\E[U|Y]) \nonumber\\
 &\le \E[U\Psi(U)|Y] - \E[U|Y]\E[\Psi(U)|Y],\label{eq:cond_ent_psi}
\end{align}
where the second line is by the concavity of $\Psi$. Now let $U = f(X)$ for some $f \in \SPFunc(\sX)$. Using \eqref{eq:cond_ent_psi} and Proposition~\ref{prop:functional_SDPI}, we get
\begin{align*}
	\Ent_\Phi\big[f(X)\big] 
	&\le \frac{1}{1-c} \E\left[f(X)\Psi(f(X))-\E[f(X)|Y]\E[\Psi(f(X)|Y]\right] \\
	&= \frac{1}{1-c}\cE(f(X),\Psi \circ f(X)|Y),
\end{align*}
where the second line follows from the easily verified identity $\cE(U,V|Z) = \E[UV - \E[U|Z]\E[V|Z]] = \E[U(V - \E[V|Z])]$.\end{proof}

Theorem~\ref{thm:SDPI_to_phiSob} provides a route to $\Phi$-Sobolev inequalities via SDPIs --- any good upper bound on $\eta_\Phi(\mu,K)$ would automatically translate into a bound on the constant in the corresponding $\Phi$-Sobolev inequality. Such functional inequalities are a powerful tool in applied probability (for example, in the context of quantifying the convergence of Markov chains to equilibrium); in the next section, we will illustrate this on the particular case of Poincar\'e inequalities (corresponding to $\Phi(u) = u^2-1$ and log-Sobolev inequalities (corresponding to $\Phi(u)=u\log u$).

It is often useful to estimate the $\Phi$-entropy of a composite function $F \circ f$ (we will see examples of this later on). The following result contains Theorem~5 of \cite{Boucheron_etal_moment_inequalities} as a special case:

\begin{theorem}\label{thm:composite_fctn_phiSob} Suppose that the assumptions of Theorem~\ref{thm:SDPI_to_phiSob} hold, and that the function $\Psi$ is differentiable. Let $F : \Reals \to \Reals^+$ be a convex, differentiable, nondecreasing function, such that $\Psi \circ F$ is convex. Then, for any $f \in \Func(\sX)$,
	\begin{align*}
\Ent_\Phi[F \circ f(X)] \le \frac{1}{1-c} \E\Big[ F'^2\big(f(X)\big)\Psi'\big(F \circ f(X)\big)\left(f(X)-f(X')\right)^2_+\Big],
\end{align*}
where $(X,X')$ is an exchangeable pair of random variables with $P_X = \mu$ and $P_{X'|X}=K^* K$, and $\Psi'$ denotes the right derivative of $\Psi$. Similarly, if $F$ is nonincreasing, then
	\begin{align*}
\Ent_\Phi[F \circ f(X)] \le \frac{1}{1-c} \E\Big[ F'^2\big(f(X)\big)\Psi'\big(F \circ f(X)\big)\left(f(X')-f(X)\right)^2_+\Big],
	\end{align*}
\end{theorem}
\begin{proof} We only consider the case when $F$ is nondecreasing, since the other case is handled similarly. Suppose that $u > v$. Then, by monotonicity and convexity of $F$,
	\begin{align*}
		0 \le F(u) - F(v) \le F'(u)(u-v).
	\end{align*}
Moreover, because $f$ is convex, the function $\Psi$ defined in \eqref{eq:first_diff_0} is nondecreasing. Using this together with the assumed convexity of $\Psi \circ F$, we have
\begin{align*}
	0 \le \Psi (F(u)) - \Psi(F(v)) \le \Psi'(F(u))F'(u)(u-v).
\end{align*}
Thus, when $u > v$,
\begin{align}
	\left(F(u)-F(v)\right)\left(\Psi(F(u))-\Psi(F(v))\right) \le F'^2(u)\Psi'(F(u))(u-v)^2.\label{eq:psi_eta_bound}
\end{align}
Therefore, using Theorem~\ref{thm:SDPI_to_phiSob}, we can write
\begin{align*}
	\Ent_\Phi[F \circ f(X)] 
	&\le \frac{1}{1-c} \cE(F \circ f(X), \Psi \circ F \circ f(X)|Y) \\
	&= \frac{1}{1-c} \E\Big[ \left( F(f(X))-F(f(X'))\right)_+ \cdot\left( \Psi\big(F(f(X))\big)-\Psi\big(F(f(X'))\big)\right)\Big] \\
	&\le \frac{1}{1-c} \E\Big[ F'^2\big(f(X)\big)\Psi'\big(F(f(X))\big)\left(f(X)-f(X')\right)^2_+\Big],
\end{align*}
where the second step is by \eqref{eq:exch_2}, while the last step is by \eqref{eq:psi_eta_bound}.
\end{proof}

\subsection{Logarithmic Sobolev and Poincar\'e inequalities}
\label{ssec:logsob}

We now particularize the above general results to two specific types of functional inequalities:
\begin{itemize}
	\item logarithmic Sobolev inequalities, with $\Phi(u) = u \log u$;
	\item Poincar\'e inequalities, with $\Phi(u) = u^2 - 1$.
\end{itemize}
These inequalities are well-known in functional analysis and probability theory (see, e.g., \cite{Bakry_logsob_notes,Diaconis_Saloff_logsob,Miclo_hypercontractive,Bobkov_Tetali_logsob,Mossel_reverse_hypercont}). We will first introduce our definitions of these inequalities following the ideas laid down in the preceding section, and then show how these definitions are related to the ``standard'' ones.

We start with Poincar\'e inequalities:
\begin{definition} We say that an admissible pair $(\mu,K) \in \Prob(\sX) \times \Chan(\sY|\sX)$ satisfies a Poincar\'e inequality with constant $\alpha \ge 0$ if
	\begin{align*}
		\Var \big[ f(X) \big] \le \alpha \,\cE\big(f(X),f(X) \big| Y\big)
	\end{align*}
	for all $f \in \PFunc(\sX)$, where $(X,Y)$ is a random pair with probability law $\mu \otimes K$. The {\em Poincar\'e constant} of $(\mu,K)$ is given by
	\begin{align*}
		\lambda(\mu,K) \deq \inf_{f \in \PFunc(\sX)} \frac{\cE\big(f(X),f(X)\big|Y\big)}{\Var\big[f(X)\big]},
	\end{align*}
	where we adopt the convention that $\frac{0}{0} = + \infty$.
\end{definition}
According to the above definition, $\alpha^* = \frac{1}{\lambda(\mu,K)}$ is the smallest value of $\alpha$ for which the pair $(\mu,K)$ will satisfy a Poincar\'e inequality. Moreover, we have the following:
\begin{proposition} For any admissible pair $(\mu,K)$,
	\begin{align*}
		\lambda(\mu,K) = 1 - S^2(\mu,K).
	\end{align*}
	That is, $(\mu,K)$ satisfies a Poincar\'e inequality with constant $\alpha$ if and only if $\eta_{\chi^2}(\mu,K) \le 1-1/\alpha$.
\end{proposition}
\begin{proof} The function $\Phi(u)=u^2-1$ satisfies the conditions of Theorem~\ref{thm:SDPI_to_phiSob} with $\Psi(u) = u$, and $\Ent_\Phi[U] = \Var[U]$.  Therefore, if $\eta_{\chi^2}(\mu,K) \le c$, then for any $f \in \PFunc(\sX)$ we have
	\begin{align*}
		\Var\big[f(X)\big] &\le \frac{1}{1-c} \cE\big(f(X), \Psi \circ f(X) \big| Y \big) \\
		&= \frac{1}{1-c} \cE\big(f(X),f(X)\big| Y\big),
	\end{align*}
	which implies that the pair $(\mu,K)$ satisfies Poincar\'e with constant $\alpha = \frac{1}{1-c}$. Therefore,
	\begin{align*}
		\lambda(\mu,K) &\ge \sup \left\{ 1-c : \eta_{\chi^2}(\mu,K) \le c \right\} \\
		&= 1-\eta_{\chi^2}(\mu,K) \\
		&= 1-S^2(\mu,K),
	\end{align*}
	where the last step is by Theorem~\ref{thm:maxcorr}.
	
Conversely, suppose that $(\mu,K)$ satisfies Poincar\'e with constant $\alpha$. A simple computation shows
	\begin{align*}
		\cE\big(f(X),f(X)\big|Y\big) = \Var\big[f(X)\big] - \Var\big[K^* f(Y)\big].
	\end{align*}
Therefore,
\begin{align*}
	\Var\big[K^*f(Y)\big] \le \left(1-\frac{1}{\alpha}\right) \Var\big[f(X)\big]
\end{align*}
for any $f \in \PFunc(\sX)$. This, in turn, implies that
\begin{align*}
	S^2(\mu,K) &= \sup_{f \in \PFunc(\sX)} \frac{\Var\big[K^*f(Y)\big]}{\Var\big[f(X)\big]} \\
	&\le \inf\left\{ 1-\frac{1}{\alpha} : \frac{1}{\alpha} \le \lambda(\mu,K) \right\} \\
	&=1-\lambda(\mu,K).
\end{align*}
\end{proof}

Now let us consider log-Sobolev inequalities:
\begin{definition} We say that an admissible pair $(\mu,K) \in \Prob(\sX) \times \Chan(\sY|\sX)$ satisfies a logarithmic Sobolev inequality with constant $\alpha \ge 0$ if
	\begin{align*}
		\Ent \big[ f(X) \big] \le \alpha \,\cE\big(f(X),\log f(X) \big| Y\big)
	\end{align*}
	for all $f \in \SPFunc(\sX)$, where $(X,Y)$ is a random pair with probability law $\mu \otimes K$. The {\em log-Sobolev constant} of $(\mu,K)$ is given by
	\begin{align}\label{eq:logSob_constant}
		\rho_1(\mu,K) \deq \inf_{f \in \SPFunc(\sX)} \frac{\cE\big(f(X),\log f(X)\big|Y\big)}{\Ent\big[f(X)\big]},
	\end{align}
again with the convention that $\frac{0}{0} = + \infty$.
\end{definition}
The following is an extension of Prop.~5.1 in \cite{DelMoral_contraction} to the case $\sX \neq \sY$, and with explicit constants:

\begin{proposition}\label{prop:SDPI_vs_logSob} For any admissible pair $(\mu,K)$,
	\begin{align}\label{eq:SDPI_vs_logSob}
		1-\eta(\mu,K) \le \rho_1(\mu,K) \le 1-\frac{(1-\log 2)\log 2}{2}\eta(\mu,K) \le 1 - \frac{1}{10} \eta(\mu,K).
	\end{align}
	That is, if $\eta(\mu,K) \le c$, then $(\mu,K)$ satisfies a log-Sobolev inequality with constant $\alpha = \frac{1}{1-c}$. Conversely, if $(\mu,K)$ satisfies log-Sobolev with constant $\alpha$, then
	\begin{align*}
		\eta(\mu,K) \le \frac{2}{(1-\log 2)\log 2}\left(1 - \frac{1}{\alpha}\right) \le 10 \left(1 - \frac{1}{\alpha}\right).
	\end{align*}
\end{proposition}
\begin{proof} The first inequality in \eqref{eq:SDPI_vs_logSob} follows from Theorem~\ref{thm:SDPI_to_phiSob} with $\Phi(u) = u\log u$ and $\Psi(u) = \log u$. To prove the second inequality, we borrow (and slightly streamline) an ingenious idea from \cite{DelMoral_contraction}. Let us fix an arbitrary $f \in \SPFunc(\sX)$. Then
	\begin{align}
		\E\left[\Ent\big[f(X)\big|Y\big]\right] &= \sum_{y \in \sY} \mu K(y) \sum_{x \in \sX} K^*(x|y) f(x) \log \frac{f(x)}{K^*f(y)} \nonumber\\
		&= \sum_{y \in \sY} \mu K(y) \Ent \big[ f(X) \big| Y=y\big] \label{eq:cond_ent_sum}.
	\end{align} 
By \cite[Lm.~5.2]{DelMoral_contraction}, the entropy $\Ent[U]$ of any nonnegative real-valued random variable $U$ with $\E[U \log U] < \infty$ admits the integral representation
	\begin{align}\label{eq:ent_sym}
		\Ent[U] &= \frac{1}{2}\int^\infty_0 e^{-t}\E\left[\left(U - \bar{U}\right)\log \frac{e^{-t}U + 1-e^{-t}}{e^{-t}\bar{U}+1-e^{-t}}\right] \d t,
	\end{align}
where $\bar{U}$ is an independent copy of $U$. Applying \eqref{eq:ent_sym} to each term in \eqref{eq:cond_ent_sum}, we obtain
\begin{align*}
	\Ent[f(X)|Y=y]= \frac{1}{2}\sum_{x,\bar{x} \in \sX} K^*(x|y)K^*(\bar{x}|y)\left[\int^\infty_0 \left(f_t(x)-f_t(\bar{x})\right) \left( \log f_t(x) - \log f_t(\bar{x})\right) \d t\right],
\end{align*}
where $f_t(x) \deq e^{-t}f(x) + 1 -e^{-t}$. Averaging this w.r.t.\ $Y \sim \mu K$ gives
\begin{align}
&	\E\left[\Ent\big[f(X)\big|Y\big]\right] \nonumber \\
& =  \frac{1}{2} \sum_{y \in \sY}\sum_{x,\bar{x} \in \sX} \mu K (y) K^*(x|y)K^*(\bar{x}|y)  \left[\int^\infty_0 \left(f_t(x)-f_t(\bar{x})\right) \left( \log f_t(x) - \log f_t(\bar{x})\right) \d t\right]\nonumber\\
& = \frac{1}{2} \sum_{x,\bar{x} \in \sX} \mu(x) \sum_{y \in \sY} K(y|x)K^*(\bar{x}|y)  \left[\int^\infty_0 \left(f_t(x)-f_t(\bar{x})\right) \left( \log f_t(x) - \log f_t(\bar{x})\right) \d t\right] \nonumber\\
&= \frac{1}{2} \int^\infty_0 \E \left[ \left(f_t(X)-f_t(X')\right)\left(\log f_t(X)-\log f_t(X')\right)\right] \d t \nonumber\\
&= \int^\infty_0 \cE\big(f_t(X), \log f_t(X)\big|Y\big) \d t, \label{eq:cond_ent_sym}
\end{align}
where $(X,X')$ is an exchangeable pair with joint law $P_{XX'} = \mu \otimes K^*K$, and in the last step we have used Eq.~\eqref{eq:exch_1}. From \eqref{eq:cond_ent_sym} and the definition of the log-Sobolev constant, it follows that
\begin{align}
	\E\left[\Ent\big[f(X)\big|Y\big]\right] &\ge \rho_1(\mu,K)\int^\infty_0 \Ent\left[f_t(X)\right] \d t.\label{eq:cond_ent_sym_2}
\end{align}
Now consider the function $\xi(u) \deq (u+1)\log(u+1) - u$, $u \ge -1$. This function is nonnegative, nonincreasing on $[-1,0]$, nondecreasing on $\Reals^+$, and
\begin{align*}
	\inf_{u \ge -1} \frac{\xi(u/2)}{\xi(u)} = \frac{1-\log 2}{2}.
\end{align*}
By monotonicity, $\xi(cu) \ge \xi(u/2) \ge \frac{1-\log 2}{2} \xi(u)$ for all $u \ge -1$ and for any $1/2 \le c \le 1$. Therefore,
\begin{align}
\int^\infty_0	\Ent\left[f_t(X)\right]\d t &= \int^\infty_0 \E\left[f_t(X)\log f_t(X) - f_t(X)+1 \right]  \d t \nonumber\\
 	&= \int^\infty_0 \E\left[\xi\left(f_t(X)-1\right)\right] \d t \nonumber\\
	& = \int^\infty_0 \E\left[ \xi\left(e^{-t}\left(f(X)-1\right)\right)\right] \d t \nonumber\\
	& \ge \int^{\log 2}_0 \E\left[ \xi\left(e^{-t}\left(f(X)-1\right)\right)\right] \d t \nonumber\\
	& \ge \frac{1-\log 2}{2} \int^{\log 2}_0 \E\left[\xi\left(f(X)-1\right)\right] \d t \nonumber\\
	& = \frac{(1-\log 2)\log 2}{2} \Ent\big[f(X)\big]. \label{eq:ent_integral}
\end{align}
Using \eqref{eq:ent_integral} in \eqref{eq:cond_ent_sym}, we obtain
\begin{align*}
	\E\left[ \Ent\big[f(X)\big|Y\big]\right] \ge \frac{(1-\log 2)\log 2}{2} \rho_1(\mu,K) \Ent\big[f(X)\big].
\end{align*}
Since $f$ was arbitrary, this implies the second inequality in \eqref{eq:SDPI_vs_logSob}.
\end{proof}

Now let us see how these results are related to the standard formulation of log-Sobolev inequalities in a discrete setting (see, e.g., \cite{Diaconis_Saloff_logsob,Bobkov_Tetali_logsob,Mossel_reverse_hypercont}). Given a finite set $\sX$, we fix an admissible pair $(\mu,M) \in \Prob(\sX) \times \Chan(\sX|\sX)$, such that the Markov kernel $M$ is \textit{reversible} w.r.t.\ $\mu$:
\begin{align}\label{eq:reversibility}
	\mu(x)M(x'|x) &= \mu(x')M(x|x'), \qquad \forall x,x' \in \sX
\end{align}
(nonreversible kernels can be handled as well, but we will not need this generalization here). From \eqref{eq:reversibility}, it follows that $M$ leaves $\mu$ invariant: $\mu M = \mu$. Define the \textit{Dirichlet form} $\cE_{\mu,M} : \Func(\sX) \times \Func(\sX) \to \Reals$ by
\begin{align}
	\cE_{\mu,M}(f,g) &\deq \frac{1}{2}\sum_{x,x' \in \sX} \left(f(x)-f(x')\right)\left(g(x)-g(x')\right)\mu(x)M(x'|x) \nonumber\\
	&\equiv \frac{1}{2}\E\left[ \left(f(X)-f(X')\right)\left(g(X)-g(X')\right)\right], \label{eq:Dirichlet_2}
\end{align}
where $(X,X') \in \sX \times \sX$ is a random pair with probability law $\mu \otimes M$. Our ``overloading'' of the notation $\cE(\cdot,\cdot)$ [compare with Eq.~\eqref{eq:Dirichlet}] is not accidental. To see this, we first need a definition:

\begin{definition} Fix some alphabet $\sY$ and a channel $K \in \Chan(\sY|\sX)$. We say that the pair $(\mu,M)$ {\em factors through $K$} if $M = K^*_\mu \circ K$, i.e., if
	\begin{align*}
		M(x'|x) &= \sum_{y \in \sY} K^*_\mu(x'|y) K(y|x), \qquad \forall (x,x') \in \sX \times \sX.
	\end{align*}
\end{definition}
\noindent In other words, $(\mu,M)$ factors through $K$ if we can generate a copy of $(X,X')$ according to the following two-stage procedure, starting with a draw $X \sim \mu$:
\begin{enumerate}
	\item Pass $X$ through the channel $K$ to get $Y$.
	\item Pass $Y$ through the adjoint channel $K^*_\mu$ to get $X'$.
\end{enumerate}
This is nothing but the well-known \textit{two-stage} (or \textit{two-component}) \textit{Gibbs sampler} \cite{Robert_Casella_MCMC,Diaconis_stoch_alt_proj}. 

\begin{proposition}\label{prop:Dirichlet_equiv} The random variables $X$ and $X'$ form an exchangeable pair. Moreover, if $(\mu,M)$ factors through some channel $K \in \Chan(\sY|\sX)$, then
	\begin{align}\label{eq:Dirichlet_equiv}
		\cE_{\mu,M}(f,g) = \cE\left(f(X),g(X)\big|Y\right), \qquad \forall f,g \in \Func(\sX)
	\end{align}
	where $(X,Y) \in \sX \times \sY$ is a random pair with law $\mu \otimes K$.
\end{proposition}
\begin{proof} Exchangeability of $(X,X')$ follows from the reversibility condition \eqref{eq:reversibility}. The identity \eqref{eq:Dirichlet_equiv} is a consequence of \eqref{eq:Dirichlet_2} and Prop.~\ref{prop:E}, Part 4).
\end{proof}

With these definitions out of the way, we can introduce the hierarchy of log-Sobolev inequalities following Mossel et al.~\cite{Mossel_reverse_hypercont}:

\begin{definition}\label{def:LSI} The pair $(\mu,M)$ satisfies {\em log-Sobolev inequality of order $p \in \Reals^+ \backslash \{0,1\}$ with constant $c$}, or $\LSI_p(c)$, if
	\begin{align*}
		\Ent\left[f^p(X)\right] \le \frac{cp^2}{4(p-1)} \cE_{\mu,M}\left(f^{p-1},f\right), \quad \forall f \in \PFunc(\sX);
	\end{align*}
$\LSI_1(c)$ if
$$
\Ent\left[f(X)\right] \le \dfrac{c}{4}\cE_{\mu,M}(f,\log f), \qquad \forall f \in \SPFunc(\sX);
$$
and $\LSI_0(c)$ if
$$
\Var[\log f(X)] \le - \frac{c}{2}\cE_{\mu,M}(f,1/f), \qquad \forall f \in \SPFunc(\sX).
$$
\end{definition}
\noindent Another important functional inequality relates the variance to the Dirichlet form $\cE_{\mu,M}$:

\begin{definition} $(\mu,M)$ satisfies a {\em Poincar\'e inequality} with constant $c \ge 0$, or $\PI(c)$, if 
	\begin{align*}
		\Var[f(X)] \le c\, \cE_{\mu,M}(f,f), \qquad \forall f \in \PFunc(\sX).
	\end{align*}
\end{definition}

\noindent We are interested in the tightest constants in log-Sobolev inequalities for $p \in [0,2]$. With that in mind, we define
\begin{align*}
		\tilde{\rho}_p(\mu,M) \deq
		\frac{p^2}{4(p-1)}\inf_{f \in \SPFunc(\sX)}  \frac{\cE_{\mu,M}(f^{p-1},f)}{\Ent[f^p(X)]}
	\end{align*}
	for $p \not\in \{0,1\}$, with the convention $\frac{0}{0} = \infty$. The constants $\tilde{\rho}_0,\tilde{\rho}_1$ are defined analogously. The {\em Poincar\'e constant} is
\begin{align*}
	\tilde{\lambda}(\mu,M) \deq \inf_{f \in \Func(\sX)} \frac{\cE_{\mu,M}(f,f)}{\Var[f(X)]}.
\end{align*}
Mossel et al.~\cite{Mossel_reverse_hypercont} proved that the function $p \mapsto \tilde{\rho}_p(\mu,M)$ is nonincreasing:
	\begin{align}\label{eq:logSob_monotonicity}
		\tilde{\rho}_0(\mu,M) \ge \tilde{\rho}_p(\mu,M) \ge \tilde{\rho}_q(\mu,M), \qquad 0 \le p \le q \le 2
	\end{align}
and moreover $\tilde{\rho}_0(\mu,M) = \frac{1}{2}\tilde{\lambda}(\mu,M)$. Log-Sobolev and Poincar\'e inequalities arise naturally in the study of the continuous-time random walk on $\sX$ with infinitesimal generator $L = M-I$. This is a pure-jump Markov process with state space $\sX$ that jumps from state $x$ to another state $x'$ with probability $M(x'|x)$, and the times between successive jumps are i.i.d.\ $\text{Exp}(1)$ random variables. Let $\{X_t\}_{t \ge 0}$ denote this process with $X_0 \sim \mu$, where $X_t \sim \mu$ for all $t$ by stationarity. For each $t \ge 0$, define the mapping $P_t : \Func(\sX) \to \Func(\sX)$ by 
$$
P_t f(x) \deq \E[f(X_t)|X_0 = x].
$$
Then one can prove the following (see, e.g., \cite[Prop.~1.7]{Montenegro_Tetali}):
\begin{enumerate}
	\item $\Var[P_tf(X_0)] \le e^{-t/c}\Var[f(X_0)]$ for all $f \in \Func(\sX)$ and all $t \ge 0$ if and only if the pair $(\mu,M)$ satisfies $\PI(c)$.
	\item $\Ent[P_tf(X_0)] \le e^{-t/c}\Ent[f(X_0)]$ for all $f \in \PFunc(\sX)$ and all $t \ge 0$ if and only if the pair $(\mu,M)$ satisfies $\LSI_1(4c)$.
\end{enumerate}
In other words, the Poincare inequality and the log-Sobolev inequality for $p=1$ completely characterize the exponential rate of decay of variance and entropy, respectively, along the trajectory of $\{X_t\}$ with $X_0 \sim \mu$. In particular, if for each $t \ge 0$ we consider the channel $M_t \in \Chan(\sX|\sX)$ with transition probabilities $M_t(x'|x) = \PP \left(X_t = x' | X_0 = x\right)$, then
\begin{align*}
	\eta_{\chi^2}(\mu,M_t) \le e^{- \tilde{\lambda}(\mu,M)t} \qquad \text{and} \qquad \eta(\mu,M_t) \le e^{-4\tilde{\rho}_1(\mu,M)t}.
\end{align*}
The main utility of the log-Sobolev inequality for $p=2$ is that the Dirichlet form $\cE_{\mu,M}(f,f)$ is much easier to deal with than $\cE_{\mu,M}(f,\log f)$; by monotonicity property of the log-Sobolev constants [cf.~\eqref{eq:logSob_monotonicity}], we end up with the handy estimate
\begin{align*}
	\eta(\mu,M_t) \le e^{-4\tilde{\rho}_2(\mu,M)t}.
\end{align*}
Thus, it is important to obtain tight upper and lower bounds on the Poincar\'e and the log-Sobolev constants of the pair $(\mu,M)$. We now show that such bounds can be given in terms of the SDPI constant $\eta(\mu,K)$ of any channel $K$ that the pair $(\mu,M)$ factors through; conversely, we can obtain bounds on $\eta(\mu,K)$ in terms of log-Sobolev and Poincar\'e constants of the pair $(\mu,K^*_\mu \circ K)$. We start with the Poincar\'e constant, in which case we have the following exact characterization:

\begin{theorem}\label{thm:SDPI_vs_Poincare} The functional $K \mapsto S^2(\mu,K)$ is constant on the collection
	\begin{align*}
		\Chan(\mu,M) &\deq \left\{ K : \text{ $(\mu,M)$ factors through $K$} \right\},
	\end{align*}
	and its value there is equal to $1-\tilde{\lambda}(\mu,M)$. Equivalently, if $K \in \Chan(\mu,M)$, then
	\begin{align*}
		\eta_{\chi^2}(\mu,K) = 1-\tilde{\lambda}(\mu,M).
	\end{align*}
\end{theorem}
\begin{proof} We need to show the following: if $(\mu,M)$ factors through $K$, then
	\begin{align}\label{eq:variance_drop}
		\cE_{\mu,M}(f,f) &= \Var\left[f(X)\right] - \Var\left[K^*_\mu f(Y)\right],
	\end{align}
	where $(X,Y) \in \sX \times \sY$ is a random pair with law $\mu \otimes K$. Assuming this is true, we then have
	\begin{align*}
			\tilde{\lambda}(\mu,M) &= \inf_{f \in \Func(\sX)} \frac{\cE_{\mu,M}(f,f)}{\Var[f(X)]} \\
			&= \inf_{f \in \Func(\sX)} \frac{\Var\left[f(X)\right] - \Var\left[K^*_\mu f(Y)\right]}{\Var\left[f(X)\right]} \\
			&= 1 - \eta_{\chi^2}(\mu,K) \\
			&= 1 - S^2(\mu,K).
	\end{align*}
	Noting that $\tilde{\lambda}(\mu,M)$ is independent of the choice of $K$, we obtain the statement of the theorem. 
	
	It remains to prove \eqref{eq:variance_drop}. Fixing $K$, we have
	\begin{align*}
		\Var\left[f(X)\right] - \Var\left[K^*_\mu f(Y)\right] &= \E[f^2(X)] - \E\left[(\E[f(X)|Y])^2\right],
	\end{align*}
	where
	\begin{align*}
		\E\left[(\E[f(X)|Y])^2\right] &= \sum_{x,x' \in \sX}\sum_{y \in \sY} \mu K(y) K^*_\mu(x|y) K^*_\mu(x'|y) f(x) f(x') \\
		&= \sum_{x,x' \in \sX}\sum_{y \in \sY} \mu(x) K(y|x) K^*_\mu(x'|y) f(x) f(x') \\
		&= \sum_{x,x' \in \sX} \mu(x) M(x'|x) f(x) f(x') \\
		&= \E[f(X)f(X')].
	\end{align*}
	Therefore,
	\begin{align*}
		\cE_{\mu,M}(f,f) &= \frac{1}{2} \E\left[\left(f(X)-f(X')\right)^2\right] \\
		&= \E[f^2(X)] - \E[f(X)f(X')] \\
		&= \E[f^2(X)] - \E\left[(\E[f(X)|Y])^2\right] \\
		&= \Var\left[f(X)\right] - \Var\left[K^*_\mu f(Y)\right].
	\end{align*}
\end{proof}

\begin{example}[Doubly symmetric binary source]\label{ex:DSBS} {\em Consider the case $\sX = \{0,1\}$, $\mu = \Bernoulli(1/2)$, $M = \BSC(\eps)$ with $\eps \le 1/2$. The resulting exchangeable pair $(X,X')$ is the \textit{doubly symmetric binary source} (DSBS) with parameter $\eps$ \cite{Wyner_common_info}. It is a matter of simple computation to show that the pair $(\mu,M)$ factors through $K = \BSC(\delta(\eps))$ with
	\begin{align}
	\label{eq:BSC_factor}
	\delta(\eps) =  \frac{1 + \sqrt{1 - 2\eps}}{2}.
\end{align}
	We know that
	\begin{align*}
		S^2\left(\Bernoulli(1/2),\BSC(\delta(\eps))\right) &= \left(1-2\delta(\eps)\right)^2 = 1-2\eps,
	\end{align*}
	which therefore gives
	\begin{align*}
		\tilde{\lambda}\left(\Bernoulli(1/2),\BSC(\eps)\right) = 2\eps.
	\end{align*}
For any $f \in \Func(\sX)$, we can compute the Dirichlet form
	\begin{align*}
		\cE_{\mu,M}(f,f) &= \frac{1}{2} \left[ \mu (0) M(1|0) + \mu(1) M(0|1)\right]\big(f(0) - f(1)\big)^2 \\
		&= \frac{\eps}{2} \big(f(0)-f(1)\big)^2,
	\end{align*}
	which gives us the Poincar\'e inequality
	\begin{align*}
		\Var[f(X)] &\le \frac{1}{4}\big(f(0)-f(1)\big)^2
	\end{align*}
	(see, e.g., \cite[Ex.~3.9]{Bobkov_Tetali_logsob}). Note that this inequality is independent of the crossover probability $\eps$.}
\end{example}
Next, we consider the case of the log-Sobolev constant $\tilde{\rho}_1(\mu,M)$, for which we can only give upper and lower bounds:

\begin{theorem}\label{thm:SDPI_vs_logSob_2} The functional $K \mapsto \rho_1(\mu,K)$ is constant on the collection of all channels $K$ such that $M = K^*_\mu \circ K$, where it takes the value $4\tilde{\rho}_1(\mu,M)$. Moreover, if $(\mu,M)$ factors through $K$, the log-Sobolev constant $\tilde{\rho}_1(\mu,M)$ satisfies
	\begin{align}\label{eq:SDPI_vs_logSob_2}
		1-\eta(\mu,K) \le 4\tilde{\rho}_1(\mu,M) \le 1-\frac{(1-\log 2)\log 2}{2}\eta(\mu,K).
	\end{align}
\end{theorem}

\begin{proof} Choose any channel $K$, such that $M = K^*_\mu \circ K$. Then, with $(X,Y) \sim \mu \otimes K$,
	\begin{align}
		1-\eta(\mu,K) &\le \rho_1(\mu,K) \label{eq:logsob_comp_1} \\
		&=\inf_{f \in \SPFunc(\sX)} \frac{\cE\big(f(X),\log f(X)\big | Y \big)}{\Ent\left[f(X)\right]} \label{eq:logsob_comp_2}\\
		&= \inf_{f \in \SPFunc(\sX)} \frac{\cE_{\mu,M}(f, \log f)}{\Ent[f(X)]} \label{eq:logsob_comp_3}\\
		&= 4 \tilde{\rho}_1(\mu,M), \label{eq:logsob_comp_4}
	\end{align}
	where \eqref{eq:logsob_comp_1} is by Proposition~\ref{prop:SDPI_vs_logSob}; \eqref{eq:logsob_comp_2} is by \eqref{eq:logSob_constant}; \eqref{eq:logsob_comp_3} is by Proposition~\ref{prop:Dirichlet_equiv}; \eqref{eq:logsob_comp_4} is by definition of $\tilde{\rho}_1(\mu,M)$.

	This proves the first inequality in \eqref{eq:SDPI_vs_logSob}. The second inequality follows from the upper bound on $\rho_1(\mu,K)$ in Proposition~\ref{prop:SDPI_vs_logSob}, as well as from the just established fact that $\rho_1(\mu,K) = 4\tilde{\rho}_1(\mu,M)$.
\end{proof}

\begin{example}[Doubly symmetric binary source, continued] {\em Consider again the case of the DSBS with parameter $\eps \le 1/2$. From the previous example, we know that $(\mu,M)$ factors through $K = \BSC(\delta(\eps))$ with crossover probability $\delta(\eps)$ given by \eqref{eq:BSC_factor}. For this channel, we have
	\begin{align*}
		\eta\big(\Bernoulli(1/2),\BSC(\delta(\eps))\big)	= 1-2\eps.
	\end{align*}
Applying this and Theorem~\ref{thm:SDPI_vs_logSob_2}, we get the following upper and lower bounds on the log-Sobolev constant $\tilde{\rho}_1$:
	\begin{align*}
		\frac{\eps}{2} \le  \tilde{\rho}_1\big( \Bernoulli(1/2), \BSC(\eps) \big) \le \frac{1}{4}\left[1-\frac{(1-\log 2)\log 2}{2}(1-2\eps)\right].
	\end{align*}
Unfortunately, neither of the bounds is tight, since the log-Sobolev constant in this case is known exactly: $\tilde{\rho}_1\big(\Bernoulli(1/2),\BSC(\eps)\big) = \eps$ \cite[Ex.~3.9]{Bobkov_Tetali_logsob}. A sharp bound can be obtained from the monotonicity property \eqref{eq:logSob_monotonicity} of the log-Sobolev constants:
\begin{align*}
	\tilde{\rho}_1\big(\Bernoulli(1/2),\BSC(\eps)\big) &\le \tilde{\rho}_0\big(\Bernoulli(1/2),\BSC(\eps)\big) \\
	&=\frac{1}{2}\tilde{\lambda}\big(\Bernoulli(1/2),\BSC(\eps)\big) \\
	&= \eps.
\end{align*}
}
\end{example}

Finally, we consider the log-Sobolev constant $\tilde{\rho}_2(\mu,M)$:

\begin{theorem}\label{thm:strong_DPI_1} For any channel $K \in \Chan(\sY|\sX)$ such that $M = K^*_\mu \circ K$,
	\begin{align}\label{eq:SDPI_vs_2Sob}
		\eta(\mu,K) \le 1 - \tilde{\rho}_2(\mu,M).
	\end{align}
\end{theorem}
\begin{proof} We use the following delicate convexity bound for the function $\Phi(u) = u\log u$ \cite{Miclo_hypercontractive}:
	\begin{align}\label{eq:Miclo_bound}
		\Phi(u) \ge \Phi(v) + (1+\log v)(u-v) + \left(\sqrt{u} - \sqrt{v}\right)^2, \qquad \forall u,v \ge 0.
	\end{align}
Let $(X,Y)$ be a random pair with law $\mu \otimes K$. Fix any function $f \in \PFunc(\sX)$ with $\E[f(X)]=1$ and use the bound \eqref{eq:Miclo_bound} to get
\begin{align}
	\Phi\big(f(x)\big) &\ge \Phi\big(K^*_\mu f(y)\big)  + \Big(1+\log K^*_\mu f(y)\Big)\Big(f(x)-K^*_\mu f(y)\Big)  + \left(\sqrt{f(x)} - \sqrt{K^*_\mu f(y)}\right)^2
	\label{eq:Miclo_bound_fg}
\end{align}
Taking conditional expectation $\E[\cdot|Y]$ of both sides of \eqref{eq:Miclo_bound_fg}, we obtain
\begin{align*}
	\E\Big[\Phi\big(f(X)\big)\Big|Y\Big] &\ge \Phi\left(\E[f(X)|Y]\right) + \E\Big[\left(\sqrt{f(X)} - \sqrt{\E[f(X)|Y]}\right)^2\Big|Y\Big] \nonumber\\
	&\ge \Phi\left(\E[f(X)|Y]\right)  + \E\Big[\left(\sqrt{f(X)} - \E\Big[\sqrt{f(X)}\Big|Y\Big]\right)^2\Big|Y\Big],
\end{align*}
where we have used the fact that
\begin{align*}
\E\big[(U-\E[U|Y])^2\big|Y\big] = \inf_{f \in \Func(\sY)} \E\big[\left(U-f(Y)\right)^2\big|Y\big].
\end{align*}
for any real-valued random variable $U$ jointly distributed with $Y$. Next we take the expectation w.r.t.\ $Y$ to get
\begin{align*}
	\Ent[f(X)] &\ge \Ent\big[\E[f(X)|Y]\big] + \cE\Big(\sqrt{f(X)},\sqrt{f(X)} \Big| Y \Big) \\
	&= \Ent\left[\E[f(X)|Y]\right] + \cE_{\mu,M}\Big(\sqrt{f},\sqrt{f}  \Big) 
\end{align*}
where we have used the fact that $\Ent[U]=\E[\Phi(U)]$ for all nonnegative random variables $U$ with $\E U = 1$, as well as Proposition~\ref{prop:Dirichlet_equiv}. Using this and the definition of $\tilde{\rho}_2(\mu,M)$, we get
\begin{align*}
\Ent\big[\E[f(X)|Y]\big] \le \left(1-\tilde{\rho}_2(\mu,M)\right)\Ent\left[f(X)\right].
\end{align*}
Since $f$ was arbitrary, we get the bound \eqref{eq:SDPI_vs_2Sob}.
\end{proof}

\subsection{The gap between SDPI and $\Phi$-Sobolev}

As evident from the proof of Theorem~\ref{thm:SDPI_to_phiSob}, we need to invoke Jensen's inequality in order to pass from a $\Phi$-entropy SDPI to a $\Phi$-Sobolev inequality. This observation prompts us to investigate the gap between these two inequalities:

\begin{theorem}\label{thm:Jensen_gap} If $\Phi$ satisfies the conditions of Theorem~\ref{thm:SDPI_to_phiSob}, then for any $f \in \PFunc(\sX)$
	\begin{align}\label{eq:Dirichlet_vs_entropy_gap}
		\cE(f(X), \Psi \circ f(X)|Y) = \E\left[ \Ent_\Phi[f(X)|Y] \right] + \E\left[f(X)\Ent_{-\Psi}[f(X)|Y]\right]
	\end{align}
Therefore, if $\eta_\Phi(P_X,P_{Y|X}) \le c$ for some $0 \le c < 1$, and if the function $u \mapsto -\Psi(u)$ is strictly convex at $u=1$, then the $\Phi$-Sobolev inequality
\begin{align}\label{eq:PsiSob}
	\Ent_\Phi[f(X)] \le \frac{1}{1-c} \cE(f(X), \Psi \circ f(X)|Y)
\end{align}
is strict for any nonconstant $f$. If $\Psi$ is affine, then $\cE(f(X), \Psi \circ f(X)|Y) = \E\left[ \Ent_\Phi[f(X)|Y] \right]$, and in that case $\eta_\Phi(P_X,P_{Y|X}) \le c$ is equivalent to \eqref{eq:PsiSob}.
\end{theorem}
\begin{remark} {\em When $\Psi$ is affine,  $\Phi$ is of the form $\Phi(u) = au^2 + bu + c$ for some $a \ge 0$, $b,c \in \Reals$. Thus, the SDPI for $\chi^2$-divergence is equivalent to the corresponding $\Phi$-Sobolev inequality (which in this case is precisely the Poincar\'e inequality).\hfill$\diamond$}
\end{remark}
\begin{proof} By definition of $\cE$,
	\begin{align}
		&\cE(f(X),\Psi \circ f(X)|Y) \nonumber\\
		&\quad= \E\Big[f(X)\Psi(f(X))-\E[f(X)|Y]\E[\Psi(f(X)|Y]\Big] \nonumber\\
		&\quad= \E[f(X)\Psi(f(X))] - \E[\E[f(X)|Y]\Psi(\E[f(X)|Y])] \nonumber\\
		& \qquad \qquad + \E[\E[f(X)|Y]\Psi(\E[f(X)|Y])] -\E\left[\E[f(X)|Y]\E[\Psi(f(X)|Y]\right] \nonumber\\
		&\quad= \E\left[\Ent_\Phi[f(X)|Y]\right] + \E\Big[f(X)\big\{\E[-\Psi(f(X))|Y] - \left(-\Psi(\E[f(X)|Y])\right)\big\} \Big].\label{eq:Jensen_gap}
	\end{align}
Since $-\Psi$ is convex, we recognize the quantity in the curly braces in \eqref{eq:Jensen_gap} as the conditional entropy $\Ent_{-\Psi}[f(X)|Y]$. This proves \eqref{eq:Dirichlet_vs_entropy_gap}.

From \eqref{eq:Dirichlet_vs_entropy_gap} we see that $\cE(f(X),\Psi \circ f(X)|Y) = \E\left[\Ent_\Phi[f(X)|Y]\right]$ for a given nonconstant $f \in \PFunc(\sX)$ if and only if $\Ent_{-\Psi}[f(X)|Y] = 0$ a.s. If $-\Psi$ is strictly convex at $1$, then $\Ent_{-\Psi}[U]=0$ if and only if $U$ is a.s.\ constant; thus, in this case, the inequality \eqref{eq:PsiSob} is strict for any nonconstant $f$. If $\Psi$ is affine, then $\Ent_{-\Psi}[U] = 0$ for all $U$, so in that case $\cE(f(X),\Psi \circ f(X)|Y) = \E\left[\Ent_\Phi[f(X)|Y]\right]$.
\end{proof}

As a corollary, we obtain the following useful formula that expresses the covariance between $f(X)$ and $\Psi \circ f(X)$ in terms of entropies:

\begin{corollary}
	\begin{align}\label{eq:cov_vs_ent}
		\Cov[f(X),\Psi \circ f(X)] = \Ent_\Phi[f(X)] + \E\left[ f(X) \right] \Ent_{-\Psi}[f(X)].
	\end{align}
\end{corollary}
\begin{proof} Consider any pair $(X,Y)$, where $Y$ is independent of $X$. In that case, $\cE(f(X),g(X)|Y) = \Cov[f(X),g(X)]$ for any pair $f,g \in \Func(\sX)$, whereas $\Ent_\Phi[f(X)|Y] = \Ent_\Phi[f(X)]$ for any $\Phi \in \cF$. The formula \eqref{eq:cov_vs_ent} follows from these observations.
\end{proof}

\section{Some applications}
\label{sec:applications}

\subsection{Concentration inequalities}

One of the main uses of logarithmic Sobolev inequalities is in the context of concentration inequalities: Given a probability space $(\sX,\mu)$ and a function $f \in \Func(\sX)$, the objective is to obtain tight upper bounds on the deviation probabilities $\PP[f(X)-\E f(X) \ge t]$ for $t \ge 0$, where $X \sim \mu$. A general procedure that allows one to pass from a suitable log-Sobolev inequality to a Gaussian tail bound of the form
\begin{align}\label{eq:Gaussian_tail}
	\PP[f(X)-\E f(X) \ge t] \le e^{-\kappa t^2}, \qquad t \ge 0
\end{align}
for some $\kappa > 0$ and for all $f$ in a suitable subset of $\Func(\sX)$ is called the \textit{Herbst argument} \cite{Ledoux_conc_book,Boucheron_etal_concentration_book,Raginsky_Sason_FnT_2nd}, and can be summarized as follows (see, e.g., \cite[Chap.~3]{Raginsky_Sason_FnT_2nd}):

We start with a pair $(\cA,\Gamma)$, where:
\begin{enumerate}
	\item $\cA \subseteq \Func(\sX)$ is a class of real-valued functions on $\sX$, such that $a f + b \in \cA$ for all $f \in \cA$ and all $a \ge 0, b \in \Reals$.
	\item $\Gamma : \cA \to \PFunc(\sX)$ is an operator with the property that $\Gamma (af + b) = a\Gamma f$ for all $f \in \cA$ and all $a \ge 0, b \in \Reals$.
\end{enumerate}
We then say that $\mu$ satisfies a modified log-Sobolev inequality with constant $c > 0$ on $(\cA,\Gamma)$ if
\begin{align}\label{eq:mLSI}
	\Ent[e^{f(X)}] \le \frac{c}{2} \E\left[e^{f(X)} \left|\Gamma f(X)\right|^2\right], \qquad \forall f \in \cA.
\end{align}
Here is how we pass from \eqref{eq:mLSI} to a Gaussian tail bound of the form \eqref{eq:Gaussian_tail}. Without loss of generality, we may assume that $\E[f(X)] = 0$. For any $\lambda \ge 0$, $\lambda f \in \cA$ and $\Gamma (\lambda f) = \lambda \Gamma f$. Therefore, replacing $f$ with $\lambda f$ in \eqref{eq:mLSI}, we arrive at
\begin{align}
	\Ent\left[e^{\lambda f(X)}\right] &\le \frac{c\lambda^2}{2} \E\left[e^{\lambda f(X)} \left|\Gamma f(X)\right|^2\right] \nonumber \\
	&\le \frac{c \| \Gamma f \|^2_\infty \lambda^2}{2} \E\left[e^{\lambda f(X)}\right], \label{eq:Herbst_1}
\end{align}
where $\| \Gamma f \|_\infty \deq \sup_{x \in \sX} |\Gamma f(x)|$. If we define the tilted distribution $\d\mu^{(\lambda)} \deq e^{\lambda f} \d\mu / \E[e^{\lambda f(X)}]$, then
\begin{align*}
	D(\lambda) \deq D\big(\mu^{(\lambda)} \big\| \mu \big) = \frac{\Ent\left[e^{\lambda f(X)}\right]}{\E\left[e^{\lambda f(X)}\right]}.
\end{align*}
Therefore, from \eqref{eq:Herbst_1} we get
\begin{align*}
	D(\lambda) \le \frac{c \| \Gamma f \|^2_\infty \lambda^2}{2}, \qquad \forall \lambda \ge 0.
\end{align*}
On the other hand, if we define the logarithmic moment-generating function $\Lambda(\lambda) \deq \log \E[e^{\lambda f(X)}]$, then it is a matter of simple calculus to show that
\begin{align}\label{eq:Herbst_2}
	D(\lambda) = \lambda^2 \frac{\d}{\d \lambda}\left( \frac{\Lambda(\lambda)}{\lambda}\right).
\end{align}
Combining \eqref{eq:Herbst_1} and \eqref{eq:Herbst_2}, we get the differential inequality
\begin{align*}
	\frac{\d}{\d\lambda} \left( \frac{\Lambda(\lambda)}{\lambda}\right) \le \frac{c \| \Gamma f \|^2_\infty}{2},
\end{align*}
which can be integrated to give $\Lambda(\lambda) \le \frac{c \| \Gamma f \|^2_\infty \lambda^2}{2}$. This shows that $f$ is $v$-subgaussian with $v = c \| \Gamma f \|^2_\infty$, and therefore it satisfies \eqref{eq:Gaussian_tail} with $\kappa = 1/2v = 1/(2c \| \Gamma f \|^2_\infty)$ (cf.~Section~\ref{ssec:info_transport_bounds}). Effectively, $\| \Gamma f \|_\infty$ is a measure of the ``variability'' of $f$.

We now show that we can use any reversible Markov kernel $M$ on $\sX$ that leaves $\mu$ invariant as a yardstick for measuring the variability of functions in $\Func(\sX)$, and that the constant $c$ in the log-Sobolev inequality \eqref{eq:mLSI} can be expressed in terms of the relative-entropy SDPI constants $\eta(\mu,K)$, where $K$ runs over all factorizations $M = K^*_\mu K$. Following Houdr\'e and Tetali \cite{Houdre_Tetali_product_graphs}, let us define the $\ell_2$ positive discrete gradient operator $\nabla^+_2 : \Func(\sX) \to \PFunc(\sX)$ via
\begin{align*}
	\nabla^+_2 f(x) \deq \left( \sum_{x' \in \sX} M(x'|x) \left(f(x)-f(x')\right)^2_+\right)^{1/2}.
\end{align*}
It is easy to see that the pair $(\cA,\Gamma) = (\Func(\sX), \nabla^+_2)$ satisfies the requirements 1 and 2 listed in the preceding paragraph.

\begin{theorem} Consider a pair $(\mu,M) \in \PProb(\sX) \times \Chan(\sY|\sX)$, where $M$ is reversible w.r.t.\ $\mu$. Then the following modified log-Sobolev inequality holds for all $f \in \Func(\sX)$:
\begin{align*}
	\Ent\left[e^{f(X)}\right] \le \frac{c}{2}\E\left[e^{f(X)} \left|\nabla^+_2 f(X)\right|^2 \right],
\end{align*}
where $X \sim \mu$, and
	\begin{align}\label{eq:M_mLSI_const}
		c = \inf \left\{ \frac{2}{1-\eta(\mu,K)} : M = K^*_\mu K \right\}.
	\end{align}
\end{theorem}
\begin{proof} Suppose that $M$ factors through some channel $K \in \Chan(\sY|\sX)$. As before, let $(X,X')$ be an exchangeable pair with law $\mu \otimes K^*K \equiv \mu \otimes M$. Applying Theorem~\ref{thm:composite_fctn_phiSob} to $\Phi(u) = u \log u$ and $F(u) = e^u$, we conclude that any $f \in \Func(\sX)$ satisfies
	\begin{align*}
		\Ent\left[e^{f(X)}\right] &\le \frac{1}{1-\eta(\mu,K)} \E\left[e^{f(X)}\left(f(X)-f(X')\right)^2_+\right] \\
		&= \frac{1}{1-\eta(\mu,K)}\sum_{x \in \sX} \mu(x) e^{f(x)} \sum_{x' \in \sX} K^* K(x'|x) \left(f(x)-f(x')\right)^2_+ \\
		&= \frac{1}{1-\eta(\mu,K)}\sum_{x \in \sX} \mu(x) e^{f(x)} \sum_{x' \in \sX} M(x'|x) \left(f(x)-f(x')\right)^2_+ \\	
		&= \frac{1}{1-\eta(\mu,K)}\sum_{x \in \sX} \mu(x) e^{f(x)} \left|\nabla^+_2 f(x)\right|^2 \\
		&= \frac{1}{1-\eta(\mu,K)}\E\left[e^{f(X)}|\nabla^+_2 f(X)|^2\right].
	\end{align*}
Optimizing over the choice of $K$, we see that \eqref{eq:mLSI} holds with $c$ given by \eqref{eq:M_mLSI_const}.
\end{proof}

\subsection{Contraction of mutual information in a Markov chain}

Consider a Markov chain $U \to X \to Y$, where the joint law $P_{XY}$ is fixed, while the alphabet $\sU$ of $U$ and the conditional distribution $P_{U|X}$ are allowed to vary arbitrarily. By the data processing inequality for the mutual information, $I(U;Y) \le I(U;X)$ for any choice of $P_{U|X}$. The question is: what is the maximum value of the ratio $\frac{I(U;Y)}{I(U;X)}$ that can be achieved by any choice of $P_{U|X}$? The following claim was made by Erkip and Cover \cite{Erkip_Cover}:
\begin{align}\label{eq:Erkip_Cover}
	\sup_{P_{U|X}} \frac{I(U;Y)}{I(U;X)} = S^2(P_X,P_{Y|X}),
\end{align}
where $S^2$ is the squared maximal correlation (see Section~\ref{ssec:maxcorr}). However, Anantharam et al.\ in a recent preprint \cite{Anantharam_etal_HGR} pointed out a flaw in the proof of \eqref{eq:Erkip_Cover}, and showed instead that
\begin{align}\label{eq:Anantharam}
	\sup_{P_{U|X}} \frac{I(U;Y)}{I(U;X)} = \eta(P_X,P_{Y|X}),
\end{align}
where $\eta$ is the relative-entropy SDPI constant. Moreover, they provided an explicit example of a source-channel pair $(P_X,P_{Y|X})$, for which the mutual-information ratio on the left-hand sides of Eqs.~\eqref{eq:Erkip_Cover} and \eqref{eq:Anantharam} is strictly larger than $S^2(P_X,P_{Y|X})$. 

We will now present a generalization of the result of Anantharam et al., and show, as a consequence, that $S^2(P_X,P_{Y|X})$ can indeed be expressed as a supremum of the ratio of two information-like quantities pertaining to the Markov chain $U \to X \to Y$ with an arbitrary choice of $P_{U|X}$. Fix a function $\Phi \in \cF$. Given a random pair $(U,V)$, we define the \textit{mutual $\Phi$-information} \cite{Cohen_etal_book}\footnote{Palomar and Verd\'u \cite{Palomar_Verdu_lautum} define $\Phi$-information between $U$ and $V$ as $D_\Phi(P_U \otimes P_V \| P_{UV})$. Their definition is equivalent to Eq.~\eqref{eq:f_info} if we replace $\Phi$ with its \textit{Csisz\'ar conjugate} $\Phi^\star(u) \deq u \Phi(1/u)$ \cite{LieseVajda06}.} as
\begin{align}
	I_\Phi(U;V) &\deq D_\Phi\big(P_{UV} \big\| P_{U} \otimes P_{V} \big) \nonumber\\
 &= \int P_U(\d u) \int P_{V}(\d v) \Phi \left( \frac{\d P_{V|U=u}}{\d P_V}(v)\right) \nonumber\\
	&= \int P_U(\d u) D_\Phi\big(P_{V|U=u} \big \| P_V\big). \label{eq:f_info}
\end{align}
If $U$ and $V$ are related via a Markov kernel $K$ (i.e., $P_{UV} = P_U \otimes K$), we may also use the notation $I_\Phi(P_U,K)$ to indicate the fact that the $\Phi$-information is a functional of the source distribution and the kernel that generates the random output given the input.
\begin{theorem}\label{thm:gen_Erkip_Cover} If $\Phi \in \cF$ is differentiable, and its derivative is uniformly bounded in some neighborhood of $1$, then
	\begin{align*}
		\sup_{P_{U|X}} \frac{I_\Phi(U;Y)}{I_\Phi(U;X)} = \eta_\Phi(P_X,P_{Y|X}).
	\end{align*}
\end{theorem}
\begin{proof} Define a probability measure $Q \in \Prob(\sU)$ by
	\begin{align*}
		Q(u) \deq \frac{P_U(u) D_\Phi\big(P_{X|U=u}\big\| P_X\big)}{\sum_{u \in \sU}P_U(u) D_\Phi\big(P_{X|U=u} \big\| P_X\big)}.
	\end{align*}
This measure is supported on the set $\tilde{\sU} \deq \{ u \in \sU: D_\Phi(P_{X|U=u} \| P_X) > 0 \}$. From data processing, we have the inclusion $\{ u \in \sU: D_\Phi(P_{Y|U=u} \| P_Y) > 0 \} \subseteq \tilde{\sU}$.  Taking all of this into account, we can write
	\begin{align*}
		\frac{I_\Phi(U;Y)}{I_\Phi(U;X)} 
		&= \frac{\sum_{u \in \tilde{\sU}} P_U(u) D_\Phi\big(P_{Y|U=u} \big\| P_Y\big)}{\sum_{u \in \tilde{\sU}} P_U(u) D_\Phi\big(P_{X|U=u} \big\| P_X \big)} \\
		&= \sum_{u \in \tilde{\sU}} Q(u) \frac{D_\Phi\big(P_{Y|U=u} \big\| P_Y\big)}{D_\Phi\big(P_{X|U=u} \big\| P_X\big)} \\
		&\le \max_{u \in \tilde{\sU}} \frac{D_\Phi\big(P_{Y|U=u} \big\| P_Y\big)}{D_\Phi\big(P_{X|U=u}\big\| P_X\big)} \\
		&= \max_{u \in \tilde{\sU}} \frac{D_\Phi\big(P_{X|U=u} P_{Y|X}  \big\| P_X P_{Y|X} \big)}{D_\Phi\big(P_{X|U=u} \big\| P_X\big)} \\
		&\le \eta_\Phi(P_X,P_{Y|X}).
	\end{align*}
To prove the reverse inequality, we adopt the construction from \cite{Anantharam_etal_HGR}. Fix an arbitrary $Q_X \in \Prob(\sX)$. For any $\eps \in (0,1)$ small enough so that $\nu \deq P_X - \eps Q_X$ is a nonnegative measure, let $P^{(\eps)}_U = \Bernoulli(\eps)$ and define $P^{(\eps)}_{X|U}$ by
\begin{align*}
	P^{(\eps)}_{X|U=0} = Q_X, \qquad P^{(\eps)}_{X|U=1} = \frac{\nu}{\bar{\eps}}.
\end{align*}
With these choices, $P^{(\eps)}_U P^{(\eps)}_{X|U} = \eps P^{(\eps)}_{X|U=0} + \bar{\eps} P^{(\eps)}_{X|U=1} = \eps Q_X + P_X - \eps P_X = P_X$. For any $\eta > 0$, define the function
\begin{align*}
	L_\eta(\eps) \deq I_\Phi\left(P^{(\eps)}_U, P_{Y|X} \circ P^{(\eps)}_{X|U}\right) - \eta I_\Phi\left(P^{(\eps)}_U, P^{(\eps)}_{X|U}\right).
\end{align*}
A simple calculation gives
\begin{align*}
	I_\Phi\left(P^{(\eps)}_U, P_{Y|X} \circ P^{(\eps)}_{X|U}\right) &= \eps D_\Phi\big( P^{(\eps)}_{X|U=0} P_{Y|X} \big\| P_Y\big) + \bar{\eps} D_\Phi\big( P^{(\eps)}_{X|U=1} P_{Y|X} \big\| P_Y\big) \\
	&= \eps D_\Phi\big(Q_X P_{Y|X} \big\| P_Y \big) + \bar{\eps} D_\Phi\Big( \frac{P_X - \eps Q_X}{\bar{\eps}}P_{Y|X}\Big\| P_Y\Big) \\
	&= \eps D_\Phi\big(Q_X P_{Y|X} \big\| P_Y \big) + \bar{\eps} D_\Phi\Big( \frac{P_Y - \eps Q_X P_{Y|X} }{\bar{\eps}}\Big\| P_Y\Big),
\end{align*}
where in the last line we have used the fact that any Markov kernel extends to a linear map on signed measures. Similarly,
\begin{align*}
	I_\Phi\left(P^{(\eps)}_U, P^{(\eps)}_{X|U}\right) &= \eps D_\Phi\big(P^{(\eps)}_{X|U=0} \big\| P_X\big) + \bar{\eps} D_\Phi\big(P^{(\eps)}_{X|U=1} \big\| P_X \big) \\
	&= \eps D_\Phi\big( Q_X \big\| P_X\big) + \bar{\eps} D_\Phi\Big( \frac{P_X - \eps Q_X}{\bar{\eps}}\Big\| P_X \Big).
\end{align*}
Let $f = \frac{\d Q_X}{\d P_X}$ and $g^{(\eps)} = \frac{1-\eps f}{\bar{\eps}}$. Then, by virtue of our choice of $\eps$, $g^{(\eps)} \in \PFunc(\sX)$, and $\E[g^{(\eps)}(X)] = 1$. With these definitions, we can rewrite the above expressions as
\begin{align*}
	I_\Phi \left(P^{(\eps)}_U, P_{Y|X} \circ P^{(\eps)}_{X|U}\right) &= \eps \Ent_\Phi \left[ P^*_{Y|X} f(Y)\right] + \bar{\eps} \Ent_\Phi \left[P^*_{Y|X} g^{(\eps)}(Y)\right]
\end{align*}
and
\begin{align*}
	I_\Phi\left(P^{(\eps)}_U, P^{(\eps)}_{X|U}\right) = \eps \Ent_\Phi \left[ f(X)\right] + \bar{\eps} \Ent_\Phi \left[ g^{(\eps)}(X)\right].
\end{align*}
Consequently,
\begin{align*}
	\frac{\d}{\d\eps} L_\eta(\eps) &= \Ent_\Phi\left[ P^*_{Y|X} f(Y)\right] - \eta \Ent_\Phi \left[ f(X)\right] + \frac{\d}{\d\eps} \left\{ \bar{\eps} \left(\Ent_\Phi \left[P^*_{Y|X} g^{(\eps)}(Y)\right] -  \eta\Ent_\Phi \left[ g^{(\eps)}(X)\right]\right)\right\} \nonumber \\
	&= \Ent_\Phi\left[ P^*_{Y|X} f(Y)\right] - \eta \Ent_\Phi \left[ f(X)\right] + \eta \Ent_\Phi \left[ g^{(\eps)}(X)\right] - \Ent_\Phi \left[P^*_{Y|X} g^{(\eps)}(Y)\right] \nonumber \\
	& \qquad \qquad + \bar{\eps}\frac{\d}{\d \eps} \left\{ \Ent_\Phi \left[P^*_{Y|X} g^{(\eps)}(Y)\right] -  \eta\Ent_\Phi \left[ g^{(\eps)}(X)\right]\right\}.
\end{align*}
Now let us choose $Q_X$ so that $\frac{D_\Phi(Q_Y \| P_Y)}{D_\Phi(Q_X \| P_X)} > \eta_\Phi(P_X,P_{Y|X}) - \delta$ for some small $\delta > 0$. Then, for any $\eta < \eta_\Phi(P_X,P_{Y|X}) - \delta$ we have
\begin{align*}
	\frac{\d}{\d\eps} L_\eta(\eps) \Big|_{\eps = 0} &=  \Ent_\Phi\left[ P^*_{Y|X} f(Y)\right] - \eta \Ent_\Phi \left[ f(X)\right] > 0,
\end{align*}
where we have used Lemma~\ref{lm:entropy_derivative} in Appendix~\ref{app:lemmas}, and where
 the strict inequality holds due to our choice of $\eta$. Thus, the function $\eps \mapsto L_\eta(\eps)$ is strictly increasing in some neighborhood of $0$. Since $L_\eta(0) = 0$, there exists some value $\eps_0 > 0$, such that $L_\eta(\eps_0) > 0$, i.e.,
\begin{align*}
	\sup_{P_{U|X}} \frac{I_\Phi(U;Y)}{I_\Phi(U;X)} \ge \frac{I_\Phi\left(P^{(\eps_0)}_U, P_{Y|X} \circ P^{(\eps_0)}_{X|U}\right)}{I_\Phi\left(P^{(\eps_0)}_U, P^{(\eps_0)}_{X|U}\right)} > \eta.
\end{align*}
Since this holds for all $0 < \eta < \eta_\Phi(P_X,P_{Y|X}) - \delta$, and $\delta > 0$ was arbitrary, we conclude, upon taking $\delta \searrow 0$, that
\begin{align*}
	\sup_{P_{U|X}} \frac{I_\Phi(U;Y)}{I_\Phi(U;X)} \ge \eta_\Phi(P_X,P_{Y|X}).
\end{align*}
Since we already established the reverse inequality, the theorem is proved.
\end{proof}
Thus, if $\Phi(u) = u\log u$, we recover the result of Anantharam et al.~\cite{Anantharam_etal_HGR}; on the other hand, choosing $\Phi(u) = (u-1)^2$, we can express the squared maximal correlation $S^2(P_X,P_{Y|X})$ as
\begin{align*}
	S^2(P_X,P_{Y|X}) = \sup_{P_{U|X}} \frac{I_{\chi^2}(U; Y)}{I_{\chi^2}(U; X)},
\end{align*}
where the \textit{$\chi^2$-information} $I_{\chi^2}(U;V)$ is the variance of the Radon--Nikodym derivative $\frac{\d P_{UV}}{\d (P_U \otimes P_V)}$ w.r.t.\ the product distribution $P_U \otimes P_V$. We also have the following result:

\begin{corollary}\label{cor:SDPI_product} Let $(X,Y)$ be a random pair taking values in a finite product space $\sX \times \sY$, such that $P_X \in \PProb(\sX)$ and $P_Y \in \PProb(\sY)$. Then for any $\Phi \in \cF$ satisfying the conditions of Theorem~\ref{thm:gen_Erkip_Cover},
	\begin{align}\label{eq:SDPI_product}
		\eta_\Phi(P_X,P_{Y|X}) \eta_\Phi(P_Y,P_{X|Y}) \ge \frac{I_\Phi(X;X')}{I_\Phi(X;X)} \vee \frac{I_\Phi(Y;Y')}{I_\Phi(Y;Y)},
	\end{align}
	where $(X,X')$ is an exchangeable pair generated according to the Markov chain
	\begin{align}\label{eq:Gibbs_sampler_MC}
		X \xrightarrow{P_{Y|X}} Y \xrightarrow{P_{X|Y}} X'
	\end{align}
	and $(Y,Y')$ is an exchangeable pair generated according to the Markov chain
	\begin{align*}
		Y \xrightarrow{P_{X|Y}} X \xrightarrow{P_{Y|X|}} Y'.
	\end{align*}
\end{corollary}
\begin{proof} Applying Theorem~\ref{thm:gen_Erkip_Cover} to the Markov chain \eqref{eq:Gibbs_sampler_MC} gives
	\begin{align*}
		\eta_\Phi(P_Y,P_{X|Y}) \ge \frac{I_\Phi(X;X')}{I_\Phi(X;Y)}.
	\end{align*}
On the other hand,
\begin{align*}
	I_\Phi(X;Y) &= \sum_{x \in \sX} P_X(x) D_\Phi(P_{Y|X=x} \| P_Y) \\
	&= \sum_{x \in \sX} P_X(x) D_\Phi(\delta_x P_{Y|X} \| P_X P_{Y|X}) \\
	&\le \eta_\Phi(P_X,P_{Y|X}) \sum_{x \in \sX} P_X(x) D_\Phi(\delta_x \| P_X) \\
	&= \eta_\Phi(P_X,P_{Y|X}) I_\Phi(X;X),
\end{align*}
where $\delta_x$ denotes the Dirac measure located at $x$. Combining these estimates gives \eqref{eq:SDPI_product}. Interchanging the roles of $X$ and $Y$, we obtain an analogous bound involving $I_\Phi(Y;Y')$ and $I_\Phi(Y;Y)$.
\end{proof}
\noindent For example, if $\Phi(u) = u \log u$, the bound \eqref{eq:SDPI_product} becomes
\begin{align*}
	\eta(P_X,P_{Y|X})\eta(P_Y,P_{X|Y}) \ge \frac{I(X;X')}{H(X)} \vee \frac{I(Y;Y')}{H(Y)},
\end{align*}
where $H(X)$ is the usual Shannon entropy of $X$. If $\Phi(u) = (u-1)^2$, then we have
\begin{align*}
	S^2(P_X,P_{Y|X})S^2(P_Y,P_{X|Y}) \ge \frac{I_{\chi^2}(X;X')}{|\sX|-1} \vee \frac{I_{\chi^2}(Y;Y')}{|\sY|-1}.
\end{align*}
Corollary~\ref{cor:SDPI_product} may be useful for obtaining lower bounds on the mixing time of Gibbs samplers. It also shows that the modified log-Sobolev constant $c$ defined in \eqref{eq:M_mLSI_const} is bounded from below as
\begin{align*}
	c \ge {2H(X)}{H(X|X')},
\end{align*}
where $(X,X')$ is an exchangeable pair with $P_X = \mu$ and $P_{X'X} = M$.

\subsection{Fastest mixing Markov chain on a graph}

Let $G = (\sV,\sE)$ be a connected undirected graph with vertex set $\sV$ and edge set $\sE \subseteq \sV \times \sV$. Since $G$ is undirected, $(x,x') \in \sE \Rightarrow (x',x) \in \sE$. We assume that each vertex has a self-loop, i.e., $(x,x) \in \sE$ for all $x \in \sV$. Consider a (discrete-time) Markov chain $\{X_t\}_{t = 0, 1,\ldots}$ with states in $\sV$, whose one-step transition probability matrix $K$ has the following properties:
\begin{enumerate}
	\item It is symmetric, i.e., $K(x'|x) = K(x|x')$ for all $x,x' \in \sV$.
	\item It respects the graph structure, i.e., $K(x'|x) \neq 0$ only if $(x,x') \in \sE$.
\end{enumerate}
Let $\mu$ be the uniform distribution on $\sV$. The first property of $K$ implies that it is reversible with respect to $\mu$, so that $\mu = \mu K$. Let $\nu$ be the distribution of the initial state $X_0$, and let $\nu_t$ denote the distribution of $X_t$, the state at time $t$, so that $\nu_t = \nu K^t$. If the Markov chain is irreducible and aperiodic (which will be the case if $K(x|x) > 0$ for all $x \in \sV$), then $\nu_t$ will converge to $\mu$. There are multiple ways of quantifying the rate of convergence; we introduce the following definition:

\begin{definition} Given a convex function $\Phi \in \cF$, the {\em $\Phi$-mixing time} of $K$ is the function $\tau_\Phi(K,\cdot) : \Reals^+ \to \Naturals$, defined by
	\begin{align*}
		\tau_\Phi(K,\eps) \deq \min \left\{ t \in \Naturals: \sup_{\nu \in \Prob(\sV)} D_\Phi\big(\nu K^t \| \mu) \le \eps \right\}
	\end{align*}
\end{definition}

Unsurprisingly, the mixing time is controlled by the SDPI constant $\eta_\Phi(\mu,K)$:
\begin{proposition} Suppose $\Phi(0) < \infty$, and let $n = |\sV|$. Then
	\begin{align}\label{eq:f_mixtime_LB}
		\tau_\Phi(K,\eps) \le \frac{\log \big(D^*_{\Phi,n}/\eps\big)}{\log \big(1/\eta_\Phi(\mu,K)\big)},
	\end{align}
	where $D^*_{\Phi,n} \deq \frac{\Phi(n)}{n} + \left(1-\frac{1}{n}\right) \Phi(0)$.
\end{proposition}
\begin{proof} For any $t \ge 0$ and any $\nu \in \Prob(\sV)$,
	\begin{align*}
		D_\Phi(\nu K^t \| \mu) &= D_\Phi(\nu K^t \| \mu K^t) \le \Big(\eta_t(\mu,K)\Big)^t D_\Phi(\nu \| \mu).
	\end{align*}
	where we have used the fact that $\mu$ is $K$-invariant. Since $\Phi$-divergences are convex, and since a convex function on a compact convex set attains its maximum on an extreme point, we have
	\begin{align*}
		D_\Phi(\nu \| \mu) &\le \max_{x \in \sV} D_\Phi(\delta_x \| \mu),
	\end{align*}
	where $\delta_x$ is the Dirac measure located at $x$. Moreover, for any $x \in \sV$,
	\begin{align*}
		D_\Phi(\delta_x \| \mu) &= \frac{1}{n}\sum_{x' \in \sV} \Phi\left(\frac{\delta_x(x')}{1/n}\right) \\
		&= \frac{\Phi(n)}{n} + \left(1-\frac{1}{n}\right) \Phi(0) \\
		&\equiv D^*_{\Phi,n}.
	\end{align*}
Since $\nu$ was arbitrary, we have
\begin{align*}
	\sup_{\nu \in \Prob(\sV)} D_\Phi(\nu K^t \| \mu) &\le D^*_{\Phi,n} \Big(\eta_t(\mu,K)\Big)^t.
\end{align*}
Solving for the smallest $t$ that would make the right-hand side smaller than $\eps$, we obtain \eqref{eq:f_mixtime_LB}.
\end{proof}

It is customary to fix some value of $\eps$ (for discrete-time chains, a common choice is $1/2$), and to speak about the scaling of the mixing times in terms of the parameters of the graph and the Markov chain. For example, if $\Phi(u) = \frac{1}{2}|u-1|$, then the chain with one-step transition kernel $K$ mixes in $O\left(\frac{1}{\log \vartheta(K)^{-1}} \right)$ steps ($\TV$), where $\vartheta(K)$ is the Dobrushin coefficient of $K$; for $\Phi(u) = (u-1)^2$, the chain mixes in $O\left(\frac{\log n}{\log [S^2(\mu,K)]^{-1}}\right)$ steps $(\chi^2)$, where $S^2(\mu,K)$ is the maximal correlation; and for $\Phi(u) = u \log u$, the chain mixes in $O\left(\frac{\log \log n}{\log \eta(\mu,K)^{-1}}\right)$ steps (relative entropy). Thus, if $\eta_\Phi(\mu,K)$ is small, the corresponding Markov chain will mix faster in the sense that it will take fewer steps for the $\Phi$-divergence between the current state distribution and the uniform distribution on $\sV$ to fall below a given value. This motivates the following
\begin{quote} \textbf{Fastest mixing Markov chain (FMMC) problem:} Let $\Chan(G) \subset \Chan(\sV|\sV)$ be the set of all Markov kernels $K \in \Chan(\sV|\sV)$ satisfying the conditions listed in the beginning of this section. For a fixed convex function $\Phi \in \cF$,
	\begin{align*}
		\text{minimize } &\quad \eta_\Phi(\mu,K) \\
		\text{subject to }&\quad K \in \Chan(G)
	\end{align*}
\end{quote}
\begin{proposition} For any $\Phi \in \cF$, the FMMC problem is a convex program.
\end{proposition}
\begin{proof} The constraint set $\Chan(G)$ is convex. To see this, consider any two $K_1,K_2 \in \cM(G)$, and let $K = \lambda K_1 + \bar{\lambda}K_2$ for some $\lambda \in (0,1)$. Since both $K_1$ and $K_2$ are symmetric, for any pair $x,x' \in \sX$ we have
	\begin{align*}
		K(x'|x) &= \lambda K_1(x'|x) + \bar{\lambda}K_2(x'|x) = \lambda K_1 (x|x') + \bar{\lambda}K_2(x|x') = K(x|x').
	\end{align*}
Similarly, suppose that $(x,x')\not\in\sE$. Then $K_1(x'|x) = K_2(x'|x) = 0$, so $K(x'|x) = 0$ as well. Thus, $K \in \cM(G)$. The objective function $K \mapsto \eta_\Phi(\mu,K)$ is likewise convex, by Proposition~\ref{prop:SDPI_convexity}. 
\end{proof}
For $\Phi(u) = (u-1)^2$, the FMMC problem was studied by Boyd et al.~\cite{Boyd_etal_FMMC}, who showed that it can be equivalently represented by a semidefinite program (SDP), for which efficient solvers are available. For a general $\Phi$, there is not much one can say without exploiting specific properties of that $\Phi$ or any symmetries of the graph $G$; however, we can provide bounds on the values of the FMMC problems for different choices of $\Phi$. With that in mind, let $\eta^*_\Phi(G)$ denote the minimum value of the FMMC objective a given choice of $\Phi$ and $G$:
$$
\eta^*_\Phi(G) \deq \inf_{K \in \Chan(G)} \eta_\Phi(\mu,K).
$$
Then we observe the following:
\begin{itemize}
	\item $\eta^*_\Phi(G) \le \inf_{K \in \Chan(G)} \vartheta(K)$ for any $\Phi \in \cF$. This follows from the fact that $\eta_\Phi(\mu,K) \le \eta_\Phi(K) \le \vartheta(K)$, by Theorem~\ref{thm:Markov_contraction_bound}.
	\item If $\Phi$ is three times differentiable and $\Phi''(1) > 0$, then $\eta^*_\Phi(G) \ge \eta^*_{\chi^2}(G)$. This follows from the fact that, for such $\Phi$, $\eta_\Phi(\mu,K) \ge \eta_{\chi^2}(\mu,K)$ [cf.~Theorem~\ref{thm:chi_2_lower_bound}]. The quantity $\eta^*_{\chi^2}(G)$ and the corresponding convex program were studied extensively by Boyd et al.~\cite{Boyd_etal_FMMC}. 
\end{itemize}
The above definition of mixing time can be generalized to any other invariant distribution $\mu$ on $\sV$: Let $\Chan_\mu(G) \subset \Chan(\sV|\sV)$ be the set of all Markov kernels $K$, such that:
\begin{enumerate}
	\item $\mu(x)K(x'|x) = \mu(x')K(x|x')$ for all $x,x' \in \sV$.
	\item $K(x'|x) \neq 0$ only if $(x,x') \in \sE$.
\end{enumerate}
Then the same definition of the mixing time applies, and we have the bound
\begin{align*}
	\tau_\Phi(K,\eps) \le \frac{\log \big(D^*_{\Phi,\mu}/\eps\big)}{\log \big(1/\eta_\Phi(\mu,K)\big)},
\end{align*}
where
\begin{align*}
	D^*_{\Phi,\mu} \deq \max_{x \in \sV} \left\{ \mu(x) \Phi\left(\frac{1}{\mu(x)}\right) + \big(1-\mu(x)\big)\Phi(0)\right\}.
\end{align*}
We can then consider the appropriate modification of the FMMC problem, and the same arguments as before can be used to show that it is given by a convex program.

\subsection{Mixing times of Swendsen-Wang and heat-bath dynamics}

Let $G = (\sV,\sE)$ be an undirected graph without self-loops. In this case, we can identify the edge set of $G$ with a subset of ${\sV \choose 2}$, the set of all two-element subsets of $\sV$. If two vertices $u,v \in \sV$ are connected by an edge, we will write $u \leftrightarrow v$. Fix an integer $q \ge 2$, and consider the set $\sX = \sX_q = \{1,\ldots,q\}^\sV$ of tuples $x = (x_v : v \in \sV)$ with coordinates in $\{1,\ldots,q\}$. The elements of $\sX$ are called \textit{$q$-colorings of $G$}, and we say that $x \in \sX$ is a \textit{proper} $q$-coloring if $x_u \neq x_v$ whenever $u \leftrightarrow v$.

The problem of computing the number $\sP_G(q)$ of proper $q$-colorings of an arbitrary $G$ (or even deciding whether it is nonzero) is intractable, although it is known that $\sP_G(q)$ is polynomial in $q$. A related problem of drawing a $q$-coloring of $G$ uniformly at random (assuming $\sP_G(q) > 0$) is also intractable \cite{Jerrum_book}. However, it turns out that the problem of computing (or approximating) $\sP_G(q)$ is closely related to the problem of sampling from the so-called \textit{$q$-state Potts model}, described by the Gibbs distribution
\begin{align}\label{eq:Potts}
	\PP_{\beta,q}(x) \deq \frac{1}{Z(\beta,q)} \exp \Bigg( \beta \sum_{u,v \in \sV \atop u \leftrightarrow v} \1\{x_u = x_v\}\Bigg),
\end{align}
where the parameter $\beta \ge 0$ is called the inverse temperature, and $Z(\beta,q)$ is the normalization constant known as the partition function. In particular, $\sP_G(q) = \lim_{\beta \to \infty} Z(\beta,q)$. Direct sampling from $\PP_{\beta,q}$ is also intractable, so one resorts to Markov Chain Monte Carlo (MCMC) methods: Pick a Markov kernel $K \in \Chan(\sX|\sX)$ that leaves the Gibbs distribution \eqref{eq:Potts} invariant, pick an arbitrary initial configuration $x_0 \in \sX$, and for each $t = 0,1,\ldots$ generate a random configuration $X_{t+1}$ according to $K(\cdot|X_t)$. With a good choice of $K$, the distribution of $X_t$ will rapidly converge to $\PP_{\beta,q}$. Two popular choices of $K$ are the \textit{heat-bath} (or \textit{Glauber}) \textit{dynamics} and the \textit{Swendsen-Wang dynamics} \cite{Winkler_MCMC_book,Grimmett_RC}. They are defined as follows:

\paragraph{Heat-bath dynamics.} At each time step $t$, given the current configuration $x_t = (x_{v,t})_{v \in \sV}$, we pick a vertex $v \in \sV$ uniformly at random, assign it a new random color $X_{v,t+1} \in \{1,\ldots,q\}$ according to the conditional distribution $\PP_{\beta,q}(X_{v,t}|X^{\backslash v}_t = x^{\backslash v}_t)$, and set $X^{\backslash v}_{t+1} = x^{\backslash v}_t$. Here, $x^{\backslash v}_t = (x_{u,t})_{u \in \sV \backslash \{v\}}$ is the time-$t$ configuration of all the vertices except $v$. Thus, the transition probabilities of the heat-bath Markov chain are given by the Markov kernel
\begin{align}\label{eq:heat_bath}
	K^{\rm HB}_{\beta,q}(x'|x) = \frac{1}{|\sV|}\sum_{v \in \sV} \PP_{\beta,q}(x'_v|x^{\backslash v}) \1\{ (x')^{\backslash v} = x^{\backslash v} \}.
\end{align}
\paragraph{Swendsen-Wang dynamics.} This construction is based on a coupling of the $q$-Potts model and the so-called \textit{random-cluster} (or {\em Fortuin-Kasteleyn}) model on $G$. The latter is defined as follows \cite{Grimmett_RC}. Let $\sY = 2^\sE = \{\sA: \sA \subseteq \sE\}$ and fix a parameter $p \in (0,1)$. Then the random-cluster model is described by the following probability measure on $\sY$:
\begin{align}\label{eq:FK}
	\QQ_{p,q}(\sA) \deq \frac{1}{\tilde{Z}(p,q)} \left(\frac{p}{\bar{p}}\right)^{|\sA|} q^{C(\sA)}, \qquad \forall \sA \in \sY
\end{align}
where $\tilde{Z}(\cdot,\cdot)$ is the partition function, $\bar{p} = 1-p$, and $C(\sA)$ is the number of connected components of the induced graph $(\sV,\sA)$. It can be shown that
\begin{align*}
	\tilde{Z}(p,q) = Z\left( \log(1/\bar{p}),q\right),
\end{align*}
where $Z(\cdot,\cdot)$ is the partition function for the $q$-Potts model. Now let $p = 1 - e^{-\beta}$, and consider the following probability measure on the Cartesian product $\sX \times \sY$:
\begin{align}\label{eq:ES}
	\MM(x,\sA) &= \frac{1}{\tilde{Z}(p,q)} \left( \frac{p}{\bar{p}}\right)^{|\sA|} \1\{ \sA \subset \sE(x)\} \\
	&= \frac{1}{Z(\beta,q)} \left( e^{\beta}-1\right)^{|\sA|} \1\{ \sA \subset \sE(x)\}.
\end{align}
where $\sE(x) \deq \left\{ \{u,v\} \in \sE \, : \, x_u = x_v \right\}$ is the set of edges on which $x$ violates the proper $q$-coloring constraint. It can be shown that $\MM$ is a coupling of $\PP_{\beta,q}$ and $\QQ_{p,q}$ with $p = 1-e^{-\beta}$, i.e., if $(X,Y)$ is a random pair with law $\MM$, then $P_X = \PP_{\beta,q}$ and $P_Y = \QQ_{1-e^{-\beta},q}$.

With these definitions at hand, we can describe the Swendsen-Wang algorithm:
\begin{itemize}
	\item Start with an arbitrary initial configuration $x_0 \in \sX$
	\item For each $t=0,1,2,\ldots$
	\begin{itemize}
		\item Draw a random set $\sA_t \in \sY$ according to the conditional distribution $\MM(Y=\cdot|X=x_t)$.
		\item Draw $X_{t+1}$ from the conditional distribution $\MM(X=\cdot|Y=\sA_t)$.
	\end{itemize}
\end{itemize}
In words, given $x_t$, we draw $\sA_t$ by deleting each edge of $\sE(x_t)$ independently with probability $p = 1-e^{-\beta}$; given $\sA_t$, we draw $X_{t+1}$ by assigning a random color independently to each connected component of $(\sV,\sA_t)$ and coloring all vertices in the same component with the same color. Thus, the Swendsen-Wang dynamics is a two-stage Gibbs sampler that generates a trajectory $\{(X_t,Y_t)\}_{t \ge 0}$ according to
\begin{align*}
	\ldots \longrightarrow X_t \xrightarrow{\quad \MM_{Y|X} \quad } Y_t \xrightarrow{\quad \MM_{X|Y} \quad} X_{t+1} \longrightarrow \ldots
\end{align*}
The discrete-time process $\{X_t\}_{t \ge 0}$ is a Markov chain with one-step transition kernel
\begin{align*}
	K^{\rm SW}_{\beta,q}(x'|x) &= \MM_{X|Y} \circ \MM_{Y|X} (x'|x) \\
	&= \sum_{\sA \in \sY} \MM_{X|Y}(x'|\sA) \MM_{Y|X}(\sA|x).
\end{align*}
By construction, the Markov kernel $K^{\rm SW}_{\beta,q}$ is reversible w.r.t.\ the Gibbs measure $\PP_{\beta,q}$.

With each of these two algorithms, the hope is that the corresponding Markov chain mixes rapidly, i.e., the distribution of the state $X_t$ converges quickly to $\PP_{\beta,q}$ as $t \to \infty$. Just as in the previous section, for a given divergence-generating function $\Phi \in \cF$, the rate at which $D_\Phi\big(P_{X_t} \big\| \PP_{\beta,q}\big)$ converges to zero is controlled by the SDPI constant $\eta_\Phi\big(\PP_{\beta,q}, K^{\bullet}_{\beta,q}\big)$, where $\bullet$ is either ${\rm HB}$ or ${\rm SW}$. The heat-bath algorithm is widely used because it is easy to implement. On the other hand, the popularity of the Swendsen-Wang algorithm is due to the fact that, empirically, it tends to mix rapidly for a wide variety of graphs and small values of $q$ (however, see \cite{Borgs_etal_torpid_SW} for examples of slow mixing of Swendsen-Wang). In a recent paper, Ullrich \cite{Ullrich_SW_vs_HB} showed that the spectral gap of Swendsen-Wang is lower-bounded by a constant multiple of the spectral gap of the heat-bath kernel, where the constant depends on the number of colors $q$, the inverse temperature $\beta$, and the maximum degree $\Delta$ of $G$. Now, the spectral gap can be related to the SDPI constant for the $\chi^2$-divergence (see Remark~\ref{rem:spectral_gap}), so Ullrich's result can immediately be converted into a statement about the $\chi^2$ SDPI constants of Swendsen-Wang and heat-bath kernels. The theorem below sharpens and extends the bound of Ullrich to other $\Phi$-divergences; just like in \cite{Ullrich_SW_vs_HB}, the theorem allows us to convert any available upper bound for the heat-bath kernel into an upper bound for the Swendsen-Wang kernel (or, conversely, any lower bound for Swendsen-Wang into a lower bound for heat-bath). 

\begin{theorem}\label{thm:SW_vs_HB} For any $\Phi \in \cF$ that satisfies the generalized homogeneity condition \eqref{eq:gen_hom}, 
	\begin{align}\label{eq:SW_vs_HB}
		\eta_\Phi\left(\PP_{\beta,q},K^{\rm SW}_{\beta,q}\right) \le \frac{q^{2\Delta+1}e^{4\beta\Delta}-1}{q^{2\Delta+1}e^{4\beta\Delta} - \left[\eta_\Phi\left(\PP_{\beta,q},K^{\rm HB}_{\beta,q} \right)\right]^2},
	\end{align}
where $\Delta = \max_{v \in \sV} \deg_G(v)$ is the maximum degree of $G$.
\end{theorem}
\begin{remark} {\em In the notation of this paper, the main result of \cite{Ullrich_SW_vs_HB} can be written as
	\begin{align}\label{eq:Ullrich}
		\sqrt{\eta_{\chi^2}\left(\PP_{\beta,q},K^{\rm SW}_{\beta,q}\right)} \le  \frac{2q^{4\Delta+2}e^{8\beta\Delta}-1 + \sqrt{\eta_{\chi^2}\left(\PP_{\beta,q},K^{\rm HB}_{\beta,q} \right)}}{2q^{4\Delta+2}e^{8\beta\Delta}}.
	\end{align}
Particularizing our bound \eqref{eq:SW_vs_HB} to the case $\Phi(u) = (u-1)^2$, we see that it is tighter than \eqref{eq:Ullrich}. A plot of the two bounds as a function of the $\chi^2$ SDPI constant of the heat-bath dynamics is shown in Figure~\ref{fig:compare_vs_Ullrich} for $q=2$, $\Delta=3$, and $\beta = 0.001$. (Admittedly, both bounds are fairly crude even for small values of $q$ and $\Delta$, due to the presence of $O(q^\Delta)$ terms.)\hfill$\diamond$}
\end{remark}

\begin{figure}\label{fig:compare_vs_Ullrich}
	\centerline{\includegraphics[width=0.5\textwidth]{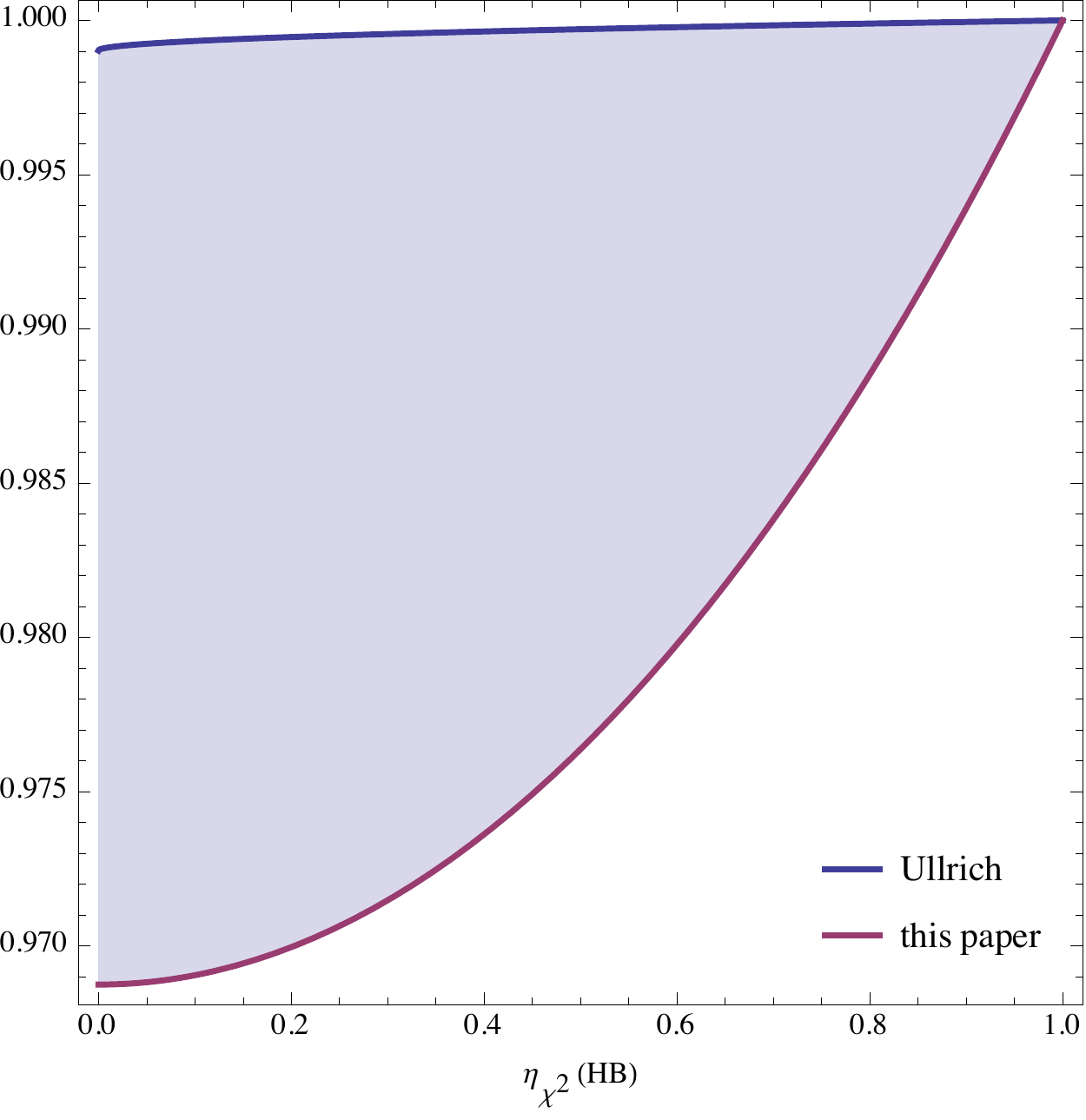}}
	\caption{Upper bounds of Ullrich and Theorem~\ref{thm:SW_vs_HB} as a function of the $\chi^2$ SDPI constant of the heat-bath dynamics, for $q=2$, $\Delta = 3$, $\beta = 0.001$.}
\end{figure}

\begin{proof} We borrow a clever trick of Ullrich \cite{Ullrich_SW_vs_HB} and compare the Swendsen-Wang kernel $K^{\rm SW}_{\beta, q}$ to $K = K^{\rm HB}_{\beta, q} \circ K^{\rm SW}_{\beta, q} \circ K^{\rm HB}_{\beta,q}$. Since the Gibbs distribution $\PP_{\beta,q}$ is invariant under both the SW and the HB kernels, it is also the invariant distribution of $K$. Moreover, for any $\nu \neq \PP_{\beta,q}$, we have
	\begin{align*}
		D_\Phi\big(\nu K \big\| \PP_{\beta, q}\big) &= D_\Phi \big( \nu K \big \| \PP_{\beta,q} K \big) \\
		&= D_\Phi \Big( \nu \big(K^{\rm HB}_{\beta, q} \circ K^{\rm SW}_{\beta, q} \circ K^{\rm HB}_{\beta,q}\big) \Big\| \PP_{\beta, q}\big(K^{\rm HB}_{\beta, q} \circ K^{\rm SW}_{\beta, q} \circ K^{\rm HB}_{\beta,q}\big) \Big) \\
		&\le \left[\eta_\Phi\left(\PP_{\beta,q},K^{\rm HB}_{\beta,q}\right)\right]^2  \eta_\Phi(\PP_{\beta,q},K^{\rm SW}_{\beta,q}) \, D_\Phi (\nu \| \PP_{\beta,q}),
	\end{align*}
	where we have repeatedly exploited the invariance of $\PP_{\beta,q}$ w.r.t.\ the SW and the HB kernels. Since $\nu$ was arbitrary, we conclude that
	\begin{align}\label{eq:eta_f_Ullrich_K}
		\eta_\Phi(\PP_{\beta,q},K) \le \left[\eta_\Phi\left(\PP_{\beta,q},K^{\rm HB}_{\beta,q}\right)\right]^2  \eta_\Phi(\PP_{\beta,q},K^{\rm SW}_{\beta,q}).
	\end{align}
On the other hand, Ullrich also proved that
	\begin{align}\label{eq:Ullrich_K_ratio}
		\max_{x,x' \in \sX} \frac{K(x'|x)}{K^{\rm SW}_{\beta,q}(x'|x)} \le q^{2\Delta+1}e^{4\beta\Delta}.
	\end{align}
From Eq.~\eqref{eq:Ullrich_K_ratio} and Corollary~\ref{cor:SDPI_comparison}, we get the estimate
\begin{align}\label{eq:Ullrich_K_comparison}
	\eta_\Phi(\PP_{\beta,q},K^{\rm SW}_{\beta,q}) \le 1 - \frac{1}{q^{2\Delta+1}e^{4\beta\Delta}}\left(1 - \eta_\Phi(\PP_{\beta,q},K)\right).
\end{align}
Finally, using \eqref{eq:eta_f_Ullrich_K} in \eqref{eq:Ullrich_K_comparison} and rearranging, we obtain \eqref{eq:SW_vs_HB}.
\end{proof}

\subsection{Reconstruction in graphical models}

The Potts model described in the preceding section is an example of a \textit{probabilistic graphical model} (or a pairwise Markov random field) \cite{Wainwright_Jordan_graphical}. Any such model is specified by a pair $(G,\mathbf{{\mathbb U}})$, where $G = (\sV,\sE)$ is an undirected graph and $\mathbf{{\mathbb U}} = \{{\mathbb U}_e\}_{e \in \sE}$ is a collection of symmetric \textit{edge potentials} ${\mathbb U}_e : \Omega\times\Omega \to \Reals^+$. Here, $\Omega$ is a finite set often referred to as \textit{state} or \textit{spin space}. The configuration space of the graphical model is the set $\sX = \Omega^\sV$ of all tuples $x = (x_v)_{v \in \sV}$, where each $x_v$ takes values in $\Omega$. Once $G$ and $\mathbf{{\mathbb U}}$ are fixed, we consider the following probability measure on $\sX$:
\begin{align*}
	\PP_{G,\mathbf{{\mathbb U}}}(x) = \frac{1}{Z(G,\mathbf{{\mathbb U}})}\prod_{\{u,v\} \in \sE} {\mathbb U}_{\{u,v\}}(x_u,x_v),
\end{align*}
where $Z$ is the normalization constant. For example, the $q$-state Potts model on $G$ [cf.~Eq.~\eqref{eq:Potts}] is of this form with $\Omega = \{1,\ldots,q\}$ and ${\mathbb U}_{uv}(x_u,x_v) = \exp\left(\beta \1\{x_u = x_v\} \right)$.

The \textit{reconstruction problem} (see, e.g., \cite{Montanari_Gerschenfeld_reconstruction,Bhatnagar_reconstruction}) for the graphical model $(G,\mathbf{{\mathbb U}})$ can be stated informally as follows: Given two disjoint sets of vertices $\sA$ and $\sB$, how much can we infer about the configuration $X_\sA \deq (X_v)_{v \in \sA}$ on $\sA$ by observing $X_\sB$? For a precise definition, let $d_G$ denote the \textit{graph distance} on $G$, i.e., $d_G(u,v)$ is the number of edges on the shortest path between $u$ and $v$.

\begin{definition} Given a function $\Phi \in \cF$, we say that the probabilistic graphical model $(G,\mathbf{{\mathbb U}})$ is {\em not $\Phi$-reconstructible} if for any set of vertices $\sA$ there exist some constants $C_\sA,c_\sA > 0$, such that
	\begin{align*}
		I_\Phi(X_\sA; X_\sB) \le C_\sA e^{-c_\sA d_G(\sA,\sB)}
	\end{align*}
for all sets $\sB$ disjoint from $\sA$, where $d_G(\sA,\sB) \deq \min_{u \in \sA,v \in \sB} d_G(u,v)$. Here, the $\Phi$-information is computed w.r.t.\ the marginal distribution of $(X_\sA,X_\sB)$ induced by $\PP_{G,{\mathbb U}}$.
\end{definition}
Alternatively, we may examine correlations between functions of $X_\sA$ and $X_\sB$:
\begin{definition} The graphical model $(G,\mathbf{{\mathbb U}})$ has {\em exponential decay of correlations} if for any $\sA \subset \sV$ there exist positive constants $C_\sA,c_\sA$, such that, for any set of vertices $\sB$ disjoint from $\sA$ and  for any two functions $f \in \Func(\sX_\sA)$ and $g \in \Func(\sX_\sB)$,
	\begin{align*}
		\Cov \left[ f(X_\sA), g(X_\sB)\right] \le C_\sA e^{-c_\sA d_G(\sA,\sB)} \sqrt{\Var[f(X_\sA)]\Var[g(X_\sB)]}.
	\end{align*}
\end{definition}
We can now establish the following result:

\begin{theorem} Suppose that $\Phi \in \cF$ is twice differentiable and strictly convex, its second derivative is nonincreasing, and the  function $\Psi$ defined in \eqref{eq:first_diff_0} is concave. Then $(G,\mathbf{{\mathbb U}})$ is not $\Phi$-reconstructible if and only if it has exponential decay of correlations.
\end{theorem}

\begin{proof} We first show that exponential decay of correlations is equivalent to $(G,\mathbf{{\mathbb U}})$ not being $\chi^2$-reconstructible. With a slight abuse of notation, we will denote by $P_\sA$ the marginal distribution of $X_\sA$, etc. By definition of maximal correlation, $(G,\mathbf{{\mathbb U}})$ has exponential decay of correlation if and only if for any $\sA \subset \sV$ there exist some $C_\sA,c_\sA>0$, such that
	\begin{align}\label{eq:exp_decay_maxcorr}
		S^2(P_\sA, P_{\sB|\sA}) \le C_\sA e^{-c_\sA d_G(\sA,\sB)}
	\end{align}
for all sets of vertices $\sB$ with $\sB \cap \sA = \varnothing$. Now let $\tilde{\sX}_\sA$ denote the support of $P_\sA$. Using the definition of $\chi^2$-information and Theorem~\ref{thm:maxcorr}, we can write
	\begin{align*}
		I_{\chi^2}(X_\sA; X_\sB) &= \sum_{x_\sA \in \tilde{\sX}_\sA} P_\sA(x_\sA) \chi^2 \big( P_{X_\sB|X_\sA = x_\sA} \big\| P_\sB \big) \\
		&= \sum_{x_\sA \in \tilde{\sX}_\sA} P_\sA(x_\sA) \chi^2\big( \delta_{x_\sA} P_{\sB|\sA} \big\| P_\sA P_{\sB|\sA}\big) \\
		&\le S^2(P_\sA, P_{\sB|\sA})\sum_{x_\sA} P_\sA(x_\sA) \chi^2\big( \delta_{x_\sA} \big\| P_\sA\big) \\
		&= S^2(P_\sA, P_{\sB|\sA}) I_{\chi^2}(X_{\sA}; X_{\sA}) \\
		&= S^2(P_\sA,P_{\sB|\sA}) \left(|\sX_\sA|-1\right).
	\end{align*}
From this and from \eqref{eq:exp_decay_maxcorr}, we see that exponential decay of correlations implies that $(G,\mathbf{{\mathbb U}})$ is not $\chi^2$-reconstructible. The converse statement follows from the inequality $S^2(P_\sA,P_{\sB|\sA}) \le I_{\chi^2}(X_\sA; X_\sB)$ \cite[Prop.~12]{YP_YW_dissipation}.

To complete the proof, let $(\bar{X}_\sA,\bar{X}_\sB)$ be a random pair with probability law $P_\sA \otimes P_\sB$. Then
\begin{align*}
	I_\Phi(X_\sA; X_\sB) = \Ent_\Phi\left[ f(\bar{X}_\sA, \bar{X}_\sB)\right],
\end{align*}
where we have defined
\begin{align*}
	f(x_\sA,x_\sB) \deq \frac{P_{\sA| \sB}(x_\sA|x_\sB)}{P_\sA(x_\sA)}.
\end{align*}
In particular, $I_{\chi^2}(X_\sA; X_\sB) = \Var \left[ f(\bar{X}_\sA,\bar{X}_\sB)\right]$. Therefore, applying Lemmas~\ref{lm:f_entropy_UB} and \ref{lm:f_entropy_LB} in Appendix~\ref{app:lemmas} and using the fact that $\| f(\bar{X}_\sA,\bar{X}_\sB)  \|_\infty \le 1/p^\sA_*$, where $p^\sA_* \deq \displaystyle\min_{x_\sA \in \tilde{\sX}_\sA} P_\sA(x_\sA)$ is the minimum nonzero probability of any configuration in $\sA$, we get
\begin{align*}
	\frac{\Phi''(1/p^\sA_*)}{2} I_{\chi^2}(X_\sA; X_\sB) \le I_\Phi(X_\sA; X_\sB) \le \Psi'(1) I_{\chi^2}(X_\sA; X_\sB).
\end{align*}
Since $\Phi$ is strictly convex, $\Phi''$ is everywhere positive. This inequality shows that the graphical model $(G,{\mathbf {\mathbb U}})$ is not $\Phi$-reconstructible if and only if it is not $\chi^2$-reconstructible, which in turn is equivalent to exponential decay of correlations.
\end{proof}

A related notion of correlation decay has to do with the diminishing influence of ``far away'' spins. A key property of Gibbs measures is the following conditional independence relation: for any $\sA \subset \sV$, the outer boundary of $\sA$, denoted by $\partial \sA$, is the set of all $v \in \sA^c$, such that $\{u,v\} \in \sE$ for some $u \in \sA$. Then under $\PP_{G,\mathbb{U}}$, 
\begin{align*}
	X_\sA \longrightarrow X_{\partial \sA} \longrightarrow X_{\sA^c}
\end{align*}
is a Markov chain. That is, the configuration of spins in a given set $\sA$ of vertices is conditionally independent of all other spins given the configuration of the neighbors of $\sA$. The following definition formalizes the notion that the influence of the spins in the boundary of $\sA$ on the spins in any subset of $\sA$ should decay with the distance from that subset to the boundary:

\begin{definition} The graphical model $(G,\mathbb{U})$ has the {\em spatial mixing property} if there exist positive constants $C,c$, such that, for any two sets of vertices $\sB \subset \sA \subset \sV$ and for any two boundary configurations $x_{\partial \sA},\bar{x}_{\partial \sA}$,
	\begin{align}\label{eq:spatial_mixing}
		\left\| \frac{P_{\sB|\partial \sA}(\cdot|x_{\partial \sA})}{P_{\sB|\partial \sA}(\cdot|\bar{x}_{\partial \sA})}-1\right\|_\infty \le C|\sA| e^{-c d_G(\sB,\partial \sA)}.
	\end{align}
\end{definition}
\begin{remark}{\em This mixing condition is slightly stronger than the condition proposed by Weitz \cite{Weitz_thesis}, which is in turn stronger (but more generally applicable) than the complete analyticity condition of Dobrushin and Shlosman \cite{Dobrushin_Shlosman_CA}. The latter is only applicable to the case when the underlying graph $G$ is the square lattice ${\mathbb Z}^d$.\hfill$\diamond$}
\end{remark}
\noindent If $(G,\mathbb{U})$ has spatial mixing, then one would expect the relative-entropy SDPI constant of the channel $P_{\sB|\partial \sA}$ at $P_{\partial \sA}$ to decay exponentially with the distance $d_G(\sB,\partial \sA)$. This is indeed the case:
\begin{theorem} Suppose that $(G,\mathbb{U})$ has the spatial mixing property. Then
	\begin{align}\label{eq:spatial_mixing_SDPI}
		\eta(P_{\partial \sA}, P_{\sB|\partial \sA}) \le \frac{2C^2|\sA|^2}{p^*_\sB} e^{-2c d_G(\sB,\partial \sA)}.
	\end{align}
\end{theorem}
\begin{proof} Using \eqref{eq:eta_general_channel}, we can upper-bound $\eta(P_{\partial \sA}, P_{\sB|\partial \sA})$ as follows:
	\begin{align*}
		\eta(P_{\partial \sA}, P_{\sB|\partial \sA}) \le \frac{1}{2}\sum_{x_\sB} \frac{1}{P_\sB(x_\sB)} \max_{x_{\partial \sA},x'_{\partial \sA}} \left| P_{\sB|\partial \sA}(x_\sB|x_{\partial \sA}) - P_{\sB|\partial \sA}(x_\sB|x'_{\partial \sA})\right|^2.
	\end{align*}
If we now pick an arbitrary boundary configuration $\bar{x}_{\partial \sA}$, then we can write
\begin{align*}
		\eta(P_{\partial \sA}, P_{\sB|\partial \sA}) &\le 2\sum_{x_\sB} \frac{1}{P_\sB(x_\sB)} \max_{x_{\partial \sA}} \left| P_{\sB|\partial \sA}(x_\sB|x_{\partial \sA}) - P_{\sB|\partial \sA}(x_\sB|\bar{x}_{\partial \sA})\right|^2 \\
		&\le 2\sum_{x_\sB} \frac{P_{\sB|\partial \sA}(x_\sB|\bar{x}_{\partial \sA})}{P_\sB(x_\sB)} \max_{x_{\partial \sA}} \left| \frac{P_{\sB|\partial \sA}(x_\sB|x_{\partial \sA})}{P_{\sB|\partial \sA}(x_\sB|\bar{x}_{\partial \sA})}-1\right|^2.
\end{align*}
Using \eqref{eq:spatial_mixing}, we get \eqref{eq:spatial_mixing_SDPI}.
\end{proof}

\section{Summary of contributions and concluding remarks}
\label{sec:summary}

In this paper, we have attempted to give a systematic and unified presentation of strong data processing inequalities (SDPIs) for discrete channels. As a reminder, given a convex function $\Phi : \Reals^+ \to \Reals$, we say that a channel $K \in \Chan(\sY|\sX)$ satisfies an SDPI with constant $c \in [0,1)$ at input distribution $\mu$ if
\begin{align}\label{eq:linear_SDPI}
	D_\Phi(\nu K \| \mu K) \le c D_\Phi(\nu \| \mu)
\end{align}
for all $\nu \neq \mu$. We denote the best constant in the above inequality by $\eta_\Phi(\mu,K)$, and let$\eta_\Phi(K) \deq \sup_\mu \eta_\Phi(\mu,K)$. For the reader's convenience, we summarize the key novel contributions:
\begin{itemize}
	\item For all sufficiently smooth $\Phi$, $\eta_\Phi(\mu,K)$ is lower-bounded by the squared maximal correlation $S^2(\mu,K)$, which is also the SDPI constant of $K$ at $\mu$ for the $\chi^2$-divergence (Theorem~\ref{thm:chi_2_lower_bound}). This refines the inequality $\eta(\mu,K) \ge S^2(\mu,K)$ due to Ahlswede and G\'acs \cite{Ahlswede_Gacs_hypercont}, as well as the inequality $\eta_\Phi(K) \ge S^2(K) \equiv \sup_\mu S^2(\mu,K)$ due to Cohen et al.~\cite{Cohen_etal_dataproc}.
	\item For all operator convex $\Phi$ (see Section~\ref{ssec:opconv} for definitions and examples), we have proved the upper bound
	$$
	\eta_\Phi(\mu,K) \le \max \left(S^2(\mu,K), \sup_{0 < \lambda < 1} \eta_{\LC_\lambda}(\mu,K)\right)
	$$
(Theorem~\ref{thm:opconv_bounds}), where $\LC_\lambda(\cdot \| \cdot)$ denotes the Le Cam divergence with parameter $\lambda$ (see Section~\ref{sec:entropies}). This refines the inequality $\eta_\Phi(K) \le S^2(K)$ for all operator convex $\Phi$, due to Choi et al.~\cite{Choi_Ruskai_Seneta}, and reduces to it upon taking the supremum of both sides w.r.t.\ $\mu$. 
\item For $\Phi(u) = u \log u$ (which gives the usual relative entropy), the SDPI constant $\eta_\Phi(\mu,K)$ can be upper-bounded in terms of the subgaussian constant $\sigma^2(y)$ of the posterior likelihood ratio $a(X,y) = \frac{K^*(X|y)}{\mu(X)}$ for each $y \in \sY$, where $X \sim \mu$. Smaller value of $\sigma^2(y)$ indicates that $a(X,y) \approx 1$ with high probability, which means that the observation $Y = y$ is nearly uninformative about the input $X$. Theorem~\ref{thm:subgaussian} gives the inequaity $\eta(\mu,K) \le 2\,\E[\sigma^2(Y)]$, which can be weakened to the bound of Theorem~\ref{thm:info_transport} using information-transportation inequalities.
\item Under mild regularity conditions on $\Phi$, the SDPI constants \textit{tensorize}: given a product distribution $\mu_1 \otimes \ldots \otimes \mu_n$ and a product channel $K_1 \otimes \ldots \otimes K_n$,
$$
\eta_\Phi(\mu_1 \otimes \ldots \otimes \mu_n, K_1 \otimes \ldots \otimes K_n) = \max_{1 \le i \le n} \eta_\Phi(\mu_i,K_i)
$$
(Theorem~\ref{thm:tensorization_1}). This extends previous tensorization results for $\Phi(u)=(u-1)^2$ due to Witsenhausen \cite{Witsenhausen_correlation} and for $\Phi(u)=u\log u$ due to Anantharam et al.~\cite{Anantharam_etal_HGR}. Theorem~\ref{thm:tensorization_2} gives a tensorization inequality for \textit{mixtures} of local channels, i.e., when an input block of length $n$ is transformed to an output block of length $n$ by drawing a coordinate index $I$ at random from $\{1,\ldots,n\}$ and then passing the $I$th symbol through the channel $K_I$.
\item Section~\ref{sec:phisob} is dedicated to an exposition of the deep links between SDPIs and $\Phi$-Sobolev inequalities \cite{Chafai_entropy}, which provide a powerful tool for nonasymptotic quantitative analysis of convergence to equilibrium in Markov processes and other random dynamical systems. For the specific case of $\Phi(u)=u\log u$, we have obtained a number of inequalities relating the optimal constants in log-Sobolev inequalities for a reversible Markov chain $M$ with invariant distribution $\mu$ to relative-entropy SDPI constants $\eta(\mu,K)$ for any channel $K$ with the property that $M = K^*_\mu \circ K$, where $K^*_\mu$ is the adjoint, or backward, channel associated to the pair $(\mu,K)$ [see Eq.~\eqref{eq:backward_2} for the definition].
\item Section~\ref{sec:applications} presents several applications of the results of preceding sections to information theory, discrete probability, and statistical physics. In particular, we discuss a connection between the strong data processing property and the concentration-of-measure phenomenon; generalize a recent result of Anantharam et al.~\cite{Anantharam_etal_HGR} on the strong contraction of mutual information in discrete Markov chains\footnote{See \cite{PW_BayesSDPI} for an extension of this result to abstract alphabets.} to a more general notion of $\Phi$-information; relate the problem of computing SDPI constants (which is a convex program) to the problem of finding the fastest mixing Markov chain on a graph; sharpen a recent result of Ullrich \cite{Ullrich_SW_vs_HB} on the mixing time of two popular MCMC schemes for a certain class of graphical models; and outline an SDPI-based characterization of the decay of correlations in discrete graphical models.
\end{itemize}
After the original breakthrough work of Ahlswede and G\'acs \cite{Ahlswede_Gacs_hypercont}, strong data processing inequalities have received a great deal of attention, with a recent surge of research activity motivated by problems in information theory. Recent work by Polyanskiy and Wu \cite{YP_YW_dissipation}  has uncovered certain limitations of SDPIs. For example, in the setting of continuous alphabets and additive-noise channels, they have shown that it is possible for a channel $K$ to have $\eta_\Phi(K)=1$ and still satisfy a weaker ``nonlinear'' strong data processing inequality of the form
$$
D_\Phi(\nu K \| \mu K) \le F_\Phi\big(D_\Phi(\nu \| \mu)\big)
$$
for some increasing function $F_\Phi : \Reals^+ \to \Reals^+$ with $F_\Phi(0)=0$, such that $F_\Phi(u)<u$ for all sufficiently small $u \neq 0$. Nevertheless, SDPIs still remain a versatile tool for many problems of current theoretical and practical interest.

\section*{Acknowledgments}

The author would like to thank V.~Anantharam, S.~Kamath, A.~Kontorovich, C.~Nair, Y.~Polyanskiy, I.~Sason, P.~Tetali, R.~van Handel, and Y.~Wu for many useful and stimulating discussions, and the two anonymous reviewers and the Associate Editor for their meticulous reading of the manuscript and for numerous useful suggestions and corrections. The author would also like to separately thank one of the anonymous reviewers for a suggestion on how to streamline the proof of Theorem~\ref{thm:SDP_var}, as well as for pointing out a subtle issue pertainig to Theorem~\ref{thm:Jensen_gap}.

\begin{appendix}
\section{Miscellaneous lemmas}
\label{app:lemmas}
	\begin{lemma}\label{lm:density_update} Let $(\mu,K) \in \PProb(\sX) \times \Chan(\sY|\sX)$ be an admissible pair, and consider any other $\nu \in \Prob(\sX)$. If $f = \d \nu/\d \mu$, then
		\begin{align*}
			K^*f = \frac{\d (\nu K)}{\d (\mu K)},
		\end{align*}
where $K^* = K^*_\mu \in \Chan(\sY|\sX)$ is the backward channel induced by the pair $(\mu,K)$, cf.~Eqs.~\eqref{eq:backward_1}, \eqref{eq:backward_2}.
	\end{lemma}
	\begin{proof} A direct calculation:
		\begin{align*}
			K^*f(y) &= \sum_{x \in \sX}  K^*(x|y) f(x) \\
			&= \sum_{x \in \sX} \frac{K(y|x)\mu(x)}{\mu K(y)} \frac{\nu(x)}{\mu(x)}  \\
			&= \frac{1}{\mu K(y)}\sum_{x \in \sX} \nu(x)K(y|x) \\
			&= \frac{\nu K(y)}{\mu K(y)} \\
			&= \frac{\d (\nu K)(y)}{\d(\mu K)}
		\end{align*}
		for any $y \in \sY$.
	\end{proof}
	
		\begin{lemma}\label{lm:f_entropy_UB} Suppose $\Phi \in \cF$ is differentiable, and the function $\Psi(u) = \frac{\Phi(u)-\Phi(0)}{u}$ is concave. Then for any nonnegative random variable $U$ with $\E U = 1$,
			\begin{align}\label{eq:f_entropy_UB}
				\Ent_\Phi[U] \le \Psi\big(1 + \Var[U] \big) - \Psi(1) \le \Psi'(1) \Var[U].
			\end{align}
		\end{lemma}
		\begin{proof} We can assume that $\Var[U] < \infty$, because otherwise there is nothing to prove. Let $P$ denote the law of $U$. Since $U$ is nonnegative and has unit mean, $Q(\d u) \deq u P(\d u)$ is a probability measure. Therefore,
			\begin{align*}
				\Ent_\Phi[U] &= \E_P[\Phi(U)]  - \Phi(1) \\
				&= \E_Q[\Psi(U)] - \Psi(1) \\
				&\le \Psi(\E_Q U) - \Psi(1) \\
				&= \Psi(\E[U^2]) - \Psi(1) \\
				&= \Psi\big(1 + \Var[U]\big) - \Psi(1),
			\end{align*}
			where the third step is by Jensen's inequality, and the remaining steps follow from definitions. This proves the first inequality in \eqref{eq:f_entropy_UB}. Now, since $\Psi$ is concave, we have
			\begin{align*}
				\Psi\big(1+ \Var[U]\big) \le \Psi(1) + \Psi'(1) \Var[U].
			\end{align*}
	Using this, we obtain the second inequality.
		\end{proof}

		\begin{lemma}\label{lm:f_entropy_LB} Suppose $\Phi \in \cF$ is twice differentiable, and $\Phi''$ is nonincreasing. Then for any nonnegative random variable $U$ with $\E U = 1$ and $\| U \|_\infty < \infty$,
			\begin{align}\label{eq:f_entropy_LB}
				\Ent_\Phi[U] \ge \frac{\Phi''(\| U \|_\infty)}{2} \Var[U].
			\end{align}
		\end{lemma}
		\begin{proof} By Taylor's theorem, for any $u \ge 0$ we have
			\begin{align*}
				\Phi(u) - \Phi(1) = \Phi'(1) (u-1) + \frac{\Phi''(v)}{2} (u-1)^2
			\end{align*}
			for some $v \in [u \wedge 1, u \vee 1]$. Since $\Phi''$ is nonincreasing, $\Phi''(v) \ge \Phi''(u \vee 1) \ge \Phi''(\| U \vee 1 \|_\infty) = \Phi''(\| U \|_\infty)$, where the equality is a consequence of the assumption that $\E U = 1$. Taking expectations w.r.t.\ $U$, we obtain \eqref{eq:f_entropy_LB}.
		\end{proof}
	
	\begin{lemma}\label{lm:var_cond_ent} Let $U$ and $Z$ be two jointly distributed random variables, where $U$ is real-valued and nonnegative, and $Z$ takes values in an arbitrary set $\sZ$. Then, for any $\Phi \in \cF$, the expectation of the conditional $\Phi$-entropy $\Ent_\Phi[U|Z]$ admits the following variational representation:
		\begin{align*}
			\E\left[\Ent_\Phi[U|Z]\right] 
			& = \inf_{\xi \in \PFunc(\sZ)} \E\left[ \Phi(U) - \Phi(\xi(Z)) - (U-\xi(Z))\Phi'(\xi(Z))\right],
		\end{align*}
	where $\Phi'$ denotes the right derivative of $\Phi$ (which exists due to convexity).
	\end{lemma}
	
	\begin{proof} This lemma is a generalization of Lemma~14.4 in \cite{Boucheron_etal_concentration_book}. Fix an arbitrary $\xi \in \PFunc(\sZ)$. Then, by convexity of $\Phi$, for any $z \in \sZ$ we have
		\begin{align*}
			\Phi(\E[U|Z=z]) \ge \Phi(\xi(z)) + \Phi'(\xi(z)) (\E[U|Z=z]-\xi(z)).
		\end{align*}
		From this, we get
		\begin{align*}
			\Ent_\Phi\big[U\big|Z=z\big] &= \E[\Phi(U)|Z=z] - \Phi(\E[U|Z=z]) \\
			&\le \E[\Phi(U)|Z=z] - \Phi(\xi(z)) - \Phi'(\xi(z))(\E[U|Z=z] - \xi(z)).
		\end{align*}
		Taking expectations of both sides w.r.t.\ $Z$, we see that
		\begin{align}\label{eq:var_cond_ent_bound}
			\E\left[\Ent_\Phi[U|Z]\right] &\le \E\left[ \Phi(U) - \Phi(\xi(Z)) - (U - \xi(Z)) \Phi'(\xi(Z))\right]
		\end{align}
		for any $\xi \in \PFunc(\sZ)$. On the other hand, if we take $\xi(z) = \E[U|Z=z]$, then the bound in \eqref{eq:var_cond_ent_bound} is achieved with equality.
	\end{proof}

	\begin{lemma}\label{lm:entropy_derivative} Let $\Phi \in \cF$ be a differentiable function, such that $\Phi'(u)$ is uniformly bounded in some neighborhood of $u=1$. Then for any nonnegative real-valued random variable $U$ with $\E U = 1$ and $\| U \|_\infty < \infty$, we have
		\begin{align}
			\frac{\d}{\d \eps} \Ent_\Phi \left[ \frac{1-\eps U}{\bar{\eps}}\right] \Bigg|_{\eps = 0} = 0.
		\end{align}
	\end{lemma}
	
	\begin{proof} Since $\E U = 1$, for all sufficiently small $\eps > 0$ we have
		\begin{align*}
			\Ent_\Phi\left[ \frac{1-\eps U}{\bar{\eps}}\right] &= \E\left[ \Phi\left(\frac{1-\eps U}{\bar{\eps}}\right)\right].
		\end{align*}
		By our assumptions on $\Phi$, there exists a constant $C > 0$, such that
		\begin{align*}
			\left| \frac{\d}{\d \eps} \Phi \left(\frac{1-\eps u}{\bar{\eps}}\right)\right| &= \left| \frac{1- u}{\bar{\eps}^2} \Phi'\left(\frac{1-\eps u}{\bar{\eps}}\right)\right| \le C|u-1|
		\end{align*}
		for all sufficiently small $\eps > 0$. Therefore, by the dominated convergence theorem, we can interchange expectation and derivative to get
		\begin{align*}
				\frac{\d}{\d \eps} \Ent_\Phi \left[ \frac{1-\eps U}{\bar{\eps}}\right] \Bigg|_{\eps = 0} &= \Phi'(1)\E\left[  (1-U)\right] = 0.
		\end{align*}
	\end{proof}

	\section{Proof of Proposition~\ref{prop:E}}
	\label{app:E_proof}
	
	Items 1)--3) are obvious. We prove 4). To that end, we first analyze the joint distribution of $X$ and $X'$. First of all, for any $x,x' \in \sX$, using the definition of $K^*$, we can write
	\begin{align*}
		P_{XX'}(x,x') &= \mu(x)K^* K(x'|x) \\
		&= \mu(x)\sum_{y \in \sY} K^*(x'|y)K(y|x) \\
		&= \mu(x)\sum_{y \in \sY} \frac{K(y|x')\mu(x')}{\mu K(y)} K(y|x) \\
		&= \mu(x')  \sum_{y \in \sY} \frac{K(y|x)\mu(x)}{\mu K(y)} K(y|x') \\
		&= \mu(x') \sum_{y \in \sY} K^*(x|y) K(y|x) \\
		&= \mu(x') K^* K(x|x') \\
		&= P_{XX'}(x',x).
	\end{align*}
In other words, the distribution of $P_{XX'}$ is \textit{exchangeable} (or $(X,X')$ is an \textit{exchangeable pair}). This implies, in particular, that the marginal distribution $P_{X'}$ is the same as $P_X$, i.e., $\mu$. Moreover, for any function $f \in \Func(\sX)$ and any $x \in \sX$,
\begin{align*}
	\E[f(X')|X=x] &= \sum_{x' \in \sX}K^* K(x'|x)f(x') \\
	&= \sum_{x' \in \sX}\sum_{y \in \sY}K^*(x'|y)K(y|x)f(x') \\
	&= \sum_{y \in \sY} K(y|x) \sum_{x' \in \sX} K^*(x'|y)f(x') \\
	&= \sum_{y \in \sY} K(y|x) K^*f(y) \\
	&= KK^*f(x).
\end{align*}
Using these facts, we can write
	\begin{align*}
	&	\E\left[ \left(f(X)-f(X')\right)\left(g(X)-g(X')\right)\right]  \\
	&\quad = \E\left[f(X)g(X)\right]  + \E[f(X')g(X')]  - \Big(  \E[f(X)g(X')] - \E[f(X')g(X)]\Big) \\
	&\quad = 2 \Big\{\E[f(X)g(X)] - \E[f(X)g(X')]\Big\},
	\end{align*}
where
	\begin{align*}
		\E[f(X)g(X')] &= \E[f(X)\E[g(X')|X]] \\
		&= \E[f(X) KK^* g(X)] \\
		&= \E[K^* f(X) K^* g(X)] \\
		&= \E[\E[f(X)|Y]\E[g(X)|Y]] \\
		&= \E[f(X)\E[g(X)|Y]].
	\end{align*}
Accordingly, we have
\begin{align*}
	\E\left[ \left(f(X)-f(X')\right)\left(g(X)-g(X')\right)\right] 
& = 2\Big\{  \E\left[f(X)\left(g(X)-\E[g(X)|Y]\right)\right]\Big\} \\
& = 2\,\cE(f(X),g(X)|Y),
\end{align*}
where the second step follows from the identity $\cE(U,V|Y) = \E[U(V-\E[V|Y])]$.
This proves \eqref{eq:exch_1}. To prove \eqref{eq:exch_2}, write
\begin{align*}
	\E\left[ \left(f(X)-f(X')\right)\left(g(X)-g(X')\right)\right] &= \E\left[ 1_{\{f(X) > f(X')\}}\left(f(X)-f(X')\right)\left(g(X)-g(X')\right)\right] \\
& \qquad + \E\left[ 1_{\{f(X) < f(X')\}}\left(f(X)-f(X')\right)\left(g(X)-g(X')\right)\right] \\
&= 2\, \E\left[ 1_{\{f(X) > f(X')\}}\left(f(X)-f(X')\right)\left(g(X)-g(X')\right)\right] \\
&= 2\, \E\left[\left(f(X)-f(X')\right)_+\left(g(X)-g(X')\right) \right],
\end{align*}
where the second step is by exchangeability of $X$ and $X'$.
	
\end{appendix}


\end{document}